\documentclass[preprint,11pt]{elsarticle}
\usepackage [utf8] {inputenc}
\usepackage[T1]{fontenc}  
\usepackage{xcolor}
\usepackage{prooftree}
\usepackage{mathpartir}
\usepackage{float}
\usepackage{hyperref}
\floatstyle{boxed}
\restylefloat{figure}   
\usepackage{amsmath}
\usepackage{bbm}
\usepackage{stmaryrd}
\usepackage{xspace}
\usepackage{graphicx}
\usepackage{mathtools,amssymb}
\usepackage[curve,matrix,arrow,cmtip]{xy}
\usepackage{wasysym}
\usepackage{txfonts}
\usepackage{pxfonts}
\usepackage{url}
\usepackage{relsize}
\usepackage{tikz}
\usepackage{enumitem}
\usepackage{chngcntr}
\usepackage{apptools}
\usepackage[english]{babel}

\newtheorem{theorem}{Theorem}[section]
\newtheorem{definition}[theorem]{Definition}
\newtheorem{proposition}[theorem]{Proposition}
\newtheorem{lemma}[theorem]{Lemma}
\newtheorem{example}[theorem]{Example}

\newenvironment{proof}[1][Proof]{\begin{trivlist}
\item[\hskip \labelsep {\bfseries #1}]}{\end{trivlist}}

\newcommand{\cau}[2]{#1\circ#2}

\newcommand{\inp}[5]{\SigmaB_{#2\in#3}#1?\Seq{#4_#2}{#5_#2}}
\newcommand{\inpp}[5]{\SigmaB_{#2\in#3}#1?\Seq{#4_#2}{#5}}
\newcommand{\oup}[5]{\bigoplus_{#2\in#3}#1!\Seq{#4_#2}{#5_#2}}
\newcommand{\oupp}[5]{\bigoplus_{#2\in#3}#1!\Seq{#4_#2}{#5}}
\newcommand{\oupTx}[5]{\textstyle{\bigoplus}_{#2\in#3}#1!\Seq{#4_#2}{#5_#2}}

\newcommand{\gt}[6]{#1\to#2: \GlSyB_{#3\in#4}\Seq{#5_#3}{#6_#3}}
\newcommand{\gtp}[6]{#1\to#2: \GlSyB_{#3\in#4}\Seq{#5_#3}{#6'_#3}}

 \newcommand{\RG}{\ensuremath{{\sf G}}}
 \newcommand{\GlSyB}{\mathlarger{\mathlarger{\mathlarger{\boxplus}}}}
 \newcommand{\PiB}{\mathlarger{\mathlarger\Pi}}
  \newcommand{\SigmaB}{\mathlarger{\mathlarger\Sigma}}

    \newcommand{\bm}{\vartheta}
     \newcommand{\bms}{\Theta}
   \newcommand{\ES}{\mathcal{P\!E}}

    \newcommand{\EGG}{\mathcal{G\!E}}

  \newcommand{\ee}{\epsilon}

  \renewcommand{\sep}{~|~}
  \newcommand{\eh}[1]{|#1|}

 \newcommand{\mylabel}[1]{\label{#1}}

\newcommand{\comocc}{\ensuremath{\gamma}}

\newcommand{\event}{\ensuremath{\eta}}
\newcommand{\pastpref}[2]{\ensuremath{\concat{{#1}}{#2}}}

\newcommand{\act}[1]{\ensuremath{\sf act}(#1)}

\newcommand{\concat}[2]{\ensuremath{#1\,{\cdot}\,#2}}

 \newcommand{\inpP}[5]{\SigmaB_{#2\in#3}#1?\Seq{#4_#2}{#5}}
\newcommand{\oupP}[5]{\bigoplus_{#2\in#3}#1!\Seq{#4_#2}{#5}}

\newcommand{\refToFigure}[1]{Figure~\ref{#1}}
\newcommand{\refToSection}[1]{Section~\ref{#1}}
\newcommand{\refToExample}[1]{Example~\ref{#1}}
\newcommand{\refToLemma}[1]{Lemma~\ref{#1}}
\newcommand{\refToTheorem}[1]{Theorem~\ref{#1}}

\newcommand{\refToDef}[1]{Definition~\ref{#1}}

\newcommand{\refToProp}[1]{Proposition~\ref{#1}}

\newcommand{\set}[1]{\{#1\}}
\newcommand{\kf}[1]{\ensuremath{\mathsf{#1}\xspace}}
 
\newcommand{\cSinferrule}[3][]{
  \mprset{fraction={===},
  fractionaboveskip=0.2ex,
  fractionbelowskip=1.30ex}
  \inferrule[#1]{#2}{~#3}
}

\newcommand{\PP}{\ensuremath{P}}
\newcommand{\la}{\lambda}

\newcommand{\M}{\lambda}
\newcommand{\Messages}{\ensuremath{{\sf Msg}}}

\newcommand{\pp}{{\sf p}}
\newcommand{\q}{{\sf q}}
\newcommand{\pr}{{\sf r}}
\newcommand{\ps}{{\sf s}}
\newcommand{\pt}{{\sf t}}
\newcommand{\Participants}{\ensuremath{{\sf Part}}}
\newcommand{\sendL}[2]{#1!#2}
\newcommand{\rcvL}[2]{#1?#2}

\newcommand{\pc}{~|~}
\newcommand{\Seq}[2]{#1;#2}

\newcommand{\inact}{\ensuremath{\mathbf{0}}}

\newcommand{\val}{v}
\newcommand{\Q}{\ensuremath{Q}}
\newcommand{\R}{\ensuremath{R}}

\newcommand{\del}[1]{\ensuremath{\delta}}
\newcommand{\Nt}{\ensuremath{{\sf N}}}

\newcommand{\parN}{\mathrel{\|}}
\newcommand{\pP}[2] {#1[\![\,#2\,]\!]}

\newcommand{\stackred}[1]{\xrightarrow{#1}}

 \newcommand{\G}{\ensuremath{{\sf G}}}

 \newcommand{\ty}{\textrm{\bf t}}
 \newcommand{\End}{\kf{End}}

\newcommand{\partcomm}[1]{\ensuremath{{\sf part}(#1)}}
\newcommand{\ptone}[1]{\ensuremath{{\sf pt}(#1)}}

\newcommand{\GlSy}{\boxplus}

 \newcommand{\participant}[1]{\partcomm{#1}}

\newcommand{\proj}[2]{#1 \!  \upharpoonright  \! #2\,}

\newcommand{\projb}[2]{#1 \!  \Rsh  \! #2\,}

\newcommand{\subt}{\leq}

\newcommand{\derN}[2]{\vdash #1 :#2}

\newcommand{\rulename}[1]{{[\textsc{#1}]}}

\newcommand{\weight}{\ensuremath{{\sf depth}}}
\newcommand{\gtCom}[3]{#1\stackrel{#3}{\to}#2}

\newcommand{\gr}{\ensuremath{~\#~}}

\newcommand{\grr}{\ensuremath{\,\#\,}}

\newcommand{\impl}{\ensuremath{\Rightarrow}}
\newcommand{\procev}{\eta}
\newcommand{\Procev}{\ES}

\newcommand{\netev}{\nu}

\newcommand{\globev}{\gamma}

\newcommand{\precP}{\ensuremath{\leq}}
\newcommand{\precN}{\ensuremath{\prec}}

 \newcommand{\GE}{\mathcal{N\!E}}
  \newcommand{\GEs}{\mathcal{G\!E}}

\newcommand{\DE}{\mathcal{D\!E}}

\newcommand{\dualev}[2]{\ensuremath{#1\Join #2}}
\newcommand{\dualevS}[2]{\ensuremath{#1\;\,\widehat\Join\;\, #2}}

\newcommand{\locev}[2]{\ensuremath{#1:: #2}}

\newcommand{\comseq}{ \sigma } 
\newcommand{\actseq}{\zeta} 
\newcommand{\Comseq}{\textit{Traces}}

\newcommand{\eqclass}[1]{\ensuremath{[#1]_{\sim}}}

\newcommand{\quotient}{\ensuremath{\Comseq/\!\!\sim}}
\newcommand{\gn}[1]{\ensuremath{{\sf gn}}(#1)}

\newcommand{\ESP}[1]{\ensuremath{\mathcal{S^P}(#1)}}

\newcommand{\ESN}[1]{\ensuremath{\mathcal{S}^{\mathcal{N}}(#1)}}

\newcommand{\ESet}{\ensuremath{\mathcal{X}}}

\newcommand{\ESNstar}[1]{\ensuremath{\mathcal{S}_{*}^{\mathcal{N}}(#1)}}
\newcommand{\ESG}[1]{\ensuremath{\mathcal{S}^{\mathcal{G}}(#1)}}

\newcommand{\emptyseq}{\ensuremath{\epsilon}}
\newcommand{\eqdef}{\ensuremath{{=_{\sf def}}}}
\newcommand{\length}[1]{\ensuremath{\sf length(#1)}}

\newcommand{\Conf}[1]{\ensuremath{\mathcal{C}(#1)}}
\newcommand{\CD}[1]{\ensuremath{\mathcal{D}(#1)}}

\newcommand{\comm}[1]{\ensuremath{{\sf cm}}(#1)}

\newcommand{\last}[1]{\ensuremath{{\sf last}}(#1)}

\newcommand{\nec}[1]{\ensuremath{{\sf nec}(#1)}}
\newcommand{\gec}[1]{\ensuremath{{\sf gec}(#1)}}

\newcommand{\postP}[2]{\ensuremath{{#2}\,\blacklozenge\,{(#1)}}}
\newcommand{\preP}[2]{\ensuremath{{#2}\,\lozenge\,{(#1)}}}
\newcommand{\post}[2]{\ensuremath{{#2}\,\blacklozenge\,{#1}}}
\newcommand{\pre}[2]{\ensuremath{{#2}\,\lozenge\,{#1}}}

\newcommand{\postG}[2]{\ensuremath{{#2}\bullet{#1}}}

\newcommand{\preG}[2]{\ensuremath{{#2}\circ{#1}}}

\newcommand{\coDefGr}{::=^{coind}}

\newcommand{\GP}{{\G}}

\newcommand{\FPaths}[1]{{\sf Tr^+}(#1)}

\newcommand{\nr}[1]{{\sf n}(#1)}

\newcommand{\projnet}[2]{{\sc proj}_{#1}(#2)}
\newcommand{\projnetfun}[1]{{\sc proj}_{#1}}

\newcommand{\Comm}[3]{#1#3#2}

\newcommand{\loc}[1]{\ensuremath{{\sf loc}(#1)}}

\newcommand{\pro}[2]{#1\,@\,#2}

\newcommand{\ev}[1]{\ensuremath{{\sf ev}(#1)}}

 \newcommand{\projS}[2]{#1 @   #2\,}

\newcommand{\at}[2]{#1[#2]}

\newcommand{\range}[3]{#1[#2\,...\,#3]}

\newcommand{\cardin}[1]{\!\!\pc\! #1\!\pc\!}

\newcommand{\occ}[2]{\ensuremath{#1\!\in\!\in\!#2}}
\newcommand{\tr}[2]{\ensuremath{#1\! \searrow \!#2}}

\newcommand{\prR}[3]{(#1,#2)\in\downarrow_#3}
\newcommand{\RR}{{\mathcal R}}
\newcommand{\iR}[3]{(#1,#2)\in #3}

\newcommand{\RT}{{\mathcal RT}}

\newcommand{\causpre}[3]{\ensuremath{{\sf pre}}(#1,#2,#3)}

\newcommand{\RS}{{\mathcal RS}}

\newcommand{\causpost}[3]{\ensuremath{{\sf post}}(#1,#2,#3)}

\newcommand{\EvSet}{\ensuremath{Ev}}

\newcommand{\qedhere}{}

\newenvironment{lemmaa}[2]{\begin{trivlist}
\item[\hskip \labelsep {\bfseries Lemma #1}] {#2} \it}{\end{trivlist}}

\begin{document}

\begin{frontmatter}

\title{Event Structure Semantics for Multiparty Sessions}
\author[1]{Ilaria Castellani\corref{cor1}\fnref{ic}
}
\ead{ilaria.castellani@inria.fr}
\fntext[ic]{This research has been supported by the ANR17-CE25-0014-01 CISC project.}
\address[1]{INRIA, Universit\'e C\^ote d'Azur, France}
\author[2]{Mariangiola Dezani-Ciancaglini}
 \ead{dezani@di.unito.it}
\address[2]{Dipartimento di Informatica, Universit\`a di Torino, Italy} 
\author[3]{Paola Giannini\fnref{pg}}
\ead{paola.giannini@uniupo.it}
\address[3]{DiSSTE,Universit\`{a} del Piemonte Orientale, Italy} 
\fntext[pg]{This original research has the financial support of the Universit\`{a}  del Piemonte Orientale.}
\cortext[cor1]{Corresponding author. INRIA, 2004 Route des Lucioles, 
BP 93, 06902 Sophia Antipolis
FRANCE}

\date{}

\begin{abstract} 
   We propose an interpretation of multiparty sessions as \emph{Flow
     Event Structures}, which allows concurrency 
  within sessions 
to be explicitly represented.  We show that this
   interpretation is equivalent, when the multiparty sessions can be
  described by global types, to an interpretation of such global types as
   \emph{Prime Event Structures}. 
\end{abstract}
\end{frontmatter}




\section{Introduction}
\label{sec:intro}

Session types were proposed in the mid-nineties~\cite{THK94,HVK98}, as
a tool for specifying and analysing web services and communication
protocols. They were first introduced in a variant of the
$\pi$-calculus to describe binary interactions between processes.
Such binary interactions may often be viewed as client-server 
protocols. 
Subsequently, session types were extended to {\em multiparty
  sessions}~\cite{CHY08,CHY16}, where several participants may
interact with each other.  A multiparty session is an interaction
among peers, and there is no need to distinguish one of the
participants as representing the server.  All one needs is an abstract
specification of the protocol that guides the interaction.  This is
called the \emph{global type} of the session. The global type
describes the behaviour of the whole session, as opposed to the local
types that describe the behaviours of single participants. In a
multiparty session, local types may be retrieved as projections from
the global type.

Typical safety properties ensured by session types are
\emph{communication safety} (absence of communication errors),
\emph{session fidelity} (agreement with the protocol) and
\emph{deadlock-freedom}~\cite{CHY16}.  When dealing with multiparty
sessions, the type system is often enhanced so as to guarantee also
the liveness property known as \emph{progress} (no participant gets
stuck)~\cite{H2016}.
\\
Some simple examples of sessions not satisfying the above properties
are: 1) a sender emitting a message while the receiver expects a
different message (communication error); 2) two participants both
waiting to receive a message from the other one (deadlock due to a
protocol violation); 3) a three-party session where the first
participant waits to receive a message from the second participant,
which keeps interacting forever with the third participant
(starvation, although the session is not deadlocked).

What makes session types particularly attractive is that they offer
several advantages at once: 1) static safety guarantees, 2) automatic
check of protocol implementation correctness, based on local types,
and 3) a strong connection with linear
logics~\cite{CP10,TCP11,Wadler14,PCPT14,CPT16}, and with concurrency
models such as communicating automata~\cite{DY12}, graphical
choreographies~\cite{LTY15,TuostoG18} and message-sequence
charts~\cite{CHY16}.

In this paper we further investigate the relationship between
multiparty session types and concurrency models, by focussing on Event
Structures~\cite{Win88}.  We consider a standard multiparty session
calculus where sessions are described as networks of sequential
processes~\cite{DY12}.  Each process implements a participant in the
session.  We propose an interpretation of such networks as \emph{Flow
  Event Structures} (FESs)~\cite{BC88a,BC94} (a subclass of Winskel's
Stable Event Structures~\cite{Win88}), which allows concurrency
between session communications to be explicitly represented. We then
introduce global types for these networks, and define an
interpretation of them as \emph{Prime Event Structures}
(PESs)~\cite{Win80,NPW81}. Since the syntax of global types does not
allow all the concurrency among communications to be expressed, the
events of the associated PES need to be defined as equivalence classes
of communication sequences up to \emph{permutation equivalence}. We
show that when a network is typable by a global type, the FES
semantics of the former is equivalent, in a precise technical sense,
to the PES semantics of the latter.  In a companion
paper~\cite{CDG21}, we investigated a similar Event Structure
semantics for a session calculus with asynchronous communication,
which led to a quite different treatment as it made use of a new
notion of asynchronous global type. A detailed comparison
with~\cite{CDG21} will be given in \refToSection{sec:related}.

 This paper is an expanded and amended version
of~\cite{CDG-LNCS19}. The main novelty is that we use a coinductive
definition for processes and global types, which simplifies several
definitions and proofs, and a more stringent definition for network
events. This definition relies on the new notion of causal set, 
which is crucial for the correctness of our ES semantics.
Finally, the present paper includes all proofs
of results, some of which require ingenuity. 

The paper is organised as follows. Section~\ref{sec:calculus}
introduces our multiparty session calculus.  In
Section~\ref{sec:eventStr} we recall the definitions of PESs and FESs,
which will be used to interpret processes
(Section~\ref{sec:process-ES}) and networks
(Section~\ref{sec:netS-ES}), respectively.  PESs are also used to
interpret global types (Section~\ref{sec:events}), which are defined
in Section~\ref{sec:types}. In Section~\ref{sec:results} we prove the
equivalence between the FES semantics of a network and the PES
semantics of its global type.  Section~\ref{sec:related} discusses
related work 
and sketches directions for future work.  The Appendix contains some technical proofs.


\section{A Core Calculus for Multiparty Sessions}\mylabel{sec:calculus}

We now formally introduce our calculus, where multiparty sessions are
represented as networks of processes.  We assume the following base
sets: \emph{session participants}, ranged over by $\pp,\q,\pr, \ldots$
and forming the set $\Participants $, and \emph{messages}, ranged over
by $\la,\la',\dots$ and forming the set $\Messages$.

Let $\pi \in \{ \sendL{\pp}{\la}, \rcvL{\pp}{\la} \pc
\pp\in \Participants, \la \in \Messages\}$ denote an \emph{action}.
The action $\sendL{\pp}{\la}$ represents an output of message $\la$ to
participant $\pp$, while the action $\rcvL{\pp}{\la}$ represents an
input of message $\la$ from participant $\pp$.   The
\emph{participant of an action}, $\ptone{\pi}$, is defined by
$\ptone{\sendL{\pp}{\la}}=\ptone{\rcvL{\pp}{\la}}=\pp$.

\begin{definition}[Processes]\mylabel{p} 
 Processes are  defined by:\\
 \[
\begin{array}{lll}
\PP & \coDefGr  & 
\oup\pp{i}{I}{\la}{\PP}%
~~\mid~~  
\inp\pp{i}{I}{\la}{\PP}%
~~\mid~~
\inact
\end{array}
\]
\noindent
where $I$ is non-empty and $\la_h\not=\la_k$ for all $h,k\in I$,
$h\neq k$, i.e.
 messages  in choices are all different.\\
Processes of the shape $\oup\pp{i}{I}{\la}{\PP}$ and
$\inp\pp{i}{I}{\la}{\PP}$ are called {\em output} and {\em input
  processes}, respectively.
\end{definition}\noindent
The symbol $ \coDefGr$, in the definition above and in later
definitions, indicates that the productions should be interpreted
\emph{coinductively}.  
Namely, they define possibly infinite processes.  However, we assume
such processes to be \emph{regular}, that is, with finitely many
distinct subprocesses. In this way, we only obtain processes which are
solutions of finite sets of equations, see~\cite{Cour83}. So, when
writing processes, we shall use (mutually) recursive equations.

Sequential composition ($;$) has higher precedence than choices ($\bigoplus$, $\SigmaB$). 
When $I$ is a singleton, $\oup\pp{i}{I}{\la}{\PP}$ will be rendered as
$\Seq{\sendL{\pp}{\la}}{\PP}$ and $\inp\pp{i}{I}{\la}{\PP}$ will be
rendered as $\Seq{\rcvL{\pp}{\la}}{\PP}$.  Trailing $\inact$ processes
will be omitted.

In a full-fledged calculus, messages would carry values, 
namely they would be of the form $\la(\val)$.  For simplicity, we
consider only pure messages here. This will allow us to project global
types directly to processes,  without having to explicitly introduce local
types,  see Section \ref{sec:types}.

\bigskip

Networks are comprised of 
pairs of the form
$\pP{\pp}{\PP}$ composed in parallel, each with a
different participant $\pp$. 
\begin{definition}[Networks]
{\em Networks} are defined by:
\[
\Nt = \pP{\pp_1}{\PP_1} \parN \cdots \parN \pP{\pp_n}{\PP_n}  \qquad
 n\geq 1,  \ 
\pp_h \neq \pp_k ~~\text{for any}~~ h, k~(1\leq h,k\leq n).
\]
\end{definition}
We assume the standard structural congruence  $\equiv$  on
networks, stating that parallel composition is associative and
commutative and has neutral element $\pP\pp\inact$ for any fresh
$\pp$.

If $\PP\neq\inact$ we write $\pP{\pp}{\PP}\in\Nt$ as short for
$\Nt\equiv\pP{\pp}{\PP}\parN\Nt'$ for some $\Nt'$.

To express the operational semantics of networks, we use an LTS whose
labels record the message exchanged during a communication together
with its sender and receiver.  The set of \emph{communications},
ranged over by $\alpha, \alpha'$, is defined to be $\{
\Comm{\pp}{\M}{\q} \pc \pp,\q\in \Participants, \M \in \Messages\}$,
where $\Comm{\pp}{\M}{\q}$ represents the emission of a message $\M$
from participant $\pp$ to participant
$\q$.  

\begin{figure}[h]
{\small
\[
\begin{array}{c}
\pP{\pp}{\oup\q{i}{I}{\la}{\PP}}\parN \pP{\q}{\inp\pp{j}{J}{\la}{\Q}}\parN\Nt\stackred{\Comm\pp{\la_k}\q}
  \pP{\pp}{\PP_k}\parN \pP{\q}{\Q_k}\parN\Nt~~~\text{where }
   k \in I{\cap}J{~~~~~~\rulename{Com}}
\end{array}
\]
}
\caption{
LTS for networks.}\mylabel{fig:netred}\mylabel{fig:procLTS}
\end{figure}

The LTS semantics of networks is specified by the unique rule
$\rulename{Com}$ given in \refToFigure{fig:procLTS}.  
Notice that rule
$\rulename{Com}$  is symmetric with respect to 
input and output
choices.  In a well-typed network (see \refToSection{sec:types}) it
will always be the case that $I\subseteq J$,  ensuring  that participant
$\pp$ can freely choose an output, since participant $\q$ offers all
corresponding inputs. 

\bigskip

In the following we will make an extensive use of finite (and possibly
empty) sequences of communications. As usual we define them as traces.
\begin{definition}[Traces]\label{traces} {\em (Finite) traces} $\comseq \in
  \Comseq$ are defined by:
 \[
  \comseq::=\ee\mid\concat\alpha\comseq
 \]
 We use
  $\cardin{\comseq}$ to denote the length of the trace $\comseq$.\\ 
  The set of {\em participants of a trace}, notation $\participant{\comseq}$, 
  is defined by 
  $\participant{\ee}=\emptyset$ and
  $\participant{\concat{\Comm\pp\la\q}\comseq}=\set{\pp,\q}\cup\participant\comseq$.
\end{definition}
\noindent 
When
$\comseq=\concat{\alpha_1}{\concat\ldots{\alpha_n}}$ ($n\geq 1)$ we
write $\Nt\stackred{\sigma}\Nt'$ as short for
$\Nt\stackred{\alpha_1}\Nt_1\cdots\stackred{\alpha_n}\Nt_n = \Nt'$.


\section{Event Structures}\label{sec:eventStr}

We recall now the definitions of \emph{Prime Event Structure} (PES)
from~\cite{Win80,NPW81} and \emph{Flow Event Structure} (FES)
from~\cite{BC88a}. The class of FESs is more general than that of
PESs: for a precise comparison of various classes of event structures,
we refer the reader to~\cite{BC91}. As we shall see in
Sections~\ref{sec:process-ES} and~\ref{sec:netS-ES}, while PESs are sufficient to interpret
processes, the  greater  generality of FESs is needed to interpret networks.
\begin{definition}[Prime Event Structure] \mylabel{pes} A prime event structure {\rm (PES)} is a
    tuple 
$S=(E,\leq, \gr)$ where:
\begin{enumerate}
\item \mylabel{pes1} $E$ is a denumerable set of events;
\item \mylabel{pes2}    $\leq\,\subseteq (E\times E)$ is a partial order relation,
called the \emph{causality} relation;
\item \mylabel{pes3}  $\gr\subseteq (E\times E)$ is an irreflexive symmetric relation, called the
\emph{conflict} relation, satisfying the property: $\forall e, e', e''\in E: e \grr e'\leq
e''\Rightarrow e \grr e''$ (\emph{conflict hereditariness}).
\end{enumerate}
\end{definition}

\begin{definition}[Flow Event Structure]\mylabel{fes} A flow event structure {\rm (FES)}
    is a  tuple $S=(E,\prec,\gr)$ where:
    \begin{enumerate}
\item\mylabel{fes1} $E$ is a denumerable set of events;
\item\mylabel{fes2} $\prec\,\subseteq (E\times E)$ is an irreflexive relation,
called the \emph{flow} relation;
\item\mylabel{fes3} $\gr\subseteq (E\times E)$ is a symmetric relation, called the
\emph{conflict} relation.
\end{enumerate}
\end{definition}
Note that the flow relation is not required to be transitive, nor
acyclic (its reflexive and transitive closure is just a preorder, not
necessarily a partial order).  Intuitively, the flow relation
represents a possible {\sl direct causality} between two events.
 Moreover,  in a FES the conflict relation is not required to be
irreflexive nor hereditary; indeed, FESs may exhibit self-conflicting
events, as well as disjunctive causality (an event may have
conflicting causes).

Any PES $S = (E , \leq, \gr)$ may be regarded as a FES, with $\prec$
given by $<$ (the strict ordering) 
or by the covering relation of $\leq$.

\bigskip

We now recall the definition of {\sl configuration}\/ for event
structures. Intuitively, a configuration is a set of events having
occurred at some stage of the computation.   Thus, the semantics
of an event structure $S$ is given by its poset of configurations ordered
by set inclusion, where $\ESet_1 \subset \ESet_2$ means that $S$ may evolve
from $\ESet_1$ to $\ESet_2$. 
\begin{definition}[PES configuration] \mylabel{configP}
  Let $S=(E,\leq, \gr)$ be a prime event structure. A
    configuration of $S$ is a finite subset $\ESet$ of $E$ such that:
\begin{enumerate}
    \item \mylabel{configP1} $\ESet$ is downward-closed: \ $e'\leq e \in \ESet \, \ \impl \ \ e'\in \ESet$; 
 \item \mylabel{configP2} $\ESet$ is conflict-free: $\forall e, e' \in \ESet, \neg (e \gr
e')$.
\end{enumerate}
\end{definition}
The definition of
configuration for FESs is slightly more elaborated.
%
%
For a subset $\ESet$ of
$E$, let $\prec_\ESet$ be the restriction of the flow relation to $\ESet$ and
$ \prec_\ESet^*$ be its transitive and reflexive closure.
%
\begin{definition}[FES configuration]
\mylabel{configF}
Let $S=(E,\prec, \gr)$ be a flow event
structure. A configuration of $S$ is a finite
subset $\ESet$ of $E$ such that:
\begin{enumerate}
\item\mylabel{configF1} $\ESet$ is downward-closed up to conflicts: 
\ $e'\prec e \in \ESet, \ e'\notin \ESet\, \  \impl \ 
\,\exists \, e''\in \ESet.\,\, e'\gr \,e''\prec e$;
\item\mylabel{configF2} $\ESet$ is conflict-free: $\forall e, e' \in \ESet, \neg (e \gr
e')$;
\item\mylabel{configF3} $\ESet$ has no causality cycles: the relation $ \prec_\ESet^*$  is a partial order.
\end{enumerate}
\end{definition}
Condition (\ref{configF2}) is the same as for prime event
structures. Condition (\ref{configF1}) is adapted to account for the more general
-- non-hereditary -- conflict relation. 
It states that any event appears in a configuration with a ``complete
set of causes''.  Condition (\ref{configF3}) ensures that any event in a
configuration is actually reachable at some stage of the computation.

\bigskip

If $S$ is a prime or flow event structure, we denote by $\Conf{S}$ its
set of 
 configurations.   Then, the \emph{domain of
  configurations} of $S$ is defined as follows:
\begin{definition}[ES configuration domain]
\mylabel{configDom}
Let $S$ be a prime or flow event
structure with set of configurations $\Conf{S}
 $. The \emph{domain of configurations} of $S$ is the partially ordered set
\mbox{$\CD{S} \eqdef (\Conf{S}, \subseteq)$.}
\end{definition}
We recall from~\cite{BC91} a useful characterisation for
configurations of FESs, which is based on the notion of proving
sequence, defined as follows:


\begin{definition}[Proving sequence]
  \mylabel{provseq} Given a flow event structure $S=(E,\prec, \gr)$, a
  \emph{proving sequence} in $S$ is a sequence $\Seq{e_1;
    \cdots}{e_n}$ of distinct non-conflicting events (i.e. $i\not=j\
  \impl\ e_i\not=e_j$ and $\neg(e_i\gr e_j)$ for all $i,j$) satisfying:
  \[
  \forall i\leq n\,\forall e \in E \,: \quad e\prec e_i\
    \ \impl\ \ \exists j<i\,. \ \ \text{ either } \ e = e_j\ \text{ or
    } \ e\grr e_j \prec e_i 
    \]
\end{definition}

Note that any prefix of a proving sequence is itself a proving
sequence. 

\bigskip

We have the following characterisation of configurations of FESs  
in terms of proving sequences.
\begin{proposition}[Representation of FES configurations as proving
  sequences~\cite{BC91}]
  \mylabel{provseqchar} Given a flow event structure $S=(E,\prec,
  \gr)$, a subset $\ESet$ of $E$ is a configuration of $S$ if and only
  if it can be enumerated as a proving sequence $\Seq{e_1;
    \cdots}{e_n}$.
\end{proposition}
Since PESs may be viewed as particular FESs, we may use
\refToDef{provseq} and \refToProp{provseqchar} both for the FESs
associated with networks (see Sections~\ref{sec:netS-ES})
and for the PESs associated with global types
(see \refToSection{sec:events}).  Note that for a PES
the condition of \refToDef{provseq} simplifies to
\[
\forall i\leq n\,\forall e \in E \,: \quad e < e_i\ \
  \impl\ \ \exists j<i\,. \ \ e = e_j 
\]

\medskip

To conclude this section, we recall from~\cite{CZ97} the definition of
\emph{downward surjectivity} (or \emph{downward-onto}, as it was
called there), a property that is required for partial functions
between two FESs in order to ensure that they preserve configurations.
We will make use of this property in \refToSection{sec:netS-ES}.

\begin{definition}[Downward surjectivity]
  \mylabel{down-onto} Let $S_i=(E_i,\prec_i, \gr_i)$, be a flow event
  structure, $i=0,1$. 
Let $e_i,e'_i$ range over $E_i$, $i=0,1$.  
A partial function $f: E_0
  \rightarrow_* E_1$ is \emph{downward surjective} if it satisfies the
  condition:
\[e_1 \prec_1 f(e_0) \implies \exists e'_0 \in E_0~.~e_1 = f(e'_0) \]
 \end{definition}


\section{Event Structure Semantics of Processes}\mylabel{sec:process-ES}

In this section, we define an event structure semantics for processes,
and show that the obtained event structures are PESs.  This semantics
will be the basis for defining the ES semantics for networks
in~\refToSection{sec:netS-ES}.  We start by introducing process
events, which are non-empty sequences of actions.

\begin{definition}[Process event]
  \mylabel{proceventP} 
{\em Process events} $\procev,
  \procev'$,  also called \emph{p-events},  are defined by: 
\[\procev \ \quad ::= \pi ~~ \mid~~
    \pastpref{\pi}{\procev} \qquad\qquad \pi \in \!
    \set{\sendL{\pp}{\la}, \rcvL{\pp}{\la} 
      \mid \pp\in \Participants, \la \in \Messages}\]
We denote by $\Procev$ the set of  p-events, and by
$\cardin{\procev}$ the length of the sequence of actions in
the p-event $\procev$. 
\end{definition}

Let $\actseq$ denote a (possibly empty) sequence of actions, and
$\sqsubseteq$ denote the prefix ordering on such sequences.  Each
 p-event  $\procev$ may be written either in the form $\procev = \concat
{\pi}{\actseq}$ or in the form $\procev =
\concat{\actseq}{\pi}$. We shall feel free to use any of these forms. 
When a  p-event  is written as $\procev = \concat{\actseq}{\pi} $,
then $\actseq$ may be viewed as the \emph{causal history} of
$\procev$, namely the sequence of past actions that must have happened
in the process for $\procev$ to be able to happen. 

We define the  \emph{action}  of a  p-event   to be
its last action: 
\[
\act{\concat{\actseq}{\pi}} = \pi
\]

\begin{definition}[Causality and conflict relations on process events] 
\mylabel{procevent-relations}
The \emph{causality} relation $\leq$ and the \emph{conflict} relation
$\gr$ on the set of  p-events  $\Procev$ are defined by:
\begin{enumerate}
\item \mylabel{ila--esp2} 
$\procev \sqsubseteq \procev' \ \impl \ \procev\precP\procev'$;\smallskip
\item \mylabel{ila--esp3} 
  $\pi\neq\pi'  \impl
  \concat{\concat{\actseq}{\pi}}{\actseq'}\,\grr\,\,\concat{\concat{\actseq}{\pi'}}{\actseq''}$.
\end{enumerate}
\end{definition}
%


\begin{definition} [Event structure of a process] \mylabel{esp}
The {\em event structure of process} $\PP$ is the triple
\[
 \ESP{\PP} = (\ES(\PP), \precP_\PP , \gr_\PP)
 \]
 where:
\begin{enumerate}
\item \mylabel{ila-esp1} 
   $\ES(\PP) \subseteq \Procev $ is the set of non-empty sequences
  of labels along the nodes and edges of a path from the root to an
  edge in the tree of $\PP$; 
\item \mylabel{ila-esp2} 
$\precP_\PP$ is the restriction of
  $\precP$ to the set 
$\ES(\PP)$;
\item \mylabel{ila-esp3} 
$\gr_\PP$ is the restriction of
  $\grr$ to the set 
$\ES(\PP)$.
\end{enumerate}
\end{definition}

It is easy to see that $\gr_\PP = (\ES(\PP)\times \ES(\PP))
\,\backslash \, (\precP_\PP \cup \geq_\PP)$.
In the following we shall feel free to drop the subscript in
$\precP_\PP$ and $\gr_\PP$.
\bigskip

Note that the set $\ES(\PP)$ may be denumerable, as shown by the following example.  
\begin{example}\mylabel{ex:rec1}
If  $\PP=\sendL{\q}\la;\PP \oplus\sendL{\q}{\la'}$, 
then 
$ \ES(\PP)  = \begin{array}[t]{l}
\set{\underbrace{\sendL{\q}\la\cdot\ldots\cdot\sendL{\q}\la}_n \mid
  n\geq 1} \quad \cup \\
\set{\underbrace{\sendL{\q}\la\cdot\ldots\cdot\sendL{\q}\la}_n\cdot\sendL{\q}{\la'}\mid
  n\geq 0} 
\end{array} $
\end{example}

\begin{proposition}\mylabel{basta10}
  Let $\PP$ be a process. Then $\ESP{\PP}$ is a prime event structure.
\end{proposition}
\begin{proof} We show that $\precP$ and $\gr$ satisfy Properties
  \ref{pes2} and \ref{pes3} of \refToDef{pes}.  Reflexivity,
  transitivity and antisymmetry of $\precP$ follow from the
  corresponding properties of $\sqsubseteq$. As for irreflexivity and
  symmetry of $\gr$, they follow from Clause \ref{ila--esp3} of
  Definition \ref{procevent-relations} and the corresponding
  properties of inequality.  To show conflict hereditariness, suppose
  that $\procev \grr \procev'\precP \procev''$.  From Clause
  \ref{ila--esp3} of Definition \ref{procevent-relations} there are
  $\pi$, $\pi'$, $\actseq$, $\actseq'$ and $\actseq$ such that
  $\pi\neq\pi'$ and $\procev=\concat{\concat{\actseq}{\pi}}{\actseq'}$
  and $\procev'=\concat{\concat{\actseq}{\pi'}}{\actseq''}$.  From
  $\procev'\precP \procev''$ we derive that
  $\procev''=\concat{\concat{\actseq}{\pi'}}{\concat{\actseq''}{\actseq_1}}$
  for some $\actseq_1$.  Therefore $\procev \grr \procev''$, again
  from Clause \ref{ila--esp3}.
 \end{proof}

%

\section{Event Structure Semantics of Networks}\mylabel{sec:netS-ES}

In this section we define the ES semantics of networks and show that
the resulting ESs, which we call \emph{network ESs}, are FESs. We also
show that when the network is binary, namely when it has only two
participants, then the obtained FES is a PES.  The formal treatment
involves defining the set of potential events of network ESs, which we
call \emph{network events}, as well as introducing the notion of
\emph{causal set} of a network event and the notion of
\emph{narrowing} of a set of network events.  This will be the subject
of \refToSection{subsec:main-defs-props}.

In \refToSection{subsec:further-props}, we first prove some properties
of the conflict relation in network ESs.  Then, we come back to
causal sets and we show that they are always finite and that each
configuration includes a unique causal set for each of its
n-events. We also discuss the relationship between causal sets and
prime configurations, which are specific configurations that are in
1-1 correspondence with events in ESs. Finally, we define a notion of
projection from n-events to p-events, and prove that this projection
(extended to sets of n-events) is downward surjective  and
preserves configurations.

\subsection{Definitions and Main Properties}
\mylabel{subsec:main-defs-props}

We start by defining network events, the potential events
of network ESs.  Since these events represent communications between
two network participants $\pp$ and $\q$, they should be pairs of \emph
{dual  p-events}, namely, of  p-events  emanating
respectively from $\pp$ and $\q$, which have both dual actions and
dual causal histories.

Formally, to define network events we need to specify the
\emph{location} of  p-events, namely the participant to which they
belong: 
\begin{definition}[Located event]
\mylabel{proceventL} We call \emph{located event} a  p-event 
$\procev$ pertaining to a participant $\pp$, written $\pp::\procev$.
 \end{definition}
 As hinted above, network events should be pairs of dual located
 events $\locev{\pp}{\actseq\cdot\pi}$ and
 $\locev{\q}{\actseq'\cdot\pi'}$ with matching actions $\pi$ and
 $\pi'$ and matching histories $\actseq$ and $\actseq'$. To formalise
 the matching condition, we first define the projections of process
 events on participants, which yield sequences of \emph{undirected
   actions} of the form $!\la$ and $?\la$, or the empty sequence
 $\emptyseq$.  Then we introduce a notion of duality between located
 events, based on a notion of duality between undirected actions.

Let $\bm$ range over $!\la$  and $?\la$, and $\bms$ range over 
 (possibly empty)  sequences of $\bm$'s. 

\begin{definition}[Projection of  p-events]\label{pd}
The projection of a  p-event  $\procev$ on a participant
$\pp$, written $\projb{\procev}{\pp}$, is defined by: \\
\[
\projb{\sendL{\q}{\la}}\pp=\begin{cases}
\sendL{}{\la}     & \text{if }\pp=\q \\
   \emptyseq   & \text{otherwise}
\end{cases}\quad
\projb{\rcvL{\q}{\la}}\pp=\begin{cases}
\rcvL{}{\la}     & \text{if }\pp=\q \\
   \emptyseq   & \text{otherwise}
\end{cases}
\]
\[
\projb{(\concat{\pi}{\procev})}{\pp}=
\concat{\projb{\pi}\pp} {\projb\procev\pp} 
\]
\end{definition}

\begin{definition}[Duality of undirected action sequences]\label{dualProcEv}
The \emph{duality of undirected action sequences}, written $\dualev\bms{\bms'}$, is the
symmetric relation induced by: 
\[ \dualev{\ee}{\ee}
\qquad\quad
\dualev\bms{\bms'} ~\impl ~ \dualev{\,\concat{\sendL{}{\la}} \bms\,}
 { \,\concat{\rcvL{}{\la}} \bms'}
 \]
\end{definition}

\begin{definition}[Duality of located events]\label{dualLocEv}
Two located events $\locev{\pp}{\procev}, \locev{\q}{\procev'}$ are \emph{dual}, written  $\dualevS{
  \locev{\pp}{\procev}}{\locev{\q}{\procev'}}$,  if $\dualev{\projb\procev\q}{\projb{\procev'}\pp}$
  and 
$\ptone{\act{\procev}} = \q\,$ and
$\,\ptone{\act{\procev'}} = \pp$. 
\end{definition}

Dual located events may be sequences of actions of different length. For instance 
$\dualevS{\locev{\pp}{\concat{\sendL{\q}{\la}}{\sendL{\pr}{\la'}}}}{\locev{\pr}{\rcvL{\pp}{\la'}}}$ and
$\dualevS{\locev{\pp}{{\sendL{\q}{\la}}}}{\locev{\q}{\concat{\sendL{\pr}{\la'}}{\rcvL{\pp}{\la}}}}$.

\begin{definition}[Network event]
\mylabel{n-event} 
{\em Network events} $\netev, \netev'$, also called
\emph{n-events}, are unordered
pairs of dual %
located events, namely:
\[
\netev ::= \set{\locev{\pp}{\procev},
    \locev{\q}{\procev'}} \qquad
  \text{where}~~~\dualevS{\locev{\pp}{\procev}}{\locev{\q}{\procev'}}
\]  
We denote by $\DE$ the set of n-events.
 \end{definition}

 We define {\em the communication of the event} $\netev$, notation
 $\comm\netev$, by $\comm\netev=\Comm\pp{\M}\q$ if
 $\netev=\set{\locev{\pp}{\concat\actseq{\sendL\q\M}},
   \locev{\q}{\concat{\actseq'}{\rcvL\pp\M}}}$ and we say that the
 n-event $\netev$ \emph{represents} the communication
 $\Comm\pp{\M}\q$.  We also define the set of \emph{locations} of an
 n-event to be $\loc{\set{\locev{\pp}{\procev}, \locev{\q}{\procev'}}}
 = \set{\pp,\q}$.

 It is handy to have a notion of occurrence of a located event in
a set of network events: 
 
 \begin{definition}\label{olesne}
   {\em A located event $\locev{\pp}{\procev}$ occurs in a set $E$ of
     n-events,} notation $\occ{\locev{\pp}{\procev}}E$, if
   $\locev{\pp}{\procev}\in\netev$ and $\netev\in E$ for some
   $\netev$.
 \end{definition} 

We define now the flow and conflict relations on network
events. While the flow relation is the expected one (a network event
inherits the causality from its constituent processes), the conflict
relation is more subtle, as it can arise also between network
events with disjoint  sets of locations. 

 In the following definition we use $\eh\bms$ to denote the length of
 the sequence $\bms$. 
\begin{definition}[Flow and conflict relations on n-events] 
\mylabel{netevent-relations}
The \emph{flow} relation $\precN$ and the \emph{conflict} relation
$\gr$ on the set of n-events $\DE$ are defined by:

\begin{enumerate}
\item\mylabel{c1} 
$\netev \precN\netev '$ if 
$
\locev{\pp}{\procev}\in \netev 
~\&~\locev{\pp}{\procev'}\in\netev' ~\&~ \procev  <  \procev' $; 
\item\mylabel{c2} 
  $\netev \grr\netev '$ if 
 \begin{enumerate}
\item \label{c21} either
  $  \locev{\pp}{\procev}\in \netev 
~\&~\locev{\pp}{\procev'}\in\netev'~\&~\procev \grr \procev' $;
\item\label{c22}
 or $\locev{\pp}{\procev}\in \netev 
~\&~\locev{\q}{\procev'}\in\netev' ~\&~\pp \neq \q ~\&~
\cardin{\projb\procev\q} = \cardin{\projb{\procev'}\pp}   ~\&~
\neg(\dualev{\projb\procev\q}{\projb{\procev'}\pp})$. 
\end{enumerate}
\end{enumerate}
\end{definition}
Two n-events are in conflict if they share a participant
with conflicting p-events (Clause (\ref{c21}))
or if some of their participants have communicated with each other in the past
in incompatible ways (Clause (\ref{c22})).  Note that the two
clauses are not exclusive, as shown in the following example.

\begin{example}\mylabel{ex:conflict} This example illustrates the use
  of 
  \refToDef{netevent-relations} in various cases. It also shows
that the flow and conflict relations may be overlapping on n-events.
\begin{enumerate}
\item \mylabel{ex:conflict1} Let
  $\netev=\set{\locev\pp{\q!\la_1\cdot\pr!\la},\locev\pr{\pp?\la}}$
  and $\netev'=\set{\locev\pp{\q!\la_2}, \locev\q{\pp?\la_2}}$.  Then
  $\netev \grr\netev '$ by Clause (\ref{c21}) 
since
  $\q!\la_1\cdot\pr!\la\grr\q!\la_2$. Note that $\netev \grr\netev '$
  can be also deduced by Clause (\ref{c22}), since
  $\projb{(\q!\la_1\cdot\pr!\la)}{\q} = \,!\la_1$ and
  $\projb{\pp?\la_2}\pp = \,?\la_2$ and $\cardin{!\la_1} =
  \cardin{?\la_2}$ and $\neg(\dualev{!\la_1}{?\la_2})$.
\item \mylabel{ex:conflict2} 
Let $\netev$ be as in (\ref{ex:conflict1})
and $\netev'=\set{\locev\pp{\q!\la_2\cdot \q!\la},
  \locev\q{\pp?\la_2\cdot \pp?\la}}$. Again, we can deduce
  $\netev \grr\netev '$ using Clause (\ref{c21}) 
  since $\q!\la_1\cdot\pr!\la\grr\q!\la_2\cdot \q!\la$.  On the other
  hand, Clause (\ref{c22}) does not apply in this case since
  $\projb{(\q!\la_1\cdot\pr!\la)}{\q} = \,!\la_1$ and
  $\projb{(\pp?\la_2\cdot\pp?\la)}\pp = \,?\la_2\cdot?\la$ and thus
  $\cardin{!\la_1} \neq \cardin{?\la_2\cdot?\la}$.
\item \mylabel{ex:conflict3} Let $\netev$ be as in
  (\ref{ex:conflict1}) and
  $\netev'=\set{\locev\q{\pp?\la_2\cdot \ps!\la},
    \locev\ps{\q?\la}}$. Here $\loc{\netev} \cap \loc{\netev'}
  =\emptyset$, so clearly Clause (\ref{c21}) does not apply.  On the
  other hand, $\netev \grr\netev '$ can be deduced by Clause
  (\ref{c22}) since $\projb{(\q!\la_1\cdot\pr!\la)}{\q} = \,!\la_1$
  and $\projb{(\pp?\la_2\cdot\ps!\la)}\pp = \,?\la_2$ and
  $\cardin{!\la_1} = \cardin{?\la_2} $ and
  $\neg(\dualev{!\la_1}{?\la_2})$.  

\item \mylabel{ex:conflict5} Let $\netev$ be as in
  (\ref{ex:conflict1}) and $\netev'=\set{\locev\pp{\q!\la_2
\cdot\pr!\la \cdot\pr!\la'}, \locev\pr{\pp?\la \cdot\pp?\la'}}$.
In this case we have both $\netev\prec\netev'$ by Clause (\ref{c1})
and $\netev\grr\netev'$ by Clause (\ref{c21}), namely, causality
is inherited from participant $\pr$ and conflict from participant $\pp$.
\end{enumerate}
\end{example}

We introduce now the notion of \emph{causal set} of an n-event
$\netev$ in a given set of events $\EvSet$. Intuitively, a causal set
of $\netev$ in $\EvSet$ is a \emph{complete set of non-conflicting
direct causes} of $\netev$ which is included in $\EvSet$.

\begin{definition}[Causal set]
\mylabel{cs}
Let $\netev\in \EvSet \subseteq \DE$. 
A set of n-events $E$ is a {\em causal set} of $\netev$ in
$\EvSet$ if $E$ is a minimal subset of $\EvSet$ such that
\begin{enumerate}
\item \mylabel{cs1} $E \cup \set{\netev}$ is conflict-free and  
\item \mylabel{cs2} $\locev{\pp}{\procev}\in\netev$ and $\procev' <
  \procev$ 
imply $\occ{\locev{\pp}{\procev'}}{E}$.
\end{enumerate}
\end{definition} 

Note that in the above definition, the conjunction of minimality and
Clause (\ref{cs2}) implies that, if $\netev'\in E$, then
$\netev'\prec\netev$. Thus $E$ is a set of direct causes of $\netev$.
Moreover, a causal set of an n-event cannot be included in another
causal set of the same n-event, as this would contradict the
minimality of the larger set. Hence, \refToDef{cs} indeed formalises
the idea that causal sets should be complete sets of compatible direct
causes of a given n-event.

\begin{example}\mylabel{ex:causal-sets}
  Let $\netev_1=\set{\locev\pp{\q!\la_1\cdot
      \pr!\la},\locev\pr{\pp?\la}}$ and
  $\netev_2=\set{\locev\pp{\q!\la_2\cdot
      \pr!\la},\locev\pr{\pp?\la}}$.  Then both $\set{\netev_1}$ and
  $\set{\netev_2}$ are causal sets of $\netev=\set{\locev\pr{\pp?\la
      \cdot\ps!  \la'},\locev\ps{\pr?\la'}}$ in $\EvSet =
  \set{\netev_1, \netev_2, \netev}$. Note that $\netev_1 \grr
  \netev_2$ and that neither $\netev_1$ nor $\netev_2$ has a causal
  set in $\EvSet$.

Let us now consider also 
$\netev'_1=\set{\locev\pp{\q!\la_1},\locev\q{\pp?\la_1}}$ and
$\netev'_2=\set{\locev\pp{\q!\la_2},\locev\q{\pp?\la_2}}$.  Then
$\netev$ still has the same causal sets $\set{\netev_1}$ and
$\set{\netev_2}$ in $\EvSet' = \set{\netev'_1, \netev'_2, \netev_1,
  \netev_2, \netev}$, while each $\netev_i$, $i=1,2$, has the unique
causal set $\set{\netev'_i}$ in $\EvSet'$, and each $\netev'_i$,
$i=1,2$, has the empty causal set in $\EvSet'$.

Finally, $\netev$ has infinitely many causal sets in $\DE$.  For
instance, if for every natural number $n$ we let
$\netev_n=\set{\locev\pp{\q!\la_n\cdot \pr!\la},\locev\pr{\pp?\la}}$,
then each $\set{\netev_n}$ is a causal set of $\netev$ in $\DE$.
Symmetrically, a causal set may cause infinitely many events in
$\DE$. For instance, the above causal sets $\set{\netev_1}$ and
$\set{\netev_2}$ of $\netev$ could also act as causal sets for any
n-event $\netev''_n=\set{\locev\pr{\pp?\la \cdot\ps!
    \la_n},\locev\ps{\pr?\la_n}}$ or, assuming the set of participants
to be denumerable, for any event $
\netev'''_n=\set{\locev\pr{\pp?\la \cdot \ps_n!
    \la'},\locev{\ps_n}{\pr?\la'}}$.  
\end{example}

When defining the set of events of a network ES, we want to prune out
all the n-events that do not have
a causal set in the set itself.  The reason is that such
n-events cannot happen.  This pruning is achieved by means of the
following narrowing function.

\begin{definition}[Narrowing of a set of n-events]
\mylabel{def:narrowing}
The \emph{narrowing} of a set $E$ of n-events, denoted by $\nr{E}$,
is the greatest fixpoint of the function $f_E$ on sets of 
n-events defined by:
\[
\begin{array}{lll}
f_E(X) &=& \set{\netev\in E \mid \exists E' \subseteq X . \, E' \text{is a causal
    set of } \netev \text{ in $X$ }  }
 \end{array}
\] 
\end{definition}

Note that we could not have taken $\nr{E}$ to be the least
fixpoint of $f_E$ rather than its greatest fixpoint.
Indeed, the least fixpoint of $f_E$ would be the empty set.

\begin{example} 
\label{ex-narrowing}
The following two examples illustrate the notions of causal set
and narrowing. 

Let $\netev_1=\set{\locev\pr{\ps?\la_1}, \locev\ps{\pr!\la_1}}$,
$\netev_2=\set{\locev\pr{\ps?\la_2}, \locev\ps{\pr!\la_2}}$,
$\netev_3=\set{\locev\pp{\pr?\la_1},\locev\pr{\ps?\la_1\cdot\pp!\la_1}}$,
$\netev_4=\set{\locev\q{\ps?\la_2},\locev\ps{\pr!\la_2\cdot\q!\la_2}}$,
$\netev_5=\set{\locev\pp{\pr?\la_1\cdot\q!\la},
  \locev\q{\ps?\la_2\cdot\pp?\la}}$.  Then 
$\nr{\set{\netev_1,\ldots,\netev_5}}= \set{\netev_1,\ldots,\netev_4}$,
 because a causal set for $\netev_5$ would need to contain both
$\netev_3$ and $\netev_4$, but this is not possible since
$\netev_3 \grr \netev_4$   
by Clause (\ref{c22}) of
\refToDef{netevent-relations}. In fact $\projb{(\ps?\la_1 \cdot\pp!\la_1)}\ps = \,?\la_1$ and
  $\projb{(\pr!\la_2  \cdot\q!\la_2)}\pr = \,!\la_2$ and $\cardin{?\la_1} =
  \cardin{!\la_2}$ and $\neg(\dualev{?\la_1}{!\la_2})$.

  Let $\netev_1=\set{\locev\pr{\ps?\la_1}, \locev\ps{\pr!\la_1}}$,
  $\netev_2=\set{\locev\pr{\ps?\la_2}, \locev\ps{\pr!\la_2}}$,
  $\netev_3=\set{\locev\pp{\pr?\la_1},\locev\pr{\ps?\la_1\cdot\pp!\la_1}}$,
  $\netev_4=\set{\locev\pp{\pr?\la_1\cdot\ps?\la_2},\locev\ps{\pr!\la_2\cdot\pp!\la_2}}$,
  $\netev_5=\set{\locev\pp{\pr?\la_1\cdot\ps?\la_2\cdot\q!\la},\locev\q{\pp?\la}}$.
  Here $\nr{\set{\netev_1,\ldots,\netev_5}}=
  \set{\netev_1,\netev_2,\netev_3}$. Indeed, a causal set for
  $\netev_4$ would need to contain both $\netev_2$ and $\netev_3$, but
  this is not possible since $\netev_2 \grr \netev_3$  by Clause
  (\ref{c21}) of \refToDef{netevent-relations}.  In fact
  $\ps?\la_2\grr\ps?\la_1  \cdot\pp!\la_1$.  Then, 
  $\netev_5$   
  will  also
  be pruned by the narrowing since any causal set for
  $\netev_5$ should contain $\netev_4$.
\end{example}

We can now  finally  define the event structure associated with a  network:

\begin{definition}[Event structure of a
   network] \mylabel{netev-relations} The {\em  event
     structure of network} $\Nt $
   is the triple
\[
   \ESN{\Nt} = (\GE(\Nt), \precN_\Nt , \grr_\Nt)
\]
where:
\begin{enumerate}
\item\mylabel{netev-relations1} 
$\begin{array}[t]{ll}\GE(\Nt) =\nr{\DE(\Nt)} ~~\text{ with }\\
\DE(\Nt) = \set{\set{\locev{\pp}{\procev},\locev{\q}{\procev'}} \sep\pP{\pp}{\PP}{\in}\Nt,\pP{\q}{\Q}{\in}\Nt,
\procev{\in}\ES(\PP),  \procev' {\in}\ES(\Q),
  \dualevS{\locev{\pp}{\procev}}{\locev{\q}{\procev'}}}
\end{array}$ 
\smallskip
\item \mylabel{c1SE} 
$\precN _\Nt$ is the restriction of
  $\precN$ to the set  $\GE(\Nt)$;
\item \mylabel{c2SE} 
$\gr _\Nt$ is the restriction of
  $\gr$ to the set  $\GE(\Nt)$.  \end{enumerate}
  \end{definition}

 The set of n-events of a network ES can be infinite, as  shown
 by  the following example. 
\begin{example}\mylabel{ex:rec2}
Let  $\PP$ be as in \refToExample{ex:rec1}, $\Q=\rcvL{\pp}\la;\Q +\rcvL{\pp}{\la'}$ 
and $\Nt=\pP{\pp}{\PP} \parN \pP{\q}{\Q}$. Then
$$ \GE(\Nt) =
\begin{array}[t]{l}
 \set{\set{\locev\pp{\underbrace{\sendL{\q}\la\cdot\ldots\cdot\sendL{\q}\la}_n},
\locev\q{\underbrace{\rcvL{\pp}\la\cdot\ldots\cdot\rcvL{\pp}\la}_n}}\mid
n\geq 1} \quad \cup  \\[15pt]
\set{\set{\locev\pp{\underbrace{\sendL{\q}\la\cdot\ldots\cdot\sendL{\q}\la}_n\cdot\sendL{\q}{\la'}},
\locev\q{\underbrace{\rcvL{\pp}\la\cdot\ldots\cdot\rcvL{\pp}\la}_n\cdot\rcvL{\pp}{\la'}}}\mid
n\geq 0} 
\end{array}
$$
 A simple variation of this example shows that even within the
events of a network ES, an n-event $\netev$ may have an infinite
number of causal sets. Let $\netev=\set{\locev\pr{\pp?\la \cdot\ps!
   \la'},\locev\ps{\pr?\la'}}$ be as in \refToExample{ex:causal-sets}.  
%
%
Consider the network $\Nt'=\pP{\pp}{\PP'} \parN \pP{\q}{\Q} \parN
\pP{\pr}{\R} \parN \pP{\ps}{S}$, where
$\PP'=\sendL{\q}\la;\PP'
\oplus\sendL{\q}{\la'}; \sendL{\pr}\la$,
$\Q$ 
is as above, $\R = \rcvL{\pp}\la;\sendL{\ps}{\la'}$ and
$S = \rcvL{\pr}\la'$.

Then $\netev$ has an infinite number of causal sets $E_n =
\set{\netev_n}$ in $\GE(\Nt')$, where
\[
\netev_n =
\set{\locev\pp{\underbrace{\sendL{\q}\la\cdot\ldots\cdot\sendL{\q}\la}_n
\cdot\,\sendL{\q}{\la'}\cdot\sendL{\pr}{\la},
\locev{\pr}{\rcvL{\pp}\la}}}
\]
On the other hand, a causal set may only cause a finite number of
events in a network ES, since the number of branches in any choice is
finite, as well as the number of participants in the network.

\end{example}

\begin{theorem}\mylabel{nf}
Let $\Nt$ be a  network. Then $\ESN{\Nt}$ is a flow event
structure with an irreflexive conflict relation. 
\end{theorem}
\begin{proof}
  The relation  $\precN _\Nt$  is irreflexive since $ \procev < \procev'$
  implies $\netev\not=\netev'$, where $\procev,\procev ',\netev,
  \netev'$ are as in  \refToDef{netevent-relations}(\ref{c1}).
  As for the conflict relation, note first that a conflict between an
  n-event and itself could not be derived by Clause (\ref{c22}) of
  \refToDef{netevent-relations}, since the two located events of
  an n-event are dual by construction. Then,  symmetry and
  irreflexivity of the conflict relation follow from the corresponding
  properties of conflict between  p-events. 
\end{proof}

Notably, n-events with disjoint sets of 
 locations  may 
be
related by the transitive closure of the flow relation, as illustrated
by the following example, which also shows how  n-events 
inherit the flow relation from the causality relation of their 
p-events.

\begin{example}\mylabel{ex1}

  Let $\Nt$ be the network 
\[
 \pP{\pp}{\sendL{\q}{\la_1}}\parN
    \pP{\q}{\Seq{\rcvL{\pp}{\la_1}}{\sendL{\pr}{\la_2}}} \parN
    \pP{\pr}{\Seq{\rcvL{\q}{\la_2}}{\sendL{\ps}{\la_3}}}\parN
    \pP{\ps}{\rcvL{\pr}{\la_3}}
\]
Then $\ESN{\Nt}$ has three network  events
\[
\begin{array}{c}
\netev_1 {=}
    \set{\locev{\pp}{\sendL{\q}{\la_1}},\locev{\q}{\rcvL{\pp}{\la_1}}}\quad
    \netev_2 {=}
    \set{\locev{\q}{\concat{{\rcvL{\pp}{\la_1}}}{\sendL{\pr}{\la_2}}},\locev{\pr}{\rcvL{\q}{\la_2}}}
\quad
  \netev_3 {=} 
    \set{\locev{\pr}{\concat{\rcvL{\q}{\la_2}}{\sendL{\ps}{\la_3}}},
      \locev{\ps}{\rcvL{\pr}{\la_3}}}
\end{array}
  \]    
The flow relation obtained by
  \refToDef{netev-relations} is: $\netev_1 \precN \netev_2$ and
  $\netev_2 \precN \netev_3$.  Note that each time the flow relation
  is inherited from the causality within a different participant, $\q$
  in the first case and $\pr$ in the second case.  
  The nonempty configurations are $\set{\netev_1}, \set{\netev_1,
    \netev_2}$ and $\set{\netev_1, \netev_2, \netev_3}$.
Note that $\ESN{\Nt}$ has only one proving sequence per configuration
(which is  the one  given by the numbering of events).
\end{example}

If  a network is  binary, 
then its FES may be turned into a PES by replacing $\precN$ with
$\precN^*$.  To prove this result, we first show  a
property of n-events of binary networks. We say that an n-event
$\netev$ is {\em binary} if  the participants occurring in 
the p-events of $\netev$ are contained in $\loc\netev$.

\begin{lemma}
  \mylabel{bnc} Let $\netev$ and $\netev'$ be binary n-events with
  $\loc\netev=\loc{\netev'}$. Then $\netev\grr\netev'$ iff
  $\locev\pp\procev\in\netev$ and $\locev\pp{\procev'}\in\netev'$
  imply $\procev\grr\procev'$.
  \end{lemma} 
\begin{proof}
 The ``if'' direction holds  by
\refToDef{netevent-relations}(\ref{c21}).  We show the
``only-if'' direction.  First observe that  for  any
n-event $\netev=\set{\locev\pp{\procev_1},\locev\q{\procev_2}}$ the
condition $\dualevS{\locev\pp{\procev_1}}
{\locev\q{\procev_2}}$ of
\refToDef{n-event} implies
$\projb{\procev_1}{\q}\Join\projb{\procev_2}{\pp}$ by
\refToDef{dualLocEv}, which in turn implies
$\cardin{\projb{\procev_1}{\q}}=\cardin{\projb{\procev_2}{\pp}}$ by
\refToDef{dualProcEv}.   If $\netev$ is a binary event, 
we also have $\cardin{\procev_1}=\cardin{\projb{\procev_1}{\q}}$ and
$\cardin{\procev_2}=\cardin{\projb{\procev_2}{\pp}}$ by \refToDef{pd}, 
since  all the actions of $\procev_1$ involve $\q$ and all the
actions of $\procev_2$ involve $\pp$, and thus 
the projections do not erase actions. \\
Assume  now 
$\netev'=\set{\locev\pp{\procev_1'},\locev\q{\procev_2'}}$.
We consider two cases  (the others being symmetric):  
\begin{itemize}[label=--]
\item $\netev\grr\netev'$  because 
  $\procev_1\grr\procev_1'$.
   Then 
$\projb{\procev_1}{\q}\Join\projb{\procev_2}{\pp}$ and
  $\projb{\procev_1'}{\q}\Join\projb{\procev_2'}{\pp}$
  imply $\procev_2\grr\procev_2'$; 
 \item $\netev\grr\netev'$  because 
  $\cardin{\projb{\procev_1}{\q}}=\cardin{\projb{\procev'_2}{\pp}}$
  and $\neg(\projb{\procev_1}{\q}\Join\projb{\procev'_2}{\pp})$. 
   As argued before, 
  we have $\cardin{\projb{\procev_2}{\pp}}
  =\cardin{\projb{\procev_1}{\q}}$ and
  $\cardin{\projb{\procev'_2}{\pp}} =\cardin{\projb{\procev'_1}{\q}}$.
  Then, from
  $\cardin{\projb{\procev_1}{\q}}=\cardin{\projb{\procev'_2}{\pp}}$
  and the above remark about binary events, we get
  $\cardin{\procev_2}=\cardin{\procev_1} =\cardin{\procev'_2} =
  \cardin{\procev'_1}$. From
  $\neg(\projb{\procev_1}{\q}\Join\projb{\procev'_2}{\pp})$ it follows
  that $\procev_1 \neq \procev'_1$ and $\procev_2
  \neq \procev'_2$. Then we may conclude, since $\cardin{\procev_i} =
  \cardin{\procev'_i}$ and $\procev_i \neq \procev'_i$ imply
  $\procev_i \grr \procev'_i$ for $i=1,2$.  \qedhere 
  \end{itemize} 
\end{proof} 

\begin{theorem}
\mylabel{binary-network-FES}
Let $\Nt = \pP{\pp_1}{\PP_1} \parN \pP{\pp_2}{\PP_2}\,$ and $\,\ESN{\Nt}
=  (\GE(\Nt), \precN _\Nt, \grr)$. Then  $\nr{\DE(\Nt)} =
  \DE(\Nt)$ and   the structure $\ESNstar{\Nt} \eqdef (\GE(\Nt), \precN^* _\Nt , \grr)$
is a prime event structure.
\end{theorem}
\begin{proof} 
  We first show that $\nr{\DE(\Nt)} = \DE(\Nt)$. By
  \refToDef{netev-relations}(\ref{netev-relations1}) 
\[
  \DE(\Nt) =
    \set{\set{\locev{\pp_1}{\procev_1},\locev{\pp_2}{\procev_2}} \sep
      \procev_1 \in \ES(\PP_1), \procev_2 \in \ES(\PP_2),
      \dualevS{\locev{\pp_1}{\procev_1}}{\locev{\pp_2}{\procev_2}}}
\]
Let $\set{\locev{\pp_1}{\procev_1},\locev{\pp_2}{\procev_2}} \in
  \DE(\Nt)$. Since
  $\dualevS{\locev{\pp_1}{\procev_1}}{\locev{\pp_2}{\procev_2}}$ and
  all the actions in $\procev_1$ involve $\pp_2$ and all the actions
  in $\procev_2$ involve $\pp_1$, we know that $\procev_1$ and
  $\procev_2$ have the same length $n \geq 1$ and for each $i, 1\leq i
  \leq n$, the prefixes of length $i$ of $\procev_1$ and $\procev_2$,
  written $\procev^i_1$ and $\procev^i_2$, must themselves be
  dual. Then
  $\set{\locev{\pp_1}{\procev^i_1},\locev{\pp_2}{\procev^i_2}} \in
  \DE(\Nt)$ for each $i, 1\leq i \leq n$, hence
  $\set{\locev{\pp_1}{\procev_1},\locev{\pp_2}{\procev_2}}$ has a
  causal
  set in $\DE(\Nt)$.\\
  We prove now that the reflexive and transitive closure $\precN^*
  _\Nt$ of $\precN _\Nt$ is a partial order.  Since by definition $\precN^*
  _\Nt$ is a preorder, we only need to show that it is antisymmetric.
  Define the length of an n-event $\netev =
  \set{\locev{\pp_1}{\procev_1},\locev{\pp_2}{\procev_2}}$ to be
  $\length{\netev} \eqdef \, \cardin{\procev_1} + \cardin{\procev_2}$
  (where $\cardin{\procev}$ is the length of $\procev$, as
 given by 
\refToDef{proceventP}). 
 Let now $\netev, \netev'\in \GE(\Nt)$, with  $\netev =
\set{\locev{\pp_1}{\procev_1},\locev{\pp_2}{\procev_2}}$ and $\netev'
= \set{\locev{\pp_1}{\procev'_1},\locev{\pp_2}{\procev'_2}}$.  By
definition $\netev \precN _\Nt \netev'$ implies $\procev_i < \procev'_i$
for some $i=1,2$, which in turn implies $\cardin{\procev_i} <
\cardin{\procev'_i}$.  
 As observed above, 
$\procev_1$ and $\procev_2$ must have the
same length, and so must $\procev'_1$ and $\procev'_2$ . 
This means that if $\netev \precN _\Nt \netev'$ then
 $\length{\netev} = \cardin{\procev_1}+ \cardin{\procev_2} < 
\cardin{\procev'_1} + \cardin{\procev'_2} = \length{\netev'}$. 
From this we can conclude
that if $\netev \precN^* _\Nt \netev'$ and $\netev' \precN^* _\Nt \netev$, then
necessarily $ \netev = \netev'$.
\\
Finally we show that the relation $\grr$ satisfies the required
properties. By \refToTheorem{nf} we only need to prove that $\grr$ is
hereditary. Let $\netev$ and $\netev'$ be as above. If
$\netev\grr\netev'$, then by \refToLemma{bnc}
$\procev_1\grr\procev_1'$ and $\procev_2\grr\procev_2'$.  Let now
$\netev'' =
\set{\locev{\pp_1}{\procev''_1},\locev{\pp_2}{\procev''_2}}$. If
$\netev'\precN^* _\Nt\netev''$, this means that there exist $\netev_1,
\ldots, \netev_n$ such that $\netev'\precN  _\Nt\netev_1 \ldots
\precN _\Nt\netev_n=\netev''$. We prove by induction on $n$ that
$\netev\grr\netev''$. For $n=1$ we have $\netev'\precN _\Nt\netev''$.  Then
by Clause (\ref{c1}) of \refToDef{netev-relations} we have $\procev'_j
< \procev''_j$ for some $j\in\set{1,2}$.
 Since $\procev_i \grr \procev'_i$ for all $i\in\set{1,2}$ 
and $\grr$ is hereditary on p-events, we deduce
$\procev_j\grr\procev''_j$, which implies $\netev\grr\netev''$. 
Suppose now $n >1$. By induction $\netev\grr\netev_{n-1}$. Since
$\netev_{n-1} \precN _\Nt\netev_n=\netev''$ we then obtain
$\netev\grr\netev''$ by the same argument as in the base case.
\end{proof}

If  a network  
has more than two participants, then the
duality requirement on its  n-events  is not sufficient to ensure the
absence of circular dependencies\footnote{This is a well-known issue
  in multiparty session types, which motivated the introduction of
  global types in \cite{CHY08}, see \refToSection{sec:types}.}.  For
instance, in the following ternary network (which may be viewed as
representing the 3-philosopher deadlock) the relation $\precN^*$ is
not a partial order.
\begin{example}\mylabel{ex2}
Let $\Nt$ be the network 
\[
\pP{\pp}{\Seq{\rcvL{\pr}{\la}}{\sendL{\q}{\la'}}} \parN
\pP{\q}{\Seq{\rcvL{\pp}{\la'}}{\sendL{\pr}{\la''}}} \parN
\pP{\pr}{\Seq{\rcvL{\q}{\la''}}{\sendL{\pp}{\la}}}
\]
Then  $\ESN{\Nt}$  has three  n-events
\[
\begin{array}{c}
\netev_1 = \set{\locev{\pp}{\rcvL{\pr}{\la}},\locev{\pr}{\concat{{\rcvL{\q}{\la''}}}{\sendL{\pp}{\la}}}}\qquad
\netev_2 =
\set{\locev{\pp}{\concat{{\rcvL{\pr}{\la}}}{\sendL{\q}{\la'}}},
\locev{\q}{\rcvL{\pp}{\la'}}}\\
\netev_3 = \set{\locev{\q}{\concat{{\rcvL{\pp}{\la'}}}{\sendL{\pr}{\la''}}},\locev{\pr}{\rcvL{\q}{\la''}}}
\end{array}
\]
By \refToDef{netev-relations}(\ref{c1}) 
we have
$\netev_1 \precN \netev_2 \precN \netev_3$ and  $\netev_3 \precN \netev_1$.
The only configuration of  $\ESN{\Nt}$  is the empty
configuration, because the only set of n-events that satisfies
downward-closure up to conflicts
is $X =\set{\netev_1, \netev_2, \netev_3}$, 
but this is not a configuration because 
$\prec_X^*$ is not  a partial order (recall that $\prec_X$ is the
restriction of $\precN$ to $X$) and hence the condition (\ref{configF3}) of
\refToDef{configF} is not satisfied. 
\end{example}

\subsection{Further Properties}
\mylabel{subsec:further-props}

In this subsection, we first prove two properties of the
conflict relation in network ESs: non disjoint n-events are always in
conflict, and conflict induced by Clause (\ref{c22}) of
\refToDef{netevent-relations} is semantically inherited.  
 We then discuss the relationship between causal sets and prime
configurations and prove two further properties of causal sets, which
are shared with prime configurations: finiteness, and the existence 
of a causal set for each event in a configuration.
Finally,  observing that the FES of a network may be viewed as the product
of the PESs of its processes, we proceed to prove a classical property
for ES products, namely that their projections on their components
preserve configurations.  To this end, we define a projection function
from n-events to participants, yielding p-events, and we show that
configurations of a network ES project down to configurations of the 
PESs of its processes. 

\bigskip

Let us start with the conflict properties. By definition, two
n-events intersect each other if and only if they share a located
event $\locev{\pp}{\procev}$. Otherwise, the two n-events are
disjoint. 
 Note that if 
$\locev{\pp}{\procev}\in(\netev\cap\netev')$, then
$\loc{\netev} = \loc{\netev'} = \set{\pp, \q}$, where $\q =\ptone{\act{\procev}}$.
The next proposition establishes that two  distinct  intersecting n-events
in $\DE$ are in conflict. 
\begin{proposition}[Sharing of located events implies conflict]\mylabel{prop:conf}
  If $\netev, \netev' \in  \DE $ and $\netev\neq \netev'$ and
  $(\netev\cap\netev') \neq \emptyset$, then $\netev\gr \netev'$.
\end{proposition} 
\begin{proof}
  Let  $\locev{\pp}{\procev}\in(\netev\cap\netev')$ and
  $\loc{\netev}=\loc{\netev'} =\set{\pp,\q}$. 
  Then there must exist $\procev_0, \procev'_0$ such that
  $\locev{\q}{\procev_0}\in \netev$ and $\locev{\q}{\procev'_0}\in
  \netev'$.   From
  $\dualevS{\locev{\pp}{\procev}}{\locev{\q}{\procev_0}}$ and
  $\dualevS{\locev{\pp}{\procev}}{\locev{\q}{\procev'_0}}$ it follows that
  $\projb{\procev_0}{\pp} = \projb{\procev'_0}{\pp}$. This, in conjunction
  with the fact that $\ptone{\act{\procev_0}} =
  \ptone{\act{\procev'_0}} = \pp$, implies that  
  neither $\procev_0 < \procev'_0$ nor $\procev'_0 < \procev_0$.
  Thus $\procev_0 \grr \procev'_0$ and therefore $\netev \grr \netev'$
  by \refToDef{netevent-relations}.
\end{proof}

Although conflict is not hereditary in FESs, we prove that a conflict
due to incompatible mutual projections (i.e., a conflict derived by
Clause (\ref{c22}) of \refToDef{netevent-relations}) is semantically
inherited.  Let $\tr\bm n$ denote the prefix of length $n$ of $\bm$.

\begin{proposition}[Semantic conflict hereditariness]
\mylabel{sem-conf}
Let $\locev{\pp}{\procev}\in \netev$ and $\locev{\q}{\procev'}\in\netev'$ with $\pp\not=\q$.
Let $n=min \set{\eh{\projb\procev\q},\eh{\projb{\procev'}\pp}}$. 
If $\neg(\dualev{\tr{(\projb\procev\q)}n}{\tr{(\projb{\procev'}\pp)}n})$,
then there exists no configuration $\ESet$ such that $\netev,
\netev'\in \ESet$. 
\end{proposition}
\begin{proof}
  Suppose ad absurdum that $\ESet$ is a configuration such that
  $\netev, \netev'\in \ESet$. If $\cardin{\projb\procev\q} =
  \cardin{\projb{\procev'}\pp}$ then $\netev \grr\netev '$ by
  \refToDef{netevent-relations}(\ref{c22}) and we reach immediately a
  contradiction.  So, assume $\cardin{\projb\procev\q} >
  \cardin{\projb{\procev'}\pp} = n$. This means that $\cardin{\procev}
  > 1$ and thus there exists a non-empty causal set $E_\netev$ of
  $\netev$ such that $E_\netev \subseteq\ESet$.  Let $\procev_0 <
  \procev$ be such that $\cardin{\projb{\procev_0}\q}
  =\cardin{\projb{\procev'}\pp} =n$.  By definition of causal set,
  there exists $\netev_0 \in E_\netev$ such that
  $\locev{\pp}{\procev_0}\in \netev_0$. By
  \refToDef{netevent-relations}(\ref{c22}) we have then $\netev_0\grr
  \netev'$, contradicting the fact that $\ESet$ is conflict-free. 
\end{proof}


We prove now two further properties of causal sets. 
For the reader familiar with ESs, the notion of causal set may be
reminiscent of that of \emph{prime configuration}~\cite{Winskel80},
which similarly consists of a complete set of causes for a given
event\footnote{In PESs, the prime configuration associated with an
  event is unique, while it is not unique in FESs and more generally
  in Stable ESs, just like a causal set.}. However, there are some
important differences: the first is that a causal set does not include
the event it causes, unlike a prime configuration. The second is that
a causal set only contains direct causes of an event, and thus it is
not downward-closed up to conflicts, as opposed to a prime
configuration. The last difference is that, while a prime
configuration uniquely identifies its caused event, a causal set may
cause different events, as shown in \refToExample{ex:causal-sets}.

A common feature of prime configurations and causal sets is that they
are both finite. For causal sets, this is implied by minimality
together with Clause (\ref{cs2}) of \refToDef{cs}, as shown by the
following lemma. 
\\

\begin{lemma}
\label{cs-prop}
 Let $\netev\in \EvSet \subseteq \DE$. If $E$ is a causal set of $\netev$ in $\EvSet$,
  then $E$ is finite. 
\end{lemma} 
\begin{proof}
  Suppose
  $\netev=\set{\locev\pp\procev,\locev\q{\procev'}}$.  We show that
  $\cardin{E} \leq \cardin{\procev} + \cardin{\procev'} - 2$, where
  $\cardin{E}$ is the cardinality of $E$.  By Condition (\ref{cs2}) of \refToDef{cs}, for
  each $\procev_0 < \procev$ and $\procev'_0 < \procev'$ there must be
  $\netev_0, \netev'_0\in E$ such that $\locev\pp\procev_0 \in
  \netev_0$ and $\locev\q\procev'_0 \in \netev'_0$. Note that
  $\netev_0$ and $\netev'_0$ could possibly coincide. Moreover, there
  cannot be $\netev'\in E$ such that $\locev\pp\procev_0
  \in \netev' \neq \netev_0$ or $\locev\q\procev'_0 \in \netev'
  \neq \netev'_0$, since this would contradict the
  minimality of
  $E$  (and also its conflict-freeness, since by
  \refToProp{prop:conf} we would have $\netev' \grr \netev_0$). 
Hence the number of events in $E$ is at most $(\cardin{\procev}
  -1) + (\cardin{\procev'} - 1)$.
\end{proof}

 A key property of causal sets, which is again shared with prime
configurations, is that each configuration includes a unique causal
set for each n-event in the configuration.

\begin{lemma}\label{csl}
  If $\ESet$ is a configuration of $\ESN{\Nt}$ and $\netev\in\ESet$, then there is a
 unique causal set $E$ of $\netev$ such that $E\subseteq\ESet$.
\end{lemma}
\begin{proof} 
  By \refToDef{def:narrowing}, if $\netev\in \GE(\Nt)$, then $\netev$
  has at least one causal set included in $\GE(\Nt)$.  Let
  $E'=\set{\netev'\in\ESet\mid \netev'\prec\netev}$. By
  \refToDef{configF}, $E'\cup \set{\netev}$ is conflict-free.
  Moreover, if $\locev{\pp}{\procev}\in \netev$ and $\procev' <
  \procev$, then by \refToProp{prop:conf} there is at most one
 $\netev'' \in E'$ such that $\locev{\pp}{\procev'}\in \netev''$. 
  Therefore, $E'\subseteq E$ for some causal set $E$ of $\netev$ by
  \refToDef{cs}.  We show that $E\subseteq E'$.  Assume ad absurdum
  that $\netev_0\in E\backslash E'$.  By definition of causal set,
  $\netev_0\precN\netev$. By definition of $E'$, $\netev_0\not\in E'$
  implies $\netev_0\not\in \ESet$.  By \refToDef{configF} this implies
  $\netev_0\gr\netev_1\prec\netev$ for some $\netev_1\in\ESet$. Then
  $\netev_1\in E'$ by definition of $E'$, and thus $\netev_1\in E$.
  Hence $\netev_0,\netev_1\in E$ and $\netev_0\gr\netev_1$,
  contradicting \refToDef{cs}.
\end{proof}

In the remainder of this section we show that projections of n-event
configurations give p-event configurations. We start by formalising
the projection function of n-events to p-events and showing that it is
downward surjective.

\begin{definition}[Projection of n-events to p-events]
\label{def:proj-network-n-event}
\[
\projnet{\pp}{\netev}=\begin{cases}
\procev     & \text{if } \locev{\pp}{\procev}\in\netev, \\
 undefined   & \text{otherwise}.
\end{cases}
\]
The projection function $\projnet{\pp}{\cdot}$ is extended to sets of
n-events in the obvious way:
\[
\projnet{\pp}{X}  = \set{ \procev \mid 
    \exists\netev \in X  \, .\, \projnet{\pp}{\netev} = \procev}
    \]
\end{definition}

\begin{example}\mylabel{ex:narrowing-ila}
  Let $\set{\netev_1,\netev_2,\netev_3}$ be the configuration defined in \refToExample{ex1}. We get
\[
\projnet{\q}{\set{\netev_1,\netev_2,\netev_3}} =
    \set{\rcvL\pp{\la_1},
      \concat{\rcvL\pp{\la_1}}{\sendL\pr{\la_2}}}
      \]
\end{example}

\begin{example}
Let  $\Nt=\pP{\pp}{\rcvL{\pr}\la;\rcvL{\q}{\la'}} \parN \pP{\q}{\sendL{\pp}{\la'}}$.  Then
\[ 
\GE(\Nt) =\nr{\set{\set{\locev\pp{\concat{\rcvL{\pr}\la}{\rcvL{\q}{\la'}}},\locev\q{\sendL{\pp}{\la'}}}}}=\emptyset
\]
Note that if we did not apply narrowing the set of events of 
$\ESN{\Nt}$ 
would be the singleton
$\set{\locev\pp{\concat{\rcvL{\pr}\la}{\rcvL{\q}{\la'}}},\locev\q{\sendL{\pp}{\la'}}}$,
which would also be a configuration $\ESet$ of $\ESN{\Nt}$. However,
$\projnet{\pp}{\netev}=\set{\concat{\rcvL{\pr}\la}{\rcvL{\q}{\la'}}}$
would not be configuration in $\ES(\PP)$, since it would contain the
event $\concat{\rcvL{\pr}\la}{\rcvL{\q}{\la'}}$ without its cause
$\rcvL{\pr}\la$. 
\end{example}

Narrowing ensures that each projection of the set of n-events of a
network  FES  on one of its participants is downward
surjective (according to~\refToDef{down-onto}):

\begin{lemma}[Downward surjectivity of projections]
\label{prop:down-onto1}
Let $\ESN{\Nt} = (\GE(\Nt), \precN_\Nt , \grr_\Nt)$ and $\ESP{\PP} =
(\ES(\PP), \precP_\PP , \gr_\PP)$ and $\pP{\pp}{\PP}\in\Nt$. Then the
partial function $\projnetfun{\pp}:\GE(\Nt) \rightarrow_* \ES(\PP)$ is
downward surjective.
\end{lemma}
\begin{proof}
  As mentioned already in~\refToSection{sec:eventStr}, any PES $S =
  (E, \leq, \gr)$ may be viewed as a FES, with $\prec$ given by $<$
  (the strict ordering underlying $\leq$). Let $\procev \in \ES(\PP)$
  and $\netev \in \GE(\Nt)$.  Then the property we need to show is:
\[\procev <_\PP \projnet{\pp}{\netev} \implies \exists \netev' \in
\GE(\Nt)~.~\procev = \projnet{\pp}{\netev'} \] Note that $\procev
<_\PP \projnet{\pp}{\netev}$ implies $\projnet{\pp}{\netev} =
\concat{\procev}{\procev'}$ for some $\procev'$. Recall that $\GE(\Nt)
= \nr{\DE(\Nt)}$, where $\nr{\cdot}$ is the narrowing function
(\refToDef{def:narrowing}).\\  By definition of narrowing,
$\occ{\locev{\pp}{\concat{\procev}{\procev'}}}{\GE(\Nt)}$ implies 
that there is $E\subseteq\GE(\Nt)$ such that $E$ is a causal set of
$\netev$ in $\GE(\Nt)$.  Therefore
$\locev{\pp}{\concat{\procev}{\procev'}}\in\netev$
implies $\occ{\locev{\pp}{\procev}}{E}$ and so
$\occ{\locev{\pp}{\procev}}{\GE(\Nt)}$, which is what we wanted to show.
\end{proof}

\begin{theorem}[Projection preserves configurations]
\mylabel{config-preservation}
  If $\pP{\pp}{\PP}\in\Nt$, then $\ESet \in \Conf{\ESN{\Nt}}$ implies
  $\projnet{\pp}{\ESet} \in \Conf{\ESP{\PP}}$.
\end{theorem}

\begin{proof}  Clearly, $\projnet{\pp}{\ESet}$ is
  conflict-free. We show that it is also downward-closed. 
 If $\netev\in\ESet$, by \refToLemma{csl} there is a
  causal set $E$ of $\netev$ such that $E \subseteq\ESet$. If
  $\locev\pp\procev\in\netev$ and $\procev'<\procev$, by \refToDef{cs}
  there is $\netev'\in E$ such that
  $\locev\pp{\procev'}\in\netev'$. We conclude that $\netev'\in\ESet$,
   and therefore $\procev' \in \projnet{\pp}{\ESet}$. 
\end{proof}

 The reader may wonder why our ES semantics for sessions is not
cast in categorical terms, like classical ES semantics for process
calculi~\cite{Win80, CZ97}, where process constructions arise as
categorical constructions (e.g., parallel composition arises
as a categorical product).  In fact, a categorical formulation of our
semantics would not be possible, due to our two-level syntax for
processes and networks, which does not allow networks to be further
composed in parallel. However, it should be clear that our
construction of a network FES from the process PESs of its
components is a form of parallel composition, and the properties
expressed by \refToLemma{prop:down-onto1} and
\refToTheorem{config-preservation} give some evidence that this
construction enjoys the properties usually required for a categorical
product of ESs.


\section{Global Types}
\mylabel{sec:types}

This section is devoted to our type system for multiparty sessions.
Global types describe the communication protocols involving all
session participants.  Usually, global types are projected into local
types and typing rules are used to derive local types for
processes~\cite{CHY08,Coppo2016,CHY16}. The simplicity of our calculus
allows us to project directly global types into processes and to have
exactly one typing rule, see \refToFigure{fig:typing}.
This section is split  in two subsections.\\
The first subsection presents the projection of global types onto
processes, together with the proof of its soundness. Moreover it
introduces a \emph{boundedness} condition on global types, which is
crucial for our type system to ensure  progress.\\
The second subsection presents the type system, as well as an LTS for
global types.  Lastly, the properties of Subject Reduction, Session
Fidelity and Progress are shown.

\subsection{Well-formed Global Types}\label{wfgt}
Global types are built from choices among communications.

\begin{definition}[Global types]
\mylabel{def:GlobalTypes} 
Global types $\G$  are defined by: 
\[
\begin{array}{lll}
      \G~~& \coDefGr & 
    \gt\pp\q i I \la \G 
      ~\mid ~\End
  \end{array}
  \]
where $I$ is not empty, $\la_h\not=\la_k\,$ for all $h,k\in I$, $h\neq k$, i.e. messages
in choices are all different.  
\end{definition}  

As for processes, $ \coDefGr$ indicates that global types are 
defined coinductively. Again, we focus on \emph{regular} terms. 
%

Sequential composition ($;$) has higher precedence than choice
($\GlSy$).  %
When $I$ is a singleton, a choice $ \gt\pp\q i I \la \G$ will be
rendered simply as $\gtCom{\pp}{\q}{\la} \, ; \G$.  In writing
global types,
we omit the final $\End$.

Given a global type, the sequences of decorations of nodes and edges on
the path from the root to an edge in the tree of the global type are traces, in
the sense of \refToDef{traces}.  We denote by $\FPaths{\G}$ the set of
traces of $\G$. By definition, $\FPaths{\End} = \emptyset$ and each
trace in $\FPaths{\G}$ is non-empty.
\label{G-traces}


The set of {\em participants of a global type $\GP$},
$\participant{\GP}$, 
is defined to be the union of the sets of participants of all its traces, namely\\
\centerline{$\participant{\GP} = \bigcup_{\comseq \in
    \FPaths{\G}} \participant{\comseq} $}  Note that the
regularity assumption ensures that the set of participants is finite.

\begin{figure}[t]
\[
 \begin{array}{c}
 \\[-1pt]
\proj\G{\pr} = \inact \text{ if }\pr\not\in\participant\G \\\\
\proj{(\gt\pp\q i I \M \RG)}\pr=\begin{cases}
  \inpP\pp{i}{I}{\M}{\proj{\RG_i}\pr}    & \text{if }\pr=\q, \\
    \oupP\q{i}{I}{\M}{\proj{\RG_i}\pr}    & \text{if }\pr=\pp, \\
      \proj{\RG_1}\pr  & \text{if } \pr\not\in\set{\pp, \q}\text{ and } \pr\in\participant{\RG_1} \text{ and }\\
      &\proj{\RG_i}\pr=\proj{\RG_1}\pr\text{ for all } i \in I
\end{cases}\\[13pt]
\end{array}
\]
\caption{Projection of  global types onto participants.} \mylabel{fig:proj}
\end{figure}
 
\bigskip

The projection of a global type onto participants is given in
\refToFigure{fig:proj}. As usual, projection is defined only when it
is defined on all participants. Because of the simplicity of our
calculus, the projection of a global type, when defined, is simply a
process.  The definition is coinductive, so a global type with an
infinite (regular) tree produces a process with a regular tree. The
projection of a choice type on the sender produces an output choice,
i.e. a process sending one of its possible messages to the receiver
and then acting according to the projection of the corresponding
branch.  Similarly for the projection on the receiver, which produces
a process which is an input choice.  Projection of a choice type on
the other participants is defined only if it produces the same process
for all the branches of the choice. This is a
standard condition for multiparty session types~\cite{CHY08}.\\

Our coinductive definition of global types is more permissive than that based
on the standard $\mu$-notation used in \cite{CHY08}, because it
allows more global types to be projected, as shown by  the
following example.
 \begin{example}\label{de}
   The global type $\G=
   \gtCom\pp\q{}:(\Seq{\la_1}{\gtCom\q\pr{\la_3}}~\GlSyB~\Seq{\la_2}{\G})$
   is projectable and
\begin{itemize}
  \item $\proj{\G}{\pp}=\PP=\sendL{\q}{\la_1} \oplus\Seq{\sendL{\q}{\la_2}}{\PP}$
  \item $\proj{\G}{\q}=\Q=\Seq{\rcvL{\pp}{\la_1}}{\sendL{\pr}{\la_3}} +\Seq{\rcvL{\pp}{\la_2}}{\Q}$
  \item $\proj{\G}{\pr}=\rcvL{\q}{\la_3}$
\end{itemize}
 On the other hand, the corresponding global type based on the
$\mu$-notation 
\[\G'=\mu\ty.  \,\gtCom\pp\q{}:  (\Seq{\la_1}{\gtCom\q\pr{\la_3}}~\GlSyB~\Seq{\la_2}{\ty})\]
 is not projectable because $\proj{\G'}{\pr}$
is not defined. 
\end{example}


 To achieve progress, we need to ensure that each network
participant occurs in every computation, whether finite or
infinite. This means that each type participant must occur in every
path of the tree of the type. Projectability already ensures that each
participant of a choice type occurs in all its branches. This implies
that if one branch of the choice gives rise to an infinite path,
either the participant occurs at some finite depth in this path, or
this path crosses infinitely many branching points in which the
participant occurs in all branches. In the latter case, since the
depth of the participant increases when crossing each branching point,
there is no bound on the depth of the participant over all paths of
the type. Hence, to ensure that all type participants occur in all 
paths, it is enough to require the existence of such bounds. This
motivates the following definition of depth and boundedness.  

 \begin{definition}[Depth  and boundedness]\label{depth}$\;$\\
 Let the two functions $\weight(\comseq,\pp)$ and $\weight(\G,\pp)$ be defined by:
\[\weight(\comseq,\pp)=\begin{cases} n &\text{ if }\comseq =
\concat{\concat{\comseq_1}{\alpha}}{\comseq_2}\text { and }\cardin{\comseq_1} = n-1\text { and }\pp \notin
\participant{\comseq_1}\text { and }\pp \in \participant{\alpha}\\
  0   & \text{otherwise }
\end{cases}
\]
Then

$\weight(\G,\pp)=
       \sup  \{\weight(\comseq,\pp)\ |\ \comseq\in\FPaths{\G}\}$\\
    
We say that a global type $\GP$ is {\em bounded}
   if $\weight(\G',\pp)$ is finite for all subtrees $\G'$ of $\GP$ and
   for all participants $\pp$.
 \end{definition}
If $\weight(\G,\pp)$ is finite, then there are  no paths in
the tree of $\G$ in which $\pp$ is delayed indefinitely. 
Note that if  
 $\weight(\G,\pp)$ is finite, $\G$ may have subtrees $\G'$ for which
 $\weight(\G',\pp)$ is  infinite 
as the following example shows.
 \begin{example}
Consider $\G'= \Seq{\gtCom\q\pr{\la}}{\G}$  where
 $\G$ is as defined in \refToExample{de}. 
 Then we have: 
\[
 \weight(\G',\pp)=2\quad\quad \weight(\G',\q)=1\quad\quad\weight(\G',\pr)=1
 \]
 whereas
\[
 \weight(\G,\pp)=1\quad\quad \weight(\G,\q)=1\quad\quad\weight(\G,\pr)=\infty
 \]
  since\\ \centerline{$\FPaths{\G}=\{\underbrace{\Comm\pp{\la_2}\q\cdots \Comm\pp{\la_2}\q}_n\cdot\Comm\pp{\la_1}\q\cdot\Comm\q{\la_3}\pr\ |\ n\geq 0\}
 \cup\{\Comm\pp{\la_2}\q\cdots\Comm\pp{\la_2}\q\cdots\} $ }  and
  $\sup\{2,3,\ldots\} = \infty$. 
 \end{example}
 
The depths of the participants  in  $\G$
which are not participants of its root communication decrease in the immediate subtrees of $\G$.  

\begin{proposition}\label{dd}
If $\G=\gt\pp\q i I \la \G$ and $\pr\in\participant{\G}{\setminus}\set{\pp,\q}$, then $\weight(\G,\pr)>\weight(\G_i,\pr)$ for all $i\in I$.
\end{proposition}
\begin{proof}
Each trace $\comseq\in\FPaths{\G}$ is of the shape $\concat{\Comm\pp{\la_i}\q}{\comseq'}$ where $i\in I$ and $\comseq'\in\FPaths{\G_i}$. 
\end{proof}

We can now show that the definition of projection given in
\refToFigure{fig:proj} is sound  for bounded global types. 
 
\begin{lemma}\label{pf}
If $\GP$ is bounded, then $\proj\GP\pr$ is a partial function for all $\pr$.  
\end{lemma}

Boundedness and projectability single out the global types we want to use in
our type system.

 \begin{definition} [Well-formed global types] \label{wfs}We say that  the global type 
 $\GP$ is {\em well formed} 
if $\GP$ is bounded 
and $\proj{\GP}{\pp}$ is defined for all  $\pp$.
 \end{definition}
Clearly it is sufficient to check that
  $\proj{\GP}{\pp}$ is defined  for all $\pp\in\participant{\GP}$, since otherwise 
 $\proj{\GP}{\pp}=\inact$.

 \subsection{Type System}
 \begin{figure}[h]
{\small \[
 \begin{array}{c}
\inact\subt\inact~\rulename{ $\subt$ -$\inact$}\quad \cSinferrule{\PP_i\subt\Q_i ~~~~~i\in I}{\inp\pp{i}{I\cup J}{\M}{\PP}\subt \inp\pp{i}{I}{\M}{\Q}}{\rulename{ $\subt$-In}}
\quad \cSinferrule{\PP_i\subt\Q_i ~~~~~i\in
  I}{\oupTx\pp{i}{I}{\M}{\PP}\subt \oup\pp{i}{I}{\M}{\Q}}{\rulename{
    $\subt$-Out}}\\ \\
\inferrule{\PP_i\subt\proj\GP{\pp_i}~~~~~i\in I~~~~~~\participant\GP\subseteq \set{\pp_i\mid i\in I}}
{\derN{\PiB_{i\in I}\pP{\pp_i}{\PP_i}}\GP} ~\rulename{Net}
\\[13pt]
\end{array}
\]}
\caption{Preorder on processes and  network typing rule.} \mylabel{fig:typing}
\end{figure}
The definition of  well-typed network is given in \refToFigure{fig:typing}.
 We first define a preorder on processes, $\PP\leq\Q$, 
meaning that {\em process $\PP$ can be used where we expect process $\Q$}.  
More precisely, $\PP\leq\Q$ if either $\PP$ is equal to $\Q$,   or we
are in one of two situations: 
either both $\PP$ and $\Q$ are output processes with the same
receiver and choice of   messages,  and their continuations after the send
are two processes $\PP'$ and $\Q'$ such that $\PP'\leq\Q'$; or they
are both input processes with the same sender and choice of messages,
and $\PP$ may receive more  messages  than $\Q$ (and thus have more
behaviours) but whenever it receives the same message as $\Q$ their
continuations are two processes $\PP'$ and $\Q'$ such that
$\PP'\leq\Q'$.
 The rules are interpreted coinductively, since the processes may have
infinite (regular) trees.\\ 
 A network is well typed
if all its participants have associated processes that behave as
specified by the projections of a global type.
In Rule \rulename{Net}, the condition
$\participant{\GP}\subseteq\set{\pp_i\mid i\in I}$ ensures that all
participants of the  global type
appear in the network.  Moreover it permits additional participants
that do not appear in the global type,
allowing the typing of sessions containing $\pP{\pp}{\inact}$ for a
fresh $\pp$ --- a property required to guarantee invariance of types
under structural congruence of networks.

\begin{example}\mylabel{exg}
The  first network of   \refToExample{ex:rec2} and the network of \refToExample{ex1}  
can be typed respectively by
\[
\begin{array}{lll}
\G&=&   \gtCom\pp\q{}:(\Seq\la\G~\GlSyB~\la') \\
\G'&=& \Seq{\Seq{\gtCom\pp\q{\la_1}}{\gtCom\q\pr{\la_2}}}{\gtCom\pr\ps{\la_3}}\\
\end{array}
\]
\end{example}

\begin{figure}
 \[
\begin{array}{c}
 \gt\pp\q i I \la \G \stackred{\Comm\pp{\la_j}\q}\G_j~~~~~~j\in I{~~~\rulename{Ecomm}}
 \\ \\
 \prooftree
 \G_i\stackred\alpha\G_i' \quad 
 \text{ for all }
i \in I \quad\participant{\alpha}\cap\set{\pp,\q}=\emptyset
 \justifies
 \gt\pp\q i I \la \G \stackred\alpha\gt\pp\q i I \la {\G'} 
 \using ~~~\rulename{Icomm}
  \endprooftree\\ \\
\end{array}
\]
\caption{
LTS for global types.
}\mylabel{ltgt}
\end{figure}

It is handy to define the LTS for  global types
given in \refToFigure{ltgt}.  Rule \rulename{Icomm} is justified by
the fact that in a projectable global type
$\gt\pp\q i I \la \G$, the behaviours of the participants different
from $\pp$ and $\q$ are the same in all branches, and hence they are
independent from the choice and may be executed before it.  This LTS
respects well-formedness of global types,  as shown in Proposition
\ref{prop:wfs}.

We start with a lemma 
relating the projections of a well-formed global type
with its  transitions.

\begin{lemma}\label{keysr}  Let $\G$ be a well-formed global type. 
\begin{enumerate}
\item\label{keysr1}
If $\proj\G\pp=\oup\q{i}{I}{\M}{\PP}$ and $\proj\G\q=\inp\pp{j}{J}{\M'}{\Q}$, then $I=J$, $\M_i=\M_i'$, $\G\stackred{\Comm\pp{\la_i}\q}\G_i$, $\proj{\G_i}\pp=\PP_i$ and $\proj{\G_i}\q=\Q_i$ for all $i\in I$.
\item\label{keysr2} If $\G\stackred{\Comm\pp{\la}\q}\G'$, then
  $\proj\G\pp=\oup\q{i}{I}{\M}{\PP}$,
  $\proj\G\q=\inp\pp{i}{I}{\M}{\Q}$, where $\la_i=\la$
for some $i\in I$, and 
$\proj{\G'}\pr=\proj\G\pr$ for all $\pr\not\in\set{\pp,\q}$.
\end{enumerate}
\end{lemma}

\begin{proposition}\label{prop:wfs}
If $\G$ is a well-formed global type and $\G\stackred{\Comm\pp{\la}\q}\G'$, then $\G'$ is a well-formed global type.
\end{proposition}
\begin{proof}
If  $\G\stackred{\Comm\pp{\la}\q}\G'$, by  \refToLemma{keysr}(\ref{keysr1}) and (\ref{keysr2})
$\proj{\G'}\pr$ is defined for all $\pr$. 
The proof that $\weight(\G'',\pr)$ for all $\pr$ and $\G''$ subtree
of $\G'$ is easy by induction on the  transition   rules
of \refToFigure{ltgt}. 
\end{proof}
\noindent
Given the previous proposition, we will focus on {\bf well-formed global types  from now on.} \\

 We end this section with the expected proofs of Subject Reduction, Session Fidelity \cite{CHY08,CHY16}  and Progress \cite{Coppo2016,Padovani15}, which use Inversion and Canonical Form lemmas.

\begin{lemma}[Inversion]\mylabel{lemma:InvSync}
If $\derN{\Nt}{\G}$, then $\PP\subt\proj\G{\pp}$ for all  $\pP{\pp}{\PP}\in\Nt$. \end{lemma}

\begin{lemma}[Canonical Form]\mylabel{lemma:CanSync}
If  $\derN\Nt\G$ and
$\pp\in\participant\G$, then $\pP{\pp}{\PP}\in\Nt$ and $\PP\subt\proj\G{\pp}$.
\end{lemma}

\begin{theorem}[Subject Reduction]\mylabel{sr}
If $\derN\Nt\G$ and $\Nt\stackred\alpha\Nt'$, then $\G\stackred\alpha\G'$ and \mbox{$\derN{\Nt'}{\G'}$.}
\end{theorem}
\begin{proof}
Let  $\alpha=\Comm\pp{\la}\q$. 
By Rule
\rulename{Com} of \refToFigure{fig:netred}, $\Nt\equiv \pP{\pp}{\PP}\parN \pP{\q}{\Q}\parN\Nt''$ where
$\PP=\oup\q{i}{I}{\la}{\PP}$ and $\Q=\inp\pp{j}{J}{\la}{\Q}$  and 
$\Nt'\equiv\pP{\pp}{\PP_h}\parN \pP{\q}{\Q_h}\parN\Nt''$ and $\la=\la_h$ for some $h\in I\cap J$.
From \refToLemma{lemma:InvSync} we get
\begin{enumerate}
\item \label{psr1} $\proj\G{\pp}=\oup\q{i}{I}{\la}{\PP'}$ with
  $\PP_i\subt\PP'_i$ for all $i\in I$, from Rule 
\rulename{ $\subt$ -Out} of \refToFigure{fig:typing}, and
\item \label{psr2} $\proj\G{\q}=\inp\pp{j}{J'}{\la}{\Q'}$ with
  $\Q_j\subt\Q'_j$ for all $j\in J'\subseteq J$, from Rule \rulename{
    $\subt$ -In} of \refToFigure{fig:typing}, and
\item \label{psr3}  $\R\subt\proj\G{\pr}$ for all   $\pP{\pr}{\R}\in\Nt''$.
\end{enumerate}
By Lemma~\ref{keysr}(\ref{keysr1}) $\G\stackred{\Comm\pp{\la_h}\q}\G_h$ and $\proj{\G_h}\pp=\PP_h'$ and $\proj{\G_h}\q=\Q'_h$. By Lemma~\ref{keysr}(\ref{keysr2}) $\proj{\G_h}\pr=\proj{\G}\pr$  for all $\pr\not\in\set{\pp,\q}$. We can then choose $\G'=\G_h$.
\end{proof}
\begin{theorem}[Session Fidelity]\mylabel{sf}
If $\derN\Nt\G$ and $\G\stackred\alpha\G'$, then $\Nt\stackred\alpha\Nt'$ and $\derN{\Nt'}{\G'}$. 
\end{theorem}
\begin{proof}
Let  $\alpha=\Comm\pp{\la}\q$. By Lemma~\ref{keysr}(\ref{keysr2}) $\proj\G\pp=\oup\pp{i}{I}{\M}{\PP}$ and $\proj\G\q=\inp\pp{i}{I}{\M}{\Q}$ and $\M=\M_i$ for some $i\in I$ and 
$\proj{\G'}\pr=\proj\G\pr$ for all $\pr\not\in\set{\pp,\q}$. By Lemma~\ref{keysr}(\ref{keysr1})  $\proj{\G'}\pp=\PP_i$ and $\proj{\G'}\q=\Q_i$. From \refToLemma{lemma:CanSync} and \refToLemma{lemma:InvSync} we get \mbox{$\Nt\equiv \pP{\pp}{\PP}\parN \pP{\q}{\Q}\parN\Nt''$} and 
\begin{enumerate}
\item \label{sfc1} $\PP=\oup\q{i}{I}{\la}{\PP'}$ with $\PP'_i\subt\PP_i$ for $i\in I$, from Rule \rulename{ $\subt$ -Out} of \refToFigure{fig:typing}, and
\item \label{sfc2}  $\Q=\inp\pp{j}{J}{\la}{\Q'}$ with $\Q'_j\subt\Q_j$ for $j\in I\subseteq J$, from Rule \rulename{ $\subt$ -In} of \refToFigure{fig:typing},  and 
\item \label{sfc3}  $\R\subt\proj\G{\pr}$ for all   $\pP{\pr}{\R}\in\Nt''$.
\end{enumerate}
We can then choose $\Nt'=\pP{\pp}{\PP'_i}\parN \pP{\q}{\Q'_i}\parN\Nt''$.
\end{proof}

We are now able to prove that in a typable network, every participant
whose process is not terminated may eventually perform a communication. This
property is generally referred to as progress.
\begin{theorem}[Progress]\mylabel{pr}
If $\derN\Nt\G$ and $\pP\pp\PP\in\Nt$, then $\Nt\stackred{\comseq\cdot\alpha}\Nt'$ and $\pp\in\participant\alpha$.
\end{theorem}
\begin{proof}
We prove by induction on $d=\weight(\G,\pp)$ that:   if $\derN\Nt\G$ and $\pP\pp\PP\in\Nt$, then   $\G\stackred{\concat{\comseq}\alpha}\G'$ with $\pp\in\participant\alpha$. This will imply  $\Nt\stackred{\concat{\comseq}\alpha}\Nt'$  by Session Fidelity  (\refToTheorem{sf}). \\
{\it Case $d=1$.}  In this case  $\G= \gt\q\pr i I \la \G $ and $\pp\in\set{\q,\pr}$ and $\G\stackred{\Comm\q{\la_h} \pr}\G_h$ for some $h\in I$ by Rule \rulename{Ecomm}.\\ 
   {\it Case $d>1$.}  In this case  $\G= \gt\q\pr i I {\la} \G $  and $\pp\not\in\set{\q,\pr}$. By \refToLemma{dd} this implies $\weight(\G_i,\pp)<d$ for all $i\in I$. Using Rule \rulename{Ecomm} we get $\G\stackred{\Comm\q{\la_i} \pr}\G_i$ for all $i\in I$.   By Session Fidelity, $\Nt\stackred{\Comm\q{\la_i}\pr}\Nt_i$ and 
   $\derN{\Nt_i}{\G_i}$ for all $i\in I$. Moreover, since  $\pp\not\in\set{\q,\pr}$ we also have  $\pP\pp\PP\in\Nt_i$ for all $i\in I$.
   By induction $\G_i\stackred{\concat{\comseq_i}{\alpha_i}}\G_i'$ with $\pp\in\participant{\alpha_i}$  for all $i\in I$. We conclude $\G\stackred{\Comm\q{\la_i} \pr\cdot\concat{\comseq_i}{\alpha_i}}\G_i'$ for all $i\in I$. 
\end{proof} 
The proof of the progress theorem shows that the execution 
strategy which uses only Rule \rulename{EComm} is fair, since there are no infinite  transition sequences where some participant is stuck.  This is due to the boundedness condition on global types.

\begin{example}
  The second network of \refToExample{ex:rec2} and the network of
  \refToExample{ex2} cannot be typed because they do not enjoy
  progress. Notice that the candidate global type for the second
  network of \refToExample{ex:rec2}:
\[
\G'' = \gtCom\pp\q{}:(\Seq\la{\G''}~\GlSyB~\Seq{\Seq{\la'}{\gtCom\pp\pr{\la}}}{\gtCom\pr\ps{\la'}} ) 
\]
is not bounded, given that $\weight(\G'',\pr)$ and $\weight(\G'',\ps)$
are not finite.\\
Moreover we cannot define a global type whose projections are greater
than or equal to the processes associated with the network of
\refToExample{ex2}.  
\end{example}


\section{Event Structure Semantics of Global Types}
\mylabel{sec:events}

 We define now the event structure associated with a global type, 
 which will be a PES whose events
are equivalence classes of particular traces.

 We recall that a trace $\comseq \in \Comseq$ is a finite sequence of 
communications (see \refToDef{traces}). 
We will use the following notational conventions:

\begin{itemize}
\item We denote by 
  $\at{\comseq}{i}$ the $i$-th
  element of  $\comseq$,  $i > 0$. 
\item If $i \leq j$, we define $\range{\comseq}{i}{j} =
  \at{\comseq}{i} \cdots \at{\comseq}{j}$ to be the subtrace of
  $\comseq$ consisting of the $(j-i+1)$ elements starting from the
  $i$-th one and ending with the $j$-th one.  If $i > j$, we convene
  $\range{\comseq}{i}{j}$ to be the empty trace $\ee$.
\end{itemize}
If not otherwise stated we assume that $\comseq$ has $n$
elements, so $\comseq=\range{\comseq}{1}{n}$. 

\bigskip

We start by defining an equivalence relation on $\Comseq$ which allows
swapping of communications with disjoint participants.

\begin{definition}[Permutation equivalence]\mylabel{def:permEq}
The permutation equivalence on
$\Comseq$ is the least equivalence $\sim$ such that
\[
\concat{\concat{\concat{\comseq}{\alpha}}{\alpha'}}{\comseq'}\,\sim\,\,
\concat{\concat{\concat{\comseq}{\alpha'}}{\alpha}}{\comseq'}
\quad\text{if}\quad \participant{\alpha}\cap\participant{\alpha'}=\emptyset
\]
We denote by $\eqclass{\comseq}$ the equivalence class of the  trace 
$\comseq$, and by $\quotient$ the set of equivalence classes on
$\Comseq$. Note that $\eqclass{\emptyseq} = \set{\emptyseq}\in
\quotient$, and $\eqclass{\alpha}= \set{\alpha} \in \quotient$ for any
$\alpha$.  Moreover $\eh{\comseq'}=\eh{\comseq}\,$ for all
$\comseq'\in\eqclass{\comseq}$. 
\end{definition}
The events associated with a global type, 
called  \emph{
  g-events}  and denoted by $\globev, \globev'$, are equivalence classes of
particular  traces  that we call \emph{pointed}.
Intuitively, in a pointed  trace all communications  but the last one
are causes of
 some subsequent communication. 
Formally:
\begin{definition}[Pointed  trace]\mylabel{pcs}
A  trace   $\comseq = \range{\comseq}{1}{n}$ is said to be \emph{pointed} if
\[
 \mbox{~~for all $i$, $1\leq i<n$,
$\,\participant{\at{\comseq}i}\cap\participant{\range\comseq{(i+1)}n}\not=\emptyset
$}
\]
\end{definition}
Note that the condition of \refToDef{pcs} must be satisfied only by the 
$\at{\comseq}i$ with $i < n$, thus it is vacuously satisfied by any
 trace   of length 1. 
\begin{example}\mylabel{exg6}
  Let $\alpha_1= \Comm\pp{\la_1}\q, \,\alpha_2= \Comm\pr{\la_2}\ps$ and
  $\alpha_3= \Comm\pr{\la_3}\pp$.  Then $\sigma_1 = \alpha_1$ 
  and $\sigma_3 =
  \concat{\concat{\alpha_1}{\alpha_2}}{\alpha_3}\,$ are pointed
   traces,  while  $\sigma_2 = \concat{\alpha_1}{\alpha_2}\,$ is
  \textit{not} a pointed trace.
\end{example}

We use $\last{\comseq}$ to denote the last communication of
$\comseq$.

\begin{lemma}\mylabel{ma1}
Let $\comseq$ be a pointed  trace.  If  $\comseq \sim \comseq'$, then $\comseq'$ is a
 pointed  trace   and $\last{\comseq}=\last{\comseq'}$.  
\end{lemma}
\begin{proof}
  Let $\comseq \sim \comseq'$. By \refToDef{def:permEq} $\comseq'$ is
  obtained from $\comseq$ by $m$ swaps of adjacent communications. The
  proof is by induction on such a number 
 $m$. \\
  If $m=0$ the result is obvious.\\
  If $m>0$, then there exists $\comseq_0$ obtained from $\comseq$ by
  $m-1$ swaps of adjacent communications and there are  
$\comseq_1$,  $\comseq_2$,  $\alpha$ and $\alpha'$ such that
\[
  \comseq_0 =
    \concat{\concat{\concat{\comseq_1}{\alpha}}{\alpha'}}{\comseq_2}\,\sim\,\,
    \concat{\concat{\concat{\comseq_1}{\alpha'}}{\alpha}}{\comseq_2}
    = \comseq'\ \ \mbox{and}\
    \ \participant{\alpha}\cap\participant{\alpha'}=\emptyset
    \]
   By
  induction hypothesis $\comseq_0$ is a pointed  trace  
  and $\last{\comseq}=\last{\comseq_0}$.  Therefore
  $\comseq_2\neq\emptyseq$ since otherwise $\alpha'$ would be the
  last communication of $\comseq_0$ and it cannot be
  $\participant{\alpha}\cap \participant{\alpha'}=\emptyset$.\
  This implies $\last{\comseq}=\last{\comseq'}$.\\
  To show that $\comseq'$ is pointed, since all the communications in
  $\comseq_1$ and $\comseq_2$ have the same successors in
  $\comseq_0$ and $\comseq'$, all we have to prove is  that the
  required property holds for the two swapped communications
  $\alpha'$ and $\alpha$ in $\comseq'$, namely:
\[
\begin{array}{l}
 \participant{\alpha'}\cap(\participant{\alpha}\cup\participant{\comseq_2})\not=\emptyset\\[2pt]
 \participant{\alpha}\cap\participant{\comseq_2}\not=\emptyset\\[2pt]
\end{array}
\]
Since $\participant{\alpha}\cap\participant{\alpha'}=\emptyset$, these
two statements are respectively equivalent to:
\[
\begin{array}{l}
\participant{\alpha'}\cap\participant{\comseq_2}\not=\emptyset \\[2pt]
\participant{\alpha}\cap(\participant{\alpha'}\cup\participant{\comseq_2})\not=\emptyset\\[2pt]
\end{array}
\]
The last two statements are known to hold since $\comseq_0$ is pointed
by induction hypothesis.
\end{proof}

\begin{definition}[Global event]\mylabel{def:glEvent}
Let $\comseq = \concat{\comseq'}{\alpha}\,$ be a pointed
 trace.  
Then $\gamma =
\eqclass{\comseq}$ is a \emph{global event}, also called \emph{g-event}, with communication
$\alpha$, notation  $\comm\gamma=\alpha$.\\
 We denote by $\GEs$ the set of g-events.\end{definition}

Notice that  ${\sf cm}(\cdot)$  is well defined due
to \refToLemma{ma1}.  
 
 \bigskip

 We now introduce an operator called ``retrieval'', which
 applied to a communication $\alpha$ and a g-event $\comocc$, yields
 the g-event corresponding to $\comocc$ before the communication
 $\alpha$ is executed.

  \begin{definition}[Retrieval of g-events before communications]\text{~}\\[-15pt]
  \mylabel{causal-path} 
    \begin{enumerate}
    \item\mylabel{causal-path1} The  {\em retrieval operator} $\circ$ applied to a
  communication and  a g-event  
       is defined by
\[
      \cau{\alpha}\eqclass{\comseq}=\begin{cases}
          \eqclass{\concat{\alpha}{\comseq}} & \text{if~
            $\participant{\alpha}\cap\participant{\comseq} \neq \emptyset$}\\
   \eqclass{\comseq}  & \text{otherwise}
\end{cases}
\]
\item\mylabel{causal-path2} The operator $\circ$ naturally extends to  nonempty traces 
\[
\cau{(\concat\alpha\comseq)}\comocc=\cau\alpha{(\cau\comseq\comocc)}\qquad\comseq\not=\emptyseq
\]
  \end{enumerate}
\end{definition} 

Using the retrieval, we can define the mapping ${\sf ev}(\cdot)$  which,
 applied to a trace $\comseq$, gives the g-event representing
the communication $\last\comseq$ prefixed by its causes occurring in
$\comseq$. 
\begin{definition}
\mylabel{causal-path3}    The {\em g-event generated by a  trace  } 
 is defined by:
\[
 \ev{\concat\comseq\alpha}=\cau\comseq{\eqclass\alpha}
 \]
 \end{definition}
 Clearly $\comm{\ev\comseq}=\last\comseq$.

\bigskip

 We proceed now to define the causality and conflict relations on
g-events. 
To define the conflict relation, 
 it is handy to define the projection of a  trace  on a participant, which 
gives the sequence of the participant's actions  in the trace.  

\begin{definition}[Projection]
\mylabel{def:projection}
\begin{enumerate}
\item The {\em projection} of $\alpha$ onto $\pr$, 
$\projS\alpha\pr$, is defined by:
\[
\projS{\Comm\pp\la\q}\pr=\begin{cases}
      \sendL\q\la & \text{if }\pr=\pp\\
      \rcvL\pp\la & \text{if }\pr=\q\\
      \ee& \text{if }\pr\not\in\set{\pp,\q}
  \end{cases}
  \]
\item The projection of a  trace  $\comseq$ onto $\pr$, $\projS\comseq\pr$, is defined by:
\[
\projS{\ee}\pr=\ee\quad\quad
\projS{(\alpha\cdot\comseq)}\pr=\projS\alpha\pr\cdot\projS{\comseq}\pr
\]
\end{enumerate}
\end{definition}

\begin{definition}[Causality and conflict relations on g-events] \label{sgeo}
The {\em causality} relation $\precP$ and the {\em conflict} relation $\gr$ on the set of g-events $\GEs$ are defined by:
\begin{enumerate}
\item\mylabel{sgeo1} $\comocc\precP\comocc'$
~if~$\comocc=\eqclass\comseq$ and $\comocc'=\eqclass{\concat\comseq{\comseq'}}$ for some $\comseq,\comseq'$;
\item\mylabel{sgeo2}
  $\eqclass{\comseq}\grr\eqclass{\comseq'}$~if~$\projS\comseq\pp\grr\projS{\comseq'}\pp$
  for some $\pp$.
\end{enumerate}
\end{definition}
\noindent If
$\comocc=\eqclass{\concat{\concat\comseq{\alpha}}{\concat{\comseq'}{\alpha'}}}$,
then the communication $\alpha$ must be done before the communication
$\alpha'$. This is expressed by the causality
$\eqclass{\concat\comseq{\alpha}}\precP\comocc$.  An example is
$\eqclass{\Comm\pp\la\q}\precP\eqclass{\concat{\Comm\pr{\la'}\ps}{\concat{\Comm\pp\la\q}{\Comm\ps{\la''}\q}}}$. \\
As regards conflict, note that if $\comseq\sim\comseq'$ then
$\projS\comseq\pp=\projS{\comseq'}\pp$ for all $\pp$, because $\sim$
does not swap communications which share some participant.  Hence,
conflict is well defined, since it does not depend on the trace chosen
in the equivalence class. The condition
$\projS\comseq\pp\gr\projS{\comseq'}\pp$ states that participant $\pp$
does the same actions in both traces up to some point, after which it
performs two different actions in $\comseq$ and $\comseq'$.  For example $\eqclass{\concat{\concat{\Comm\pp\la\q}{\Comm\pr{\la_1}\pp}}{\Comm\q{\la'}\pp}}\gr  \eqclass{\concat{\Comm\pp\la\q}{\Comm\pr{\la_2}\pp}}$, since 
$\projS{(\concat{\concat{\Comm\pp\la\q}{\Comm\pr{\la_1}\pp}}{\Comm\q{\la'}\pp})}\pp=
   \concat{\concat{\sendL\q\la}{\rcvL\pr{\la_1}}}{\rcvL\q{\la'}}\gr\concat{\sendL\q\la}{\rcvL\pr{\la_2}}=
\projS{(\concat{\Comm\pp\la\q}{\Comm\pr{\la_2}\pp})}\pp$.

\begin{definition}[Event structure of a global type] \mylabel{eg}The {\em
    event structure of  the global type
    } 
    $\GP$ is the triple
\[
 \ESG{\GP} = (\EGG(\GP), \precP_\GP , \grr_\GP)
 \]
 where:
\begin{enumerate}
\item\mylabel{eg1a} 
 $\EGG(\GP) =
\set{ \ev{\comseq}\ |\ \comseq\in\FPaths{\G}}$  
\item\mylabel{eg2} $\precP_\GP$ is the restriction of $\precP$ to the set $\EGG(\GP)$;
\item\mylabel{eg3} $\gr_\GP$ is the restriction of $\gr$ to the set $\EGG(\GP)$.
\end{enumerate}
\end{definition}
 Note that, 
 in case the tree of $\G$ is infinite, the set $\EGG(\G)$  is  
 denumerable.  
\begin{example}\mylabel{ex:eventsGT}
  Let
  $\G_1=\Seq{\gtCom\pp\q{\la_1}}{\Seq{\gtCom\pr\ps{\la_2}}{\gtCom\pr\pp{\la_3}}}$
  and
  $\G_2=\Seq{\gtCom\pr\ps{\la_2}}{\Seq{\gtCom\pp\q{\la_1}}{\gtCom\pr\pp{\la_3}}}$.
  Then $\EGG(\G_1)=\EGG(\G_2)=\set{\comocc_1,\comocc_2,\comocc_3}$
  where
  \[
  \comocc_1=\{\Comm\pp{\la_1}\q\}\qquad
    \comocc_2=\{\Comm\pr{\la_2}\ps\}\qquad
    \comocc_3=\{\concat{\concat{\Comm\pp{\la_1}\q}{\Comm\pr{\la_2}\ps}}{\Comm\pr{\la_3}\pp},
    \concat{\concat{\Comm\pr{\la_2}\ps}{\Comm\pp{\la_1}\q}}{\Comm\pr{\la_3}\pp}
    \}
    \]
    with $\comocc_1\precP\comocc_3$ and
  $\comocc_2\precP\comocc_3$. The configurations are
  $\set{\comocc_1}$, $\set{\comocc_2}$,   $\set{\comocc_1,\comocc_2}$    and
  $\set{\comocc_1,\comocc_2,\comocc_3}$, and the proving sequences
  are
 \[
  \comocc_1\qquad \comocc_2\qquad
    \Seq{\comocc_1}{\comocc_2}\qquad
    \Seq{\comocc_2}{\comocc_1}\qquad
    \Seq{\Seq{\comocc_1}{\comocc_2}}{\comocc_3}\qquad
    \Seq{\Seq{\comocc_2}{\comocc_1}}{\comocc_3}
    \]
  If $\G'$ is as in \refToExample{exg}, then
  $\EGG(\G')=\set{\comocc_1,\comocc_2,\comocc_3}$ where
\[
  \comocc_1=\{\Comm\pp{\la_1}\q\}\qquad
    \comocc_2=\{\concat{\Comm\pp{\la_1}\q}{\Comm\q{\la_2}\pr}\}\qquad
    \comocc_3=\{\concat{\concat{\Comm\pp{\la_1}\q}{\Comm\q{\la_2}\pr}}{\Comm{ \pr }{\la_3}{ \ps }}\}
    \]
with
  $\comocc_1\precP\comocc_2\precP\comocc_3$. The configurations are
  $\set{\comocc_1}$, $\set{\comocc_1,\comocc_2}$ and
  $\set{\comocc_1,\comocc_2,\comocc_3}$,  and there is a unique
   proving sequence
  corresponding to each configuration. \end{example}
\begin{theorem}\mylabel{basta12}
Let $\GP$ be a  global type. 
Then  $\ESG{\GP}$  is a prime event structure.
\end{theorem}

\begin{proof}
  We show that $\precP$ and $\gr$ satisfy Properties (\ref{pes2}) and
  (\ref{pes3}) of \refToDef{pes}. Reflexivity and transitivity of
  $\precP$ follow from the properties of concatenation and of
  permutation equivalence. As for antisymmetry, by
  \refToDef{sgeo}(\ref{sgeo1}) $\eqclass{\comseq}\,\leq
  \,\eqclass{\comseq'}$ implies $\comseq' \sim
  \concat{\comseq}{\comseq_1}$ for some $\comseq_1$ and
  $\eqclass{\comseq'}\,\leq \,\eqclass{\comseq}$ implies $\comseq \sim
  \concat{\comseq'}{\comseq_2}$ for some $\comseq_2$.  Hence $\comseq
  \sim \concat{\concat{\comseq}{\comseq_1}}{\comseq_2}$, which implies
  $\comseq_1 = \comseq_2 = \ee$.  Irreflexivity and symmetry of $\gr$
  follow from the corresponding properties of $\gr$ on p-events. \\ 
    As for conflict hereditariness, suppose that $\eqclass{\comseq} \grr
  \eqclass{\comseq'}\precP \eqclass{\comseq''}$.  By
  \refToDef{sgeo}(\ref{sgeo1}) and (\ref{sgeo2}) we have respectively
  that $\concat{\comseq'}{\comseq_1}\sim\comseq''$ for some
  ${\comseq_1}$ and $\projS\comseq\pp\grr\projS{\comseq'}\pp$ for some
  $\pp$.  Hence also
  $\projS\comseq\pp\grr\projS{(\concat{\comseq'}{\comseq_1})}\pp$,
  whence by \refToDef{sgeo}(\ref{sgeo2}) we conclude that
  $\eqclass{\comseq} \grr \eqclass{\comseq''}$. 
\end{proof}

Observe that while our interpretation of networks as FESs exactly
reflects the concurrency expressed by the syntax of networks, our
interpretation of global types as PESs exhibits more concurrency than
that given by the syntax of global types.


\begin{figure}[!h]

\begin{center}

\bigskip

$\Nt = \pP\pp{\q!\la_1 ; \pr!\la\oplus \q!\la_2 ;
  \pr!\la} \parN\pP\q{\pp?\la_1 ; \ps!\la' +
  \pp?\la_2  ; \ps!\la' }\parN\pP\pr{\pp?\la  ; \ps!\la'' } \parN\pP\ps{\q?\la'  ; \pr?\la'' }  $


\setlength{\unitlength}{1mm}
\begin{picture}(100,70)

\put(8,60){{\small $\netev_1 = \set{\locev{\pp}{\,\q!\la_1}, \locev{\q}{\,\pp?\la_1}}$}} 
\put(5,35){{\small $\netev''_1 =\set{\locev{\q}{\,\pp?\la_1\cdot\ps!\la'},
        \locev{\ps}{\,\q?\la'}}$}} 
\put(29,10){{\small $\netev =\set{\locev{\pr}{\pp?\la\cdot \ps! \la''}, \locev{\ps}{\q?\la'\cdot\pr?\la''}}$}}
\put(62,60){{\small $\netev_2 =\set{\locev{\pp}{\,\q!\la_2}, \locev{\q}{\,\pp?\la_2}}$}}
\put(56,35){{\small $\netev''_2 =\set{\locev{\q}{\,\pp?\la_2\cdot\ps!\la'},
        \locev{\ps}{\,\q?\la'}}$}} 
\put(-15,25) {{\small $\netev'_1 =\set{\locev{\pp}{\,\q!\la_1\cdot\pr!\la},
      \locev{\pr}{\,\pp?\la}}$}}
\put(80,25) {{\small $\netev'_2 =\set{\locev{\pp}{\,\q!\la_2{\cdot}\pr!\la}\,, \locev{\pr}{\,\pp?\la}}$}}
\put(50,47.5){{$\gr$}}

\thicklines
\linethickness{0.3mm}
\put(32,57){\vector(0,-1){17}}
\put(72,57){\vector(0,-1){17}}
\multiput(51.5,64)(0,-1){13}{\bf{$\cdot$}}
\multiput(51.5,43)(0,-1){28}{\bf{$\cdot$}}
\put(33,31){\vector(1,-2){7}}
\put(71,31){\vector(-1,-2){7}}
\put(21,57){\vector(-1,-1){25}}
\put(83,57){\vector(1,-1){25}}
\put(13,21){\vector(2,-1){12}}
\put(93,21){\vector(-2,-1){12}}

\end{picture}
\end{center}

\caption{FES of the network $\Nt$.}
\mylabel{fig:network-FES}
\end{figure}
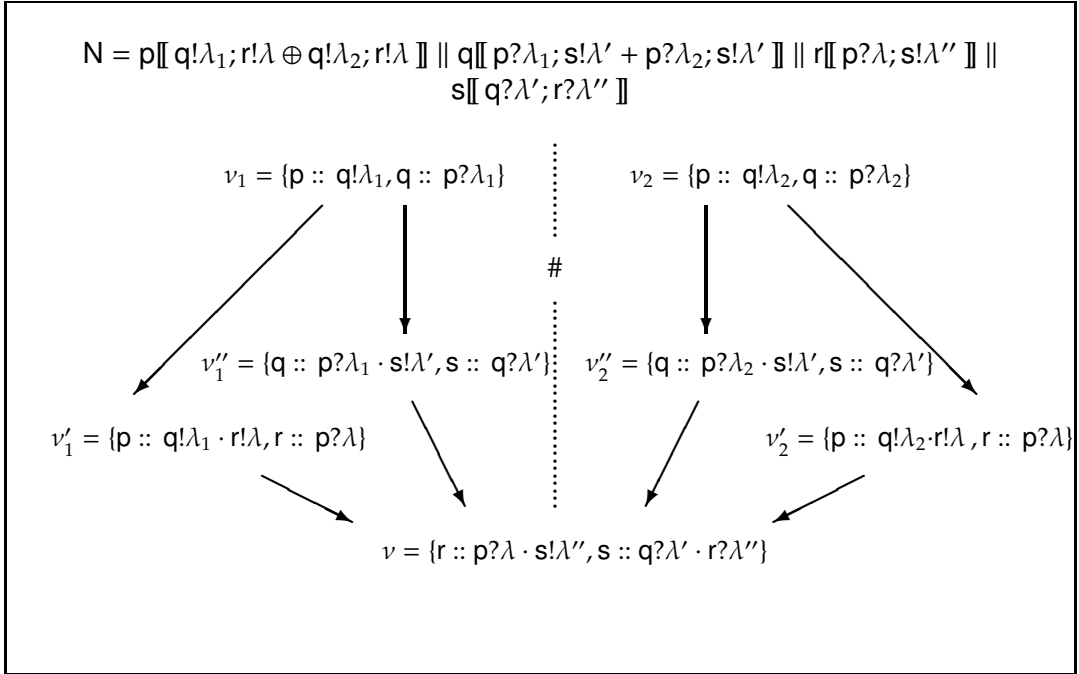


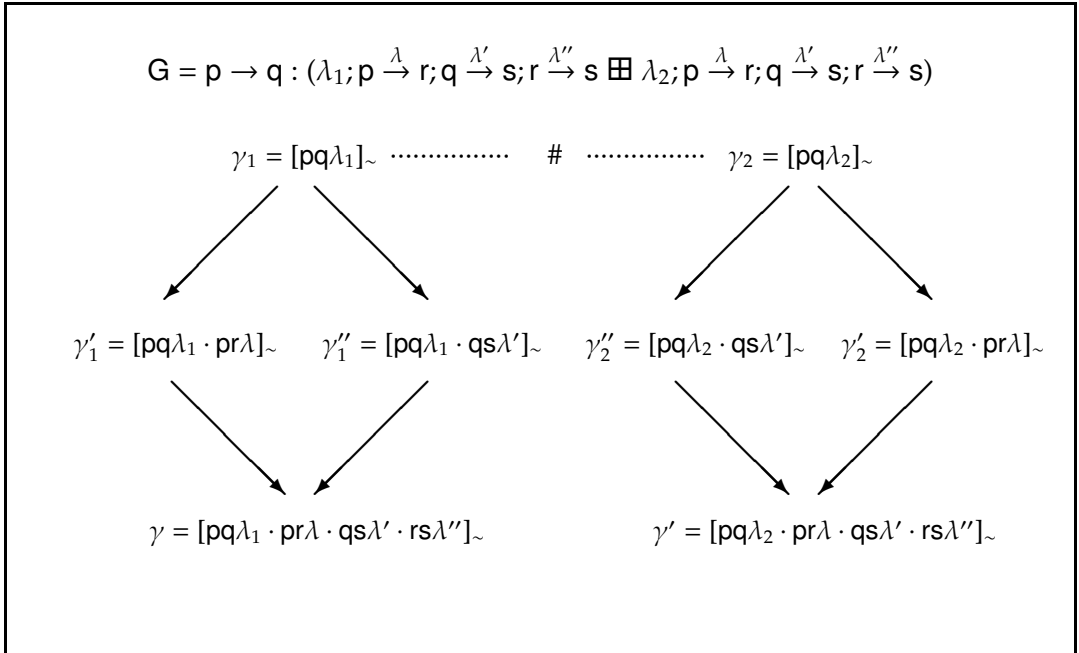
\begin{figure}[!h]

\begin{center}

\bigskip

$\G = 
\gtCom\pp\q{}:(\Seq{\la_1}{\Seq{\Seq{\gtCom\pp\pr{\la}}{\gtCom\q\ps{\la'}}}}{\gtCom\pr\ps{\la''}}~\GlSyB~
\Seq{\la_2}{\Seq{\Seq{\gtCom\pp\pr{\la}}{\gtCom\q\ps{\la'}}}}{\gtCom\pr\ps{\la''}})$

\setlength{\unitlength}{1mm}
\begin{picture}(100,70)

\put(9,60){{\small $\globev_1 = \eqclass{\pp\q\la_1}$}}   
\put(21,35){{\small $\globev''_1 = \eqclass{\pp\q\la_1 \cdot \q\ps\la'}$}}
\put(-2,10){{\small $\globev = \eqclass{\pp\q\la_1\cdot \pp\pr\la \cdot
\q\ps\la' \cdot \pr\ps\la''}$}}
\put(65,10){{\small $\globev' = \eqclass{\pp\q\la_2\cdot \pp\pr\la \cdot
\q\ps\la' \cdot \pr\ps\la''}$}}
\put(75,60){{\small $ \globev_2 = \eqclass{\pp\q\la_2}$}}
\put(-12,35){{\small $\globev'_1 = \eqclass{\pp\q\la_1 \cdot
      \pp\pr\la}$}}
\put(56,35){{\small $\globev''_2 = \eqclass{\pp\q\la_2 \cdot \q\ps\la'}$}}
\put(90,35){{\small $\globev'_2 = \eqclass{\pp\q\la_2\cdot \pp\pr\la}$}}
\put(50,60){{$\gr$}}

\thicklines
\linethickness{0.3mm}
\put(20,57){\vector(1,-1){15}}
\put(35,31){\vector(-1,-1){15}}
\put(1,31){\vector(1,-1){15}}
\put(102,31){\vector(-1,-1){15}}
\put(68,31){\vector(1,-1){15}}
\multiput(30,60)(1,0){16}{\bf{$\cdot$}}
\multiput(56,60)(1,0){16}{\bf{$\cdot$}}
\put(15,57){\vector(-1,-1){15}}
\put(87,57){\vector(1,-1){15}}
\put(83,57){\vector(-1,-1){15}}
\end{picture}
\end{center}

\caption{PES of the type $\G$.}
\mylabel{fig:type-PES}
\end{figure}


We conclude this section with two pictures that summarise the
features of our ES semantics and illustrate the difference between the
FES of a network and the PES of its type.  In general these two ESs
are not isomorphic, unless the FES of the network is itself a PES.

Consider the network FES pictured in \refToFigure{fig:network-FES},
where the arrows represent the flow relation and all the n-events on
the left of the dotted line are in conflict with all the n-events on
the right of the line. In particular, notice that the conflicts
between n-events with a common location are deduced by Clause
(\ref{c21}) of \refToDef{netevent-relations}, while the conflicts
between n-events with disjoint sets of locations, such as $\netev'_1$
and $\netev''_2$, are deduced by Clause (\ref{c22}) of
\refToDef{netevent-relations}. Observe also that the n-event $\netev$
has two different causal sets in $\GE(\Nt)$, namely $\set{\netev'_1,
  \netev''_1}$ and $\set{\netev'_2, \netev''_2}$. The reader familiar
with ESs will have noticed that there are also two prime
configurations\footnote{A prime configuration is a
  configuration with a unique maximal element, its \emph{culminating}
  event.} whose maximal element is $\netev$, namely
$\set{\netev_1,\netev'_1, \netev''_1, \netev}$ and
$\set{\netev_2,\netev'_2, \netev''_2, \netev}$.  It is easy to see
that the network $\Nt$ can be typed with the global type 
shown in \refToFigure{fig:type-PES}. 

Consider now the PES of the type $\G$ pictured in
\refToFigure{fig:type-PES}, where the arrows represent the covering
relation of the partial order of causality and inherited conflicts are
not shown.  Note that while the FES of $\Nt$ has a unique maximal
n-event $\netev$, the PES of its type $\G$ has two maximal g-events
$\globev$ and $\globev'$. This is because an n-event only records the
computations that occurred at its locations, while a g-event records
the global computation and keeps a record of each choice, including
those involving locations that are disjoint from those of its last
communication.  Indeed, g-events correspond exactly to prime
configurations.  

Note that the FES of a network may be easily recovered from the
PES of its global type by using the following function
$\gn{\cdot}$ that maps g-events to n-events:
\[
\gn{\globev} = \set{\locev{\pp}{\projS{\comseq}\pp},
  \locev{\q}{\projS{\comseq}\q}} \quad \mbox{if }
\globev = \eqclass{\comseq}\mbox{ with}~~\partcomm{{\sf cm}(\globev)} = \set{\pp, \q}
\]

On the other hand, the inverse construction is not as direct. First of
all, an n-event in the network FES may give rise to several g-events
in the type PES, as shown by the n-event $\netev$ in
\refToFigure{fig:network-FES}, which gives rise to the pair of
g-events $\globev$ and $\globev'$ in
\refToFigure{fig:type-PES}. Moreover, the local information contained
in an n-event is not sufficient to reconstruct the corresponding
g-events: for each n-event, we need to consider all the prime
configurations that culminate with that event, and then map each of
these configurations to a g-event. Hence, we need a function ${\sf
  ng}(\cdot)$ that maps n-events to sets of prime configurations of
the FES, and then maps each such configuration to a g-event. We will
not explicitly define this function here, since we miss another
important ingredient to compare the FES of a network and the PES of
its type, namely a structural characterisation of the FESs that
represent typable networks.  Indeed, if we started from the FES of a
non typable network, this construction would not be correct. Consider
for instance the network $\Nt'$ obtained from $\Nt$ 
by
omitting the output $\pr! \la$ from the second branch of the process
of $\pp$. 
Then
the FES of $\Nt'$ would not contain the n-event  $\netev'_2$ and the
event $\netev$ would have the unique causal set
$\set{\netev'_1,\netev''_1}$, and the unique prime configuration
culminating with $\netev$ would be
$\set{\netev_1,\netev'_1,\netev''_1, \netev}$.  Then our construction would
give a PES that differs from that of type $\G$ only for the absence of
the g-events $\globev'_2$ and $\globev'$. However, the network $\Nt'$
is not typable and thus we would expect the construction to fail.
Note that in the FES of $\Nt'$, the n-event  $\netev''_2$
is a cause of
$\netev$ but does not belong to any causal set of $\netev$. Thus a
possible well-formedness property to require for FESs to be images of
a typable network would be that each cause of each n-event belong to
some causal set of that event. However, this would still not be enough to
exclude the FES of the non typable network $\Nt''$ obtained from $\Nt'$ 
by omitting the
output $\ps! \la'$ from the second branch of the process of $\q$.

\label{wf-discussion}

To conclude, 
in the absence of a semantic counterpart for the well-formedness
properties of global types,
which eludes us for the time being, we will follow
another approach here, namely we will compare the FESs of networks and
the PESs of their types at a more operational level, by looking at their
configuration domains and by relating their configurations to the
transition sequences of the underlying networks and types.


\section{Equivalence of the two Event Structure Semantics}\label{sec:results}

\begin{figure}[h]
\centering\footnotesize
\begin{minipage}{8cm}
\xymatrix{
&&\ar@{.>}[dl]_{\raisebox{1ex}{\small Th.8.8}}
\netev_1;\ldots;\netev_n=\nec{\comseq}&&\\
\Nt\ar@{-}[r] &\ar@{.>}[d]_{\raisebox{1ex}{\small SR}}{\comseq=\comm{\netev_1}\cdot\ldots\cdot\comm{\netev_n}}\ar[rrr]&&\ar@{.>}[ul]_{\raisebox{1ex}{\small Th.8.7}}& \Nt'\\ 
\G\ar@{-}[rrr] &\ar@{.>}[dr]_{\raisebox{1ex}{\small Th.8.15~~~~}}&&{\comseq=\comm{\globev_1}\cdot\ldots\cdot\comm{\globev_n}}\ar[r]\ar@{.>}[u]_{\raisebox{1ex}{\small SF}}& \G'\\ 
&&\gec{\comseq}=\globev_1;\ldots;\globev_n \ar@{.>}[ur]_{\raisebox{1ex}{\small ~~~~Th.8.16}}&&
}
\end{minipage}
\caption{Isomorphism proof in a nutshell.}\label{ipn}\end{figure}
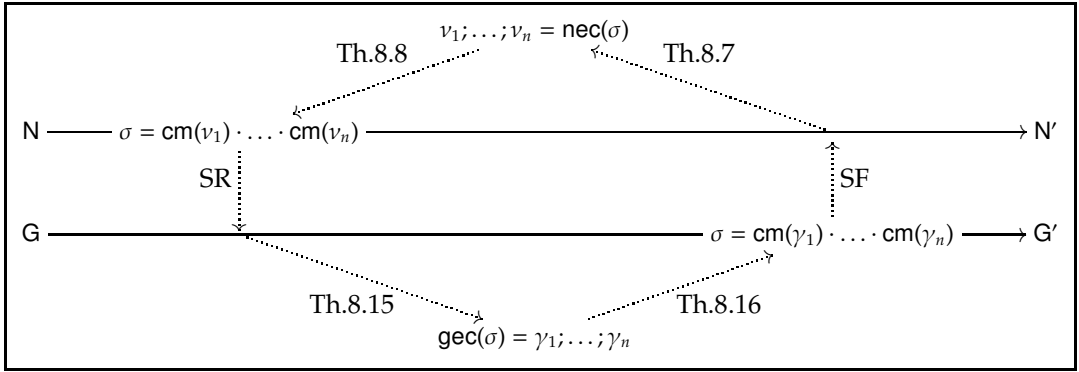

In this section we establish our main result for typable networks,
namely the isomorphism between the domain of configurations of the FES
of such a network and the domain of configurations of the PES of its
global type.   To do so, we will first relate the transition
sequences of networks and global types to the configurations of their
respective ESs. Then, we will exploit our results of Subject Reduction
(\refToTheorem{sr}) and Session Fidelity (\refToTheorem{sf}), which
relate the transition sequences of networks and their global types, to
derive a similar relation between the configurations of their
respective ESs. The schema of our proof is described by the 
diagram in \refToFigure{ipn}. 
Here, SR
stands for Subject Reduction and SF for Session Fidelity,  $\netev_1;\ldots;\netev_n$ and $\globev_1;\ldots;\globev_n$ are proving sequences of $\ESN\Nt$ and $\ESG\G$, respectively.
Finally 
$\nec{\comseq}$ and $\gec{\comseq}$ denote   the proving sequence
of n-events and the proving sequence of g-events corresponding to the trace
$\comseq$ (as given in \refToDef{nec} and \refToDef{gecdef}).  
\refToTheorem{uf12} says that, if
$\Seq{{\netev_1};\cdots}{\netev_n}$ is a proving sequence of
$\ESN{\Nt}$, then $\Nt\stackred\comseq\Nt'$, where $\comseq=\comm{\netev_1}\cdot\ldots\cdot\comm{\netev_n}$. By Subject Reduction (\refToTheorem{sr})
$\G\stackred\comseq\G'$. This implies  
that $\gec{\comseq}$ is a proving
sequence of $\ESG{\RG}$ by \refToTheorem{uf13}. Dually, \refToTheorem{uf14}  says that, if
$\Seq{{\globev_1};\cdots}{\globev_n}$ is a proving sequence of
$\ESG{\G}$, then $\G\stackred\comseq\G'$, where $\comseq=\comm{\globev_1}\cdot\ldots\cdot\comm{\globev_n}$.  By Session Fidelity (\refToTheorem{sf})  $\Nt\stackred\comseq\Nt'$. Lastly $\nec{\comseq}$ is a proving
sequence of $\ESN{\Nt}$ by \refToTheorem{uf10}. The equalities in the top and bottom lines are proved in Lemmas~\ref{ecn}(\ref{ecn22}) and~\ref{ecg}(\ref{ecg1}).

 This section is divided in two subsections: \refToSection{lozenges},
 which handles the upper part of the above diagram, and
 \refToSection{bullets}, which handles the lower part of the diagram and
 then connects the two parts using both SR and SF within
 \refToTheorem{iso}, our closing result.

\subsection{Relating Transition Sequences of Networks and Proving
  Sequences of their ESs}\label{lozenges}

 The aim of this subsection is to relate the traces that label
the transition sequences of networks with the configurations of their
FESs.  We start by showing how network communications affect
n-events in the associated ES.  To this end  we define two
partial operators $\lozenge$ and $\blacklozenge$, which applied to a
communication $\alpha$ and an n-event $\netev$ yield another n-event
$\netev'$ (when defined), which represents the event $\netev$ before
the communication $\alpha$ or after the communication $\alpha$,
respectively. We call ``retrieval'' the $\lozenge$ operator (in
agreement with \refToDef{causal-path}) and ``residual'' the
$\blacklozenge$ operator.

Formally, the operators $\lozenge$ and $\blacklozenge$ are defined as follows.

\begin{definition}[Retrieval and residual of n-events with
  respect to communications]\text{~}\\[-15pt]\mylabel{def:PostPre}
\begin{enumerate}
\item\mylabel{def:PostPre1} The {\em retrieval operator} $\lozenge$ applied to a
  communication and a  located event  returns the located event
  obtained
  by  prefixing the process event by the projection of the communication:
\[
  \preP{\locev{\pp}{\event}}\alpha=\locev{\pp}{\concat{(\projS{\alpha}{\pp})}\event}
  \]
\item\mylabel{def:PostPre2} The {\em residual operator} $\blacklozenge$ applied to a
  communication and a located event returns the located event obtained
  by erasing
  from the process event the projection of the communication (if possible):
\[
  \postP{\locev{\pp}\event}\alpha=
    \locev{\pp}{\event'}
      \quad
      \text{if }\event=
        \concat{(\projS{\alpha}{\pp})}{\event'}
        \]

  \item\mylabel{def:PostPre3} The operators $\lozenge$ and
    $\blacklozenge$ naturally extend to n-events  and to traces:
\[
\begin{array}{c}
 \preP{\set{\locev{\pp}{\event},\locev{\q}{\event'}}}{\alpha}=
  \set{\preP{\locev{\pp}{\event}}{\alpha},\preP{\locev{\q}{\event'}}{\alpha}}
  \\
\postP{\set{\locev{\pp}{\event},\locev{\q}{\event'}}}{\alpha}=
  \set{\postP{\locev{\pp}\event}{\alpha},\postP{\locev{\q}{\event'}}{\alpha}}
 \\
  \pre\netev{\epsilon}=\netev\qquad
  \pre\netev{(\concat\alpha\comseq)}=\preP{\pre\netev\comseq}\alpha\qquad\post\netev{(\concat\alpha\comseq)}=\postP{\post\netev\alpha}\comseq \qquad \comseq\not=\emptyseq
  \end{array}
  \]
    \end{enumerate}
\end{definition}
\noindent
Note that the operator $\lozenge$ is always defined. Instead
$\post{\locev{\pr}\event}{\Comm\pp\la\q}$ is undefined if $\pr \in
\set{\pp,\q}$ and either $\event$ is just one atomic action or
$\projS{\Comm\pp\la\q}\pr$ is not the first atomic action of $\event$.

\bigskip

 The retrieval and residual operators are inverse of each other. Moreover they preserve the flow and conflict relations. 

\begin{lemma}[Properties of 
  retrieval and residual for n-events]\text{~}\\[-10pt]
\mylabel{prop:prePostNet}
\begin{enumerate}
\item \mylabel{ppn1} If $\post{\netev}{\alpha}$ is defined, then
  $\preP{\post{\netev}{\alpha}}{\alpha}=\netev$;
\item \mylabel{ppn1b} 
$\postP{\pre{\netev}{\alpha}}{\alpha}=\netev$;
 \item  \mylabel{ppn2b}  If  $\netev\precN \netev'$, 
then $\pre{\netev}{\alpha}\precN \pre{\netev'}{\alpha}$;
\item \mylabel{ppn2} If  $\netev\precN \netev'$ and both $\post{\netev}{\alpha}$ and
  $\post{\netev'}{\alpha}$ are defined, then
  $\post{\netev}{\alpha}\precN \post{\netev'}{\alpha}$; 
\item \mylabel{ppn3b} If  $\netev\grr \netev'$, then $\pre{\netev}{\alpha}\grr
  \pre{\netev'}{\alpha}$;
\item \mylabel{ppn3} If  $\netev\grr
  \netev'$ and both $\post{\netev}{\alpha}$ and $\post{\netev'}{\alpha}$
  are defined, then $\post{\netev}{\alpha}\grr\post{\netev'}{\alpha}$;
  \item \mylabel{ppn7} If  $\pre{\netev}{\alpha}\grr
  \pre{\netev'}{\alpha}$, then $\netev\grr \netev'$.
\end{enumerate}
\end{lemma}

Starting from the trace $\comseq \neq \ee$  that labels a transition
sequence  in a network, one can
reconstruct the corresponding sequence of n-events in its 
FES.  Recall
that  $\range\comseq1{i}$ is the prefix of length $i$ of 
$\comseq$ and $\range{\comseq}{i}{j}$ is the empty trace if $i\geq
j$.


\begin{definition}[Building sequences of n-events from traces]
  \mylabel{nec} If $\comseq$ is a trace with
  $\at\comseq{i}=\Comm{\pp_i}{\la_i}{\q_i}$, $1\leq i\leq n$, we
  define
  the {\em sequence of n-events corresponding to $\comseq$} by
  \[
  \nec{\comseq}=\Seq{\netev_1;\cdots}{\netev_n}
  \]
 where
  $\netev_i=\pre{\set{\locev{\pp_i}{\sendL{\q_i}{\la_i}},\locev{\q_i}{\rcvL{\pp_i}
        {\la_i}}}}{\range\comseq1{i-1}}$ for $1\leq i\leq n$.
  \end{definition}
   
  It is immediate to see that, if $\comseq=\Comm{\pp}{\la}{\q}$, then
  $\nec{\comseq}$ is the event
  $\set{\locev{\pp}{\sendL{\q}{\la}},\locev{\q}{\rcvL{\pp}{\la}}}$.

\bigskip
%
 We show now that two n-events occurring in $\nec{\comseq}$ cannot be in conflict  and that from $\nec{\comseq}$ we can recover $\comseq$. Moreover we relate the retrieval and residual operators with the mapping $\nec{\cdot}$. 

\begin{lemma}[Properties of $\nec{\cdot}$]\mylabel{ecn}\text{~}\\[-10pt]
\begin{enumerate}
\item\mylabel{ecn2} Let
  $\nec{\comseq}=\Seq{\netev_1;\cdots}{\netev_n}$.  Then 
\begin{enumerate}
 \item\mylabel{ecn22} $\comm{\netev_i}=\at\comseq{i}$  for all $i$, $1\leq i\leq n$;
\item\mylabel{ecn21} If $1\leq h,k\leq n$, then $\neg(\netev_h\gr\netev_k)$.
  \end{enumerate}
\item\mylabel{ecn1} $\neg(\nec\alpha\grr\pre\netev\alpha)$ for all $\netev$.
\item\mylabel{ecn3} Let $\comseq=\alpha\cdot\comseq'$ and
  $\comseq'\neq \ee$.  If
  $\nec{\comseq}=\Seq{\netev_1;\cdots}{\netev_n}$ and
  $\nec{\comseq'}=\Seq{\netev'_2;\cdots}{\netev'_n}$, then
  $\pre{\netev'_i}\alpha=\netev_i$ and
  $\post{\netev_i}\alpha=\netev_i'$ for all $i$, $2\leq i\leq n$.
 \end{enumerate}
  \end{lemma}
  \begin{proof}
  (\ref{ecn22}) 
 Immediate  from \refToDef{nec},  since
$\comm{\pre{\netev}{\sigma}}=\comm{\netev}$ for any event $\netev$. 

(\ref{ecn21}) We show that neither Clause (\ref{c21}) nor
Clause (\ref{c22}) of \refToDef{netevent-relations} can be used to
derive $\netev_h\grr\netev_k$.  Notice that $\netev_i=
\set{\locev{\pp_i}{\projS{{\range\comseq1{i}}}{\pp_i}},
  \locev{\q_i}{\projS{{\range\comseq1{i}}}{\q_i}}}$. So if
$\locev{\pp}\procev\in{\netev_h}$ and
$\locev{\pp}{\procev'}\in{\netev_k}$  with $h<k$,  then
 either  $\procev < \procev'$  or $\procev =
\procev'$. 
Therefore Clause (\ref{c21}) 
does not apply.  If
$\locev{\pp}{\procev}\in{\netev_h}$ and $\locev{\q}{\procev'}\in{\netev_k}$ and 
  $\pp \neq \q$ and $\cardin{\proj\procev\q} =
  \cardin{\proj{\procev'}\pp}$,
  then it must be
  $\dualev{\proj\procev\q =
    \proj{(\projS{\range\comseq1{h}}{\pp})}{\q}}
{\proj{(\projS{\range\comseq1{k}}{\q})}{\pp}
= \proj{\procev'}\pp}$.  Therefore Clause (\ref{c22}) 
cannot be used.

     (\ref{ecn1}) We show that neither Clause (\ref{c21}) nor
    Clause (\ref{c22}) of \refToDef{netevent-relations} can be used to
    derive $\nec\alpha\grr\pre\netev\alpha$.  Let
    $\participant{\alpha} = \set{\pp, \q}$. Then $\nec{\alpha}
    =\set{\locev{\pp}{\projS{\alpha}{\pp}},
      \locev{\q}{\projS{\alpha}{\q}}}$. Note that
    $\locev{\pp}{\procev}\in \pre{\netev}{\alpha}$ iff $\procev =
    \concat{(\projS{\alpha}{\pp})}{\procev'}$ and
    $\locev{\pp}{\procev'}\in \netev$. Since $\projS{\alpha}{\pp} <
    \concat{(\projS{\alpha}{\pp})}{\procev'}$, Clause (\ref{c21})
    of \refToDef{netevent-relations}
    cannot be used. Now suppose $\locev{\pr}{\procev}\in
    \pre{\netev}{\alpha}$ for some $\pr \notin \set{\pp, \q}$.  In
    this case $\proj{(\projS{\alpha}{\pp})}{\pr} =
    \proj{(\projS{\alpha}{\q})}{\pr} = \ee$. Therefore,  since
    $\dualev{\ee}{\ee}$, Clause
    (\ref{c22}) of \refToDef{netevent-relations} does not
    apply.

(\ref{ecn3}) Notice that $\at{\comseq}{i}=\at{\comseq'}{i-1}$ for all
$i$, $2\leq i\leq n$. Then,
by \refToDef{nec}
\[
\begin{array}{lll}\netev_i &=&\pre{\nec{\at{\comseq}{i}}}{\range\comseq1{i-1}}=
\pre{(\pre{\nec{\at{\comseq}{i}}}{\range\comseq2{i-1}})}\alpha=\\
&&\pre{(\pre{\nec{\at{\comseq'}{i-1}}}{\range{\comseq'}1{i-2}})}\alpha=\pre{\netev'_i}\alpha
\end{array}
\]
for all $i$, $2\leq i\leq n$.
\\
By
\refToLemma{prop:prePostNet}(\ref{ppn1b})
$\pre{\netev'_i}\alpha=\netev_i$ implies
$\post{\netev_i}\alpha=\netev_i'$ for all $i$, $2\leq i\leq n$.
  \end{proof}
  

 It is handy to notice that if $\post{\netev}{\alpha}$ is undefined and $\netev$ is an event of a network with communication $\alpha$, then either $\netev=\nec{\alpha}$ or $\netev\grr\nec{\alpha}$. 

\begin{lemma}\mylabel{dc} 
If $\Nt\stackred\alpha\Nt'$ and $\netev\in\GE(\Nt)$, then $\netev=\nec{\alpha}$ or $\netev\grr\nec{\alpha}$ or $\post{\netev}{\alpha}$ is defined.
\end{lemma}

\begin{proof}
Let
$\nec{\alpha} =\set{\locev{\pp}{\projS{\alpha}{\pp}},
  \locev{\q}{\projS{\alpha}{\q}}}$ and
$\netev=\set{\locev{\pr}{\event}, \locev{\ps}{\event'}}$.  By
\refToDef{def:PostPre}(\ref{def:PostPre3}) $\post{\netev}{\alpha}$ is
defined iff $\event= \concat{(\projS{\alpha}{\pr})}{\event_0}$ and
$\event'= \concat{(\projS{\alpha}{\ps})}{\event'_0}$ for some
$\event_0, \event'_0$.\\
There are 2 possibilities:
\begin{itemize}
\item
$\set{\pr,\ps} \cap \set{\pp,\q} = \emptyset$. Then
$\projS{\alpha}{\pr} = \projS{\alpha}{\ps} = \ee$ and
$\post{\netev}{\alpha} = \netev$;
\item $\set{\pr,\ps} \cap \set{\pp,\q} \neq \emptyset$. 
Suppose $\pr = \pp$. There are three possible subcases:
\begin{enumerate}
\item \mylabel{ila1} $\event = \concat{\pi}{\actseq}$ with $\pi \neq \projS{\alpha}{\pp}$.
Then $\locev{\pr}{\event} \grr \locev{\pp}{\projS{\alpha}{\pp}}$ and
thus $\netev\grr\nec{\alpha}$;
\item \mylabel{ila2} 
$\event = \projS{\alpha}{\pp}$. Then either
$\event'= \projS{\alpha}{\q}$ and $\netev=\nec{\alpha}$, or
$\event' \neq \projS{\alpha}{\q}$ 
and $\netev\grr\nec{\alpha}$ by
\refToProp{prop:conf};
\item \mylabel{ila3} 
$\event = \concat{(\projS{\alpha}{\pp})}{\event_0}$. Then 
$\post{\locev{\pp}{\event}}{\alpha} = \locev{\pp}{\event_0}$. Now, if
$\ps \neq \q$ we have $\post{\locev{\ps}{\event'}}{\alpha} =
\locev{\ps}{\event'}$, and thus $\post{\netev}{\alpha} =
\set{\locev{\pp}{\event_0}, \locev{\ps}{\event'}}$.
Otherwise,
$\netev=\set{\locev{\pp}{\concat{(\projS{\alpha}{\pp})}{\event_0}},
  \locev{\q}{\event'}}$. By \refToDef{n-event} $\dualevS{\locev{\pp}{\concat{(\projS{\alpha}{\pp})}{\event_0}}}{\locev{\q}{\event'}}$, which implies 
  $\event'=\concat{(\projS{\alpha}{\q})}{\event'_0}$ for some $\event'_0$.
\end{enumerate}
\end{itemize} 
\end{proof}

 The following lemma, which is technically quite challenging, 
  relates the n-events of two networks which differ for one communication by means of the retrieval and residual operators. 

\begin{lemma}\label{epp}
Let $\Nt\stackred{\alpha}\Nt'$. Then 
\begin{enumerate}
\item\label{epp1} $\set{\nec\alpha}\cup\set{\pre\netev\alpha\mid\netev\in\GE(\Nt')}\subseteq\GE(\Nt)$; 
\item\label{epp2}  $\set{\post\netev\alpha\mid\netev\in\GE(\Nt)\text{ and }\post\netev\alpha\text{ defined} }\subseteq\GE(\Nt')$. 
\end{enumerate}
\end{lemma}

 We may now prove the correspondence between the traces labelling
the transition sequences of a network and the proving sequences of its
FES.  

\begin{theorem}\mylabel{uf10}  If $\Nt\stackred\comseq\Nt'$, 
  then $\nec{\comseq}$ is a proving sequence in $\ESN{\Nt}$.
\end{theorem}
\begin{proof}
  The proof is by  induction on ${\comseq}$.\\
 \emph{Base case.}  Let $\comseq=\alpha$. 
From
  $\Nt\stackred\alpha\Nt'$ and  \refToLemma{epp}(\ref{epp1})
$\nec{\alpha}\in\GE(\Nt)$. Since $\nec{\alpha}$ has no causes, 
by \refToDef{provseq} we conclude that $\nec{\alpha}$
is a proving sequence in $\ESN{\Nt}$.
\\
\emph{Inductive case.}
\noindent
Let
  $\comseq=\concat{\alpha}{\comseq'}$.
From
$\Nt\stackred{\comseq}\Nt'$ we get
$\Nt\stackred{\alpha}\Nt''\stackred{\comseq'}\Nt'$ for some
$\Nt''$. Let $\nec{\comseq}=\Seq{\netev_1;\cdots}{\netev_{n}}$ and
$\nec{\comseq'}=\Seq{\netev'_2;\cdots }{\netev'_{n}}$.  By induction
$\nec{\comseq'}$ is a proving sequence in 
$\ESN{\Nt''}$.  \\
We show that $\nec{\comseq}$ is a proving sequence in 
$\ESN{\Nt}$.  By
\refToLemma{ecn}(\ref{ecn21}) $\nec{\comseq'}$ is conflict free.  By
\refToLemma{ecn}(\ref{ecn3}) $\netev_i=\pre{\netev_i'}{\alpha}$ for
all $i$, $2\leq i\leq n$. This implies $\netev_i\in\GE(\Nt)$ for all
$i$, $2\leq i\leq n$ by \refToLemma{epp}(\ref{epp1}) and
$\neg{(\netev_1\gr\netev_j)}$ for all $i,j$, $2\leq i,j\leq n$ by
\refToLemma{prop:prePostNet}(\ref{ppn7}).   Finally, since
$\netev_1=\nec{\alpha}$,  by \refToLemma{ecn}(\ref{ecn1}) we
obtain $\neg{(\netev_1\gr\netev_i)}$ for all $i$, $2\leq i\leq n$.  We
conclude that $\nec{\comseq}$ is conflict-free  and included in $\GE(\Nt)$.    Let
$\netev\in\GE(\Nt)$ and $\netev\precN \netev_k$ for some $k$, $1\leq
k\leq n$.  This implies $k > 1$  since $\nec{\alpha}$ has no
causes. Hence $\netev_k=\pre{\netev_k'}{\alpha}$.  By \refToLemma{dc},  we know that 
$\netev=\nec{\alpha}$ or $\netev\gr\nec{\alpha}$ or
$\post{\netev}{\alpha}$ is defined.  We consider the three
cases. Let $\participant\alpha=\set{\pp,\q}$. \\  
 {\em Case $\netev = \nec{\alpha}$}. In this case
we conclude immediately since $\nec{\alpha} = \netev_1$ and $1<k$. \\
{\em Case $\netev\gr\nec{\alpha}$}.  Since $\nec{\alpha}=
\netev_1$,  if
$\netev_1\prec\netev_k$  we are done. 
If $\netev_1\not\prec\netev_k$,  then $\loc{\netev_k} \cap \set{\pp,\q} =
\emptyset$   
otherwise $\netev_1\grr\netev_k$. 
We get 
$\netev_k=\pre{\netev_k'}{\alpha}=\netev_k'$.  Since
$\netev\precN \netev_k$, there exists
$\locev\pr\procev\in \netev$ and $\locev\pr{\procev'}\in
\netev_k=\netev_k'$ such that $\procev<\procev'$,
where $\pr \notin \set{\pp,\q}$
because $\pr \in \loc{\netev_k}$. 
 Since
$\nec{\comseq'}$ is a proving sequence in  
$\ESN{\Nt''}$,  by
\refToLemma{csl} there is $\netev_h'  \in \GE(\Nt'')$ such that $\locev\pr\procev\in
\netev_h'$.  Since $\pre{\locev\pr\procev}\alpha=\locev\pr\procev$ we
get $\locev\pr\procev\in \netev_h$. This implies
$\netev_h\prec\netev_k$,  where $\netev_h\gr\netev$  by
\refToProp{prop:conf}. 
\\
{\em Case $\post{\netev}{\alpha}$ defined}. We get
$\post{\netev}{\alpha}\precN{\netev'_k}$ by
\refToLemma{prop:prePostNet}(\ref{ppn2}).  Since $\nec{\comseq'}$ is a
proving sequence in $\ESN{\Nt''}$, there is $h<k$ such that either
$\post{\netev}{\alpha}=\netev_h'$ or
$\post{\netev}{\alpha}\grr\netev_h'\prec\netev_k'$.  In the first case
$\netev=\preP{\post{\netev}{\alpha}}{\alpha}=\pre{\netev_h'}{\alpha}=\netev_h$
by \refToLemma{prop:prePostNet}(\ref{ppn1}).  
In the second case: \begin{itemize}\item from
  $\post{\netev}{\alpha}\grr\netev_h'$ we get
  $(\preP{\post{\netev}{\alpha}}{\alpha})\grr(\pre{\netev_h'}{\alpha})$
  by
  \refToLemma{prop:prePostNet}(\ref{ppn3b}),
  which implies $\netev\gr\netev_h$ by
  \refToLemma{prop:prePostNet}(\ref{ppn1}), and
\item from $\netev_h'\prec\netev_k'$ we get
  $(\pre{\netev_h'}{\alpha})\prec(\pre{\netev_k'}{\alpha})$ by
  \refToLemma{prop:prePostNet}(\ref{ppn2b}),  namely 
  $\netev_h\prec\netev_k$. \qedhere
\end{itemize}
\end{proof}

\begin{theorem}\mylabel{uf12}
If  
$\Seq{\netev_1;\cdots}{\netev_n}$ is a proving sequence in $\ESN{\Nt}$, 
then $\Nt\stackred\comseq\Nt'$, where $\comseq=\comm{\netev_1}\cdots\comm{\netev_n}$.
\end{theorem}

\begin{proof}
  The proof is by induction on $n$.\\
  Case $n=1$.  Let
  $\netev_1=\set{\locev{\pp}{\actseq\cdot\sendL{\q}{\la}},
\locev{\q}{\actseq'\cdot\rcvL{\pp}{\la}}}$.
   Then  $\comm{\netev_1}=\Comm{\pp}{\la}{\q}$.  We first show
  that $\actseq=\actseq'=\ee$. Assume ad absurdum that
  $\actseq\neq\ee$  or  $\actseq'\neq\ee$. By narrowing, this
  implies that there is $\netev\in\GE(\Nt)$ such that
  $\netev\prec\netev_1$, 
 contradicting the fact that $\netev_1$ is a proving sequence.\\ 
By
  \refToDef{netev-relations}(\ref{netev-relations1}) we have
  $\Nt=\pP{\pp}{\PP}\parN \pP{\q}{\Q}\parN\Nt_0$ with
  $\sendL{\q}{\la}\in \ES(\PP)$ and $\rcvL{\pp}{\la}\in\ES(\Q)$.
  Whence by \refToDef{esp}(\ref{ila-esp1}) we get
  $\PP=\oup\q{i}{I}{\la}{\PP}$ and $\Q=\inp\pp{j}{J}{\la}{\Q}$ where
  $\la = \la_k$ for some $k\in
  I \cap J$. Therefore 
\[
  \Nt\stackred{\Comm{\pp}{\la}{\q}}\pP{\pp}{\PP_k}\parN
    \pP{\q}{\Q_k}\parN\Nt_0
    \]
    Case $n>1$. Let $\netev_1$ and $\Nt$
  be as in the basic case, $\Nt''=\pP{\pp}{\PP_k}\parN
  \pP{\q}{\Q_k}\parN\Nt_0 $ and $\alpha=\Comm{\pp}{\la}{\q}$.
 Since $\Seq{\netev_1;\cdots}{\netev_n}$ is a proving sequence, we have
$\neg(\netev_l\grr\netev_{l'})$ for all $l, l'$ such that $1\leq l, l'\leq n$. 
Moreover, for all $l$, $2\leq l\leq n$ we have
$\netev_l\not=\netev_1=\nec\alpha$, thus
$\post{\netev_l}{\alpha}$ is defined by \refToLemma{dc}.  Let  $\netev_l'=\post{\netev_l}{\alpha}$ for all $l$, $2\leq l\leq n$,  then
$\netev_l'\in\GE(\Nt'')$ by \refToLemma{epp}(\ref{epp2}).  
\\
 We show  that
$ \Seq{\netev_2';\cdots}{\netev_n'}$   is a
proving sequence in $\ESN{\Nt''}$.  First notice that for all
$l$, $2\leq l\leq n$, $\neg(\netev_l\grr\netev_{l'})$ implies
$\neg(\netev_l'\grr\netev_{l'}')$ by 
\refToLemma{prop:prePostNet}(\ref{ppn3b}) and (\ref{ppn1}). Let now 
$\netev\prec\netev_h'$ for some $h$, $2\leq h\leq n$.  By
\refToLemma{prop:prePostNet}(\ref{ppn2b}) and (\ref{ppn1})
$\pre\netev\alpha\prec\preP{\post{\netev_h}{\alpha}}{\alpha}=\netev_h$. This
implies by \refToDef{provseq} that there is $h'<h$ such that either
$\pre\netev\alpha=\netev_{h'}$ or
$\pre\netev\alpha\gr\netev_{h'}\prec\netev_h$. Therefore, since
$\netev_l'$ is defined for all $l$, $2\leq l\leq n$, we get either $\netev=\netev_{h'}'$  by
\refToLemma{prop:prePostNet}(\ref{ppn1b})  or
$\netev\gr\netev_{h'}'\prec\netev_h'$  by
\refToLemma{prop:prePostNet}(\ref{ppn3}) and
(\ref{ppn2}). \\
By induction $\Nt''\stackred{\comseq'}\Nt'$ where
$\comseq'=\comm{\netev_2'}\cdots\comm{\netev_n'}$.
Since $\comm{\netev_l}=\comm{\netev_l'}$ for all $l$,
$2\leq l\leq n$ we get $\comseq=\alpha\cdot\comseq'$.  Hence
$\Nt\stackred{\alpha} \Nt''\stackred{\comseq'}\Nt'$ is the required
transition sequence.
\end{proof}

%

\subsection{Relating Transition Sequences of Global Types and Proving
  Sequences of their ESs}\label{bullets}

 In this subsection, we relate the traces that label the
transition sequences of global types with the configurations of their
PESs.  As for n-events, we need retrieval and residual operators
for g-events. The first  operator was already introduced in
\refToDef{causal-path}, so we only need to define the second, which is
given next.

\begin{definition}[Residual of g-events after communications]\text{~}\\[-10pt]\mylabel{def:PostPreGl}
\begin{enumerate} 
\item\mylabel{def:PostPreGl1}  The {\em  residual operator} $\bullet$ applied to a 
communication  and  a g-event is defined by:   
\[
\postG{\eqclass\comseq}{{\alpha}}=\begin{cases}
 \eqclass{\comseq'}      & \text{if }  \comseq\sim\concat{\alpha}{\comseq'}\text { and } \comseq'\neq\emptyseq\\
 \eqclass{\comseq}    & \text{if  $\participant{\alpha}
   \cap \participant{\comseq} = \emptyset$ } 
\end{cases}
   \]
\item\mylabel{def:PostPreGl2} The operator 
$\bullet$ naturally extends to 
nonempty
traces: 
\[
\postG{\comocc}{(\concat\alpha\comseq)}=\postG{(\postG{\comocc}\alpha)}\comseq \qquad  \comseq\not=\emptyseq
\]
  \end{enumerate}
\end{definition} 

 The operator $\bullet$
gives the global event obtained by erasing the
communication,  if it occurs in head position (modulo
$\sim$) in  
the event and leaves the event unchanged if the participants of the global event and of the communication are disjoint. 
Note that the  operator
$\postG{\eqclass\comseq}{\alpha}$ is undefined whenever either 
$\eqclass\comseq=\{\alpha\}$ or  
 one of the participants of $\alpha$ occurs in $\comseq$
but its first communication is different from $\alpha$. 

\bigskip

 The following lemma 
 gives some simple properties
of the retrieval and residual operators for g-events. 
The first five  statements  correspond to those of \refToLemma{prop:prePostNet}
for n-events. 
The last three  statements  
give properties 
 that are relevant  only for the operators $\circ$ and $\bullet$.
\begin{lemma}[Properties of 
  retrieval and residual for g-events]\text{~}\\[-10pt]\mylabel{prop:prePostGl}
\begin{enumerate}
\item \mylabel{ppg1a} If $\postG{\comocc}{\alpha}$ is defined, then $\preG{(\postG{\comocc}{\alpha})}{\alpha}=\comocc$;
\item \mylabel{ppg1b} 
$\postG{(\preG{\comocc}{\alpha})}{\alpha}=\comocc$;
 \item  \mylabel{ppg4a}  If  $\comocc_1< \comocc_2$, 
then $\preG{\comocc_1}{\alpha}< \preG{\comocc_2}{\alpha}$;
\item \mylabel{ppg4b} If  $\comocc_1<\comocc_2$ and  both $\postG{\comocc_1}{\alpha}$ and
  $\postG{\comocc_2}{\alpha}$ 
  are 
  defined, then
  $\postG{\comocc_1}{\alpha}< \postG{\comocc_2}{\alpha}$;
 \item \mylabel{prop:prePostGl5}  If $\comocc_1\gr \comocc_2$, 
then $\preG{\comocc_1}{\alpha}\gr \preG{\comocc_2}{\alpha}$;
  \item \mylabel{ppg3} If $\comocc<\preG{\comocc'}{\alpha}$, then either $\comocc=\eqclass\alpha$ or $\postG{\comocc}{\alpha}<{\comocc'}$;
 \item \mylabel{prop:prePostGl6}  If $\participant{\alpha_1}\cap\participant{\alpha_2}=\emptyset$, then $\preG{(\preG{\comocc}{\alpha_2})}{\alpha_1}=\preG{(\preG{\comocc}{\alpha_1})}{\alpha_2}$;
 \item \mylabel{prop:prePostGl7} If $\participant{\alpha_1}\cap\participant{\alpha_2}=\emptyset$ and both $\postG{(\preG{\comocc}{\alpha_1})}{\alpha_2}$, $\postG{\comocc}{\alpha_2}$ are defined, then $\preG{(\postG{\comocc}{\alpha_2})}{\alpha_1}= \postG{(\preG{\comocc}{\alpha_1})}{\alpha_2}$.
\end{enumerate}
\end{lemma}

 The next lemma relates the retrieval and residual operator with the global types which are branches of choices. 

\begin{lemma}\label{paldf}
 The following hold: \\[-10pt]
\begin{enumerate}
\item\label{paldf1} If $\comocc\in\EGG( \G)$, then $\preG\comocc{\Comm\pp{\la}\q}\in \EGG( \gt\pp\q i I \la \G)$, where $\la=\la_k$ and $\G=\G_k$ for some $k\in I$;
\item\label{paldf2} If $\comocc\in\EGG( \gt\pp\q i I \la \G)$ and $\postG\comocc{\Comm\pp{\la_k}\q}$ is defined, then $\postG\comocc{\Comm\pp{\la_k}\q}\in\EGG( \G_k)$, where $k\in I$.
\end{enumerate}
\end{lemma}
\begin{proof}
(\ref{paldf1}) By \refToDef{eg}(\ref{eg1a}) $\comocc\in\EGG( \G)$ implies $\comocc=\ev\comseq$ for some $\comseq\in\FPaths\G$. Since $\preG\comocc{\Comm\pp{\la}\q}=\ev{\concat{\Comm\pp{\la}\q}\comseq}$ by \refToDef{causal-path} and $\concat{\Comm\pp{\la}\q}\comseq\in\FPaths{ \gt\pp\q i I \la \G}$ we conclude $\preG\comocc{\Comm\pp{\la}\q}\in \EGG( \gt\pp\q i I \la \G)$ by \refToDef{eg}(\ref{eg1a}).

(\ref{paldf2}) By \refToDef{eg}(\ref{eg1a}) $\comocc\in\EGG( \gt\pp\q i I \la \G)$ implies $\comocc=\ev\comseq$ for some $\comseq\in\FPaths{\gt\pp\q i I \la \G}$. We get $\comseq=\concat{\Comm\pp{\la_h}\q}{\comseq'}$ with $\comseq'\in\FPaths{\G_h}$ or $\comseq'=\ee$ for some $h\in I$. The hypothesis $\postG\comocc{\Comm\pp{\la_k}\q}$ defined implies either $h=k$ and $\comseq'\not=\ee$ or $\participant{\comseq'}\cap\set{\pp,\q}=\emptyset$ and $\postG\comocc{\Comm\pp{\la_k}\q}=\ev{\comseq'}$ by \refToDef{def:PostPreGl}(\ref{def:PostPreGl1}). In the first case $\comseq'\in\FPaths{\G_k}$.
In the second case $\comseq''\in\FPaths{\G_k}$ for some $\comseq''\sim\comseq'$ by definition of projection,  which prescribes the same behaviours to all participants different from $\pp,\q$,  see \refToFigure{fig:proj}. We conclude  $\postG\comocc{\Comm\pp{\la_k}\q}\in\EGG( \G_k)$ by \refToDef{eg}(\ref{eg1a}).
\end{proof}

The following lemma 
plays the role  
of  \refToLemma{epp} for n-events.

\begin{lemma}\label{paltr}
Let $\G\stackred\alpha \G'$.
\begin{enumerate}
\item\label{paltr1} If $\comocc\in\EGG(\G')$, then $\preG\comocc{\alpha}\in \EGG(\G)$;
\item\label{paltr2} If $\comocc\in\EGG(\G)$ and $\postG\comocc{\alpha}$ is defined, then $\postG\comocc{\alpha}\in\EGG( \G')$.
\end{enumerate}
\end{lemma}

 We show next that each trace gives rise to a sequence of
g-events,  compare with \refToDef{nec}.  

\begin{definition}[Building sequences of g-events from traces]
\label{gecdef}
We define the {\em sequence of global events corresponding to a
   trace  $\comseq$} by
\[
  \gec{\comseq}=\Seq{\comocc_1;\cdots}{\comocc_n}
  \]
  where
   $\comocc_i=\ev{\range\comseq1i}$ for all $i$,  $1\leq i\leq n$. 
  \end{definition}
  
  We show that $\gec\cdot$ has similar properties as  $\nec\cdot$, see \refToLemma{ecn}(\ref{ecn2}).  The proof is straightforward.  
  
  \begin{lemma}\mylabel{ecg}
Let  $\gec{\comseq}=\Seq{\comocc_1;\cdots}{\comocc_n}$. 

\begin{enumerate}
\item \mylabel{ecg1} $\comm{\comocc_i}=\at\comseq{i}$  for all $i$,  $1\leq i\leq n$. 
\item \mylabel{ecg2} If $1\leq h,k\leq n$, then $\neg (\comocc_h\gr\comocc_k)$;
 \end{enumerate}
 
  \end{lemma}
  
 We may now prove the correspondence between the traces labelling
the transition sequences of a global type and the proving sequences of
its PES. Let us stress the difference between the set of traces
$\FPaths{\G}$ of a global type $\G$ as defined at page
\pageref{G-traces} and the set of traces that label the transition
sequences of $\G$, which is a larger set due to the internal Rule
\rulename{Icomm} of the LTS for global types given in \refToFigure{ltgt}.

\begin{theorem}\mylabel{uf13}
 If 
$\G\stackred\comseq \G'$,
then
 $\gec{\comseq}$ is a proving sequence in $ \ESG{ \G}$. 
\end{theorem}
\begin{proof}
By induction on ${\comseq}$.\\
 \emph{Base case.} 
Let $\comseq=\alpha$,  then $\gec{\alpha}=\eqclass\alpha$.   We use a further induction on 
the inference of the transition  $ \G\stackred\alpha \G'$.\\
 Let $ \G=\gt\pp\q i I \la \G$, $\G'= \G_h$  and $\alpha=\Comm\pp{\la_h}\q$ for some $h\in I$. By \refToDef{eg}(\ref{eg1a}) $\eqclass{\Comm\pp{\la_h}\q}\in \EGG( \G)$.\\
Let $\G= \gt\pp\q i I \la \G$ and $\G'= \gt\pp\q i I \la {\G'}$ and $ \G_i\stackred{\alpha}\G_i'$  for all $ i \in I$ and $\participant\alpha\cap\set{\pp,\q}=\emptyset$. By induction $\eqclass{\alpha}\in\EGG( \G_i)$ for all $ i \in I$. 
By \refToLemma{paldf}(\ref{paldf1}) $\preG{\eqclass{\alpha}}{\Comm\pp{\la_i}\q}\in\EGG(\G)$ for all $ i \in I$.
By \refToDef{eg}(\ref{eg1a}) $\preG{\eqclass{\alpha}}{\Comm\pp{\la_i}\q}=\eqclass{\alpha}$, since $\participant\alpha{\cap}\set{\pp,\q}{=}\emptyset$. We conclude $\eqclass{\alpha}\in\EGG(\G)$. \\
\emph{Inductive case.}   Let
$\comseq=\concat{\alpha}{\comseq'}$ with $\comseq'\not=\ee$.  
  From
 $\G\stackred{\comseq}\G'$ we get    $\G\stackred{\alpha}\G_0\stackred{\comseq'}\G'$ for some $\G_0$.  Let $\gec{\comseq}=\Seq{\comocc_1;\cdots}{\comocc_{n}}$ and $\gec{\comseq'}=\Seq{\comocc'_2;\cdots}{\comocc'_{n}}$.
  By induction 
    $\gec{\comseq'}$  is a proving
    sequence in $\ESG{\G_0}$. 
   By Definitions~\ref{gecdef} and~\ref{causal-path}  $\comocc_i=\preG{\comocc'_i}{\alpha}$, which implies 
    $\postG{\comocc_i}{\alpha}=\comocc'_i$ by \refToLemma{prop:prePostGl}(\ref{ppg1b}) for  all $i$,  $2\leq i\leq n$.\\
    We can show that $\comocc_1=\eqclass\alpha\in\EGG(\G)$ as in the proof of the base case. By \refToLemma{paltr}(\ref{paltr1}) $\comocc_i\in\EGG(\G)$ since $\comocc_i'\in\EGG(\G_0)$ and 
    $\postG{\comocc_i}{\alpha}=\comocc'_i$ for  all $i$,  $2\leq i\leq n$.
We  prove that   $\gec{\comseq}$  is a
proving sequence in  $\ESG{\G}$. 
Let
$\comocc<  \comocc_k$ for some $k$, $1\leq k\leq n$.  Note that this implies 
$k>1$. 
 Since $\comocc_k=\preG{\comocc'_k}{\alpha}$ by \refToLemma{prop:prePostGl}(\ref{ppg3})
    either $\comocc=\eqclass\alpha$ or $\postG{\comocc}{\alpha}<\comocc'_h$. 
    If $\comocc=\eqclass\alpha=\comocc_1$ we are done. Otherwise $\postG{\comocc}{\alpha}\in\EGG(\G_0)$ by \refToLemma{paldf}(\ref{paldf2}).  
 Since  $\gec{\comseq'}$  is a proving
    sequence in $\ESG{\G_0}$, there is $h<k$ such that $\postG{\comocc}{\alpha}=\comocc_h'$ and this implies $\comocc=\preG{(\postG{\comocc}{\alpha})}{\alpha}=\preG{\comocc_h'}{\alpha}=\comocc_h$  by \refToLemma{prop:prePostGl}(\ref{ppg1a}). 
\end{proof}

\begin{theorem}\mylabel{uf14}
  If $\Seq{\comocc_1;\cdots}{\comocc_n}$ is a proving sequence in $
  \ESG{ \G}$, then $ \G\stackred\comseq \G'$, where
  $\comseq=\concat{\concat{\comm{\comocc_1}}\cdots}{\comm{\comocc_n}}$.
\end{theorem}

\begin{proof} 
The proof is by induction on the length $n$ of the proving sequence. Let $\comm{\comocc_1}=\alpha$ and $\set{\pp,\q}=\participant\alpha$.\\
 {\it Case $n=1$.}   Since $\globev_1$ is the
first event of a proving sequence, we have $\globev_1 =\eqclass{\alpha}$. We show this case by induction on $d=\weight(\G,\pp)=\weight(\G,\q)$.\\
{\it Case $d=1$.} Let $\alpha = \Comm\pp\la\q$ and $\G=\gt\pp\q i I \la \G$ and $\la=\la_h$ for some $h\in I$. Then $\G\stackred\alpha\G_h$ by rule \rulename{Ecomm}.\\
{\it Case $d>1$.} Let  $\G= \gt\pr\ps i I \la \G$ and $\set{\pr,\ps}\cap\set{\pp,\q}=\emptyset$. By \refToDef{def:PostPreGl}(\ref{def:PostPreGl1})  $\postG{\comocc_1}{\Comm\pr{\la_i}\ps}$
is defined for all $i\in I$ since $\set{\pr,\ps}\cap\set{\pp,\q}=\emptyset$. This implies $\postG{\comocc_1}{\Comm\pr{\la_i}\ps}\in\EGG(\G_i)$  for all $i\in I$ by \refToLemma{paldf}(\ref{paldf2}).
By induction hypothesis $\G_i\stackred\alpha\G'_i$ for all $i\in I$. Then we can apply rule  \rulename{Icomm} to derive $\G\stackred\alpha\gtp\pr\ps i I \la \G$.\\
{\it Case $n>1$.}    
Let $\G\stackred\alpha\G''$
be the transition as obtained from the base case. 
We show that $\postG{\comocc_j}{\alpha}$ is defined for all $j$, $2\leq j\leq n$. 
If $\postG{\comocc_k}{\alpha}$ were undefined
for some $k$, $2\leq k\leq n$, then by
\refToDef{def:PostPreGl}(\ref{def:PostPreGl1}) either
$\comocc_k=\comocc_1$ or $\comocc_k=\eqclass\comseq$ with
$\comseq\not\sim\concat\alpha{\comseq'}$ and
$\participant\alpha\cap\participant\comseq\not=\emptyset$. In the second case
$\pro\alpha{\pp}\grr\pro\comseq{\pp}$ or $\pro\alpha{\q}\grr\pro\comseq{\q}$, which
implies $\comocc_k\grr\comocc_1$.  So both cases are impossible.  If
$\postG{\comocc_j}{\alpha}$ is defined, by
\refToLemma{paltr}(\ref{paltr2}) we get $
\postG{\comocc_j}{\alpha}\in\EGG(\G'')$ for all $j$,
$2\leq j\leq n$.\\
 We show that $\comocc'_2 ;\cdots ;\comocc'_n$ is a
proving sequence in $\ESG{\G''}$ where $\comocc'_j =
\postG{\comocc_j}{\alpha}$ for all $j$,
$2\leq j\leq n$.
By \refToLemma{prop:prePostGl}(\ref{ppg1a}) $\comocc_j =
\preG{\comocc'_j}{\alpha}$ for all $j$, $2\leq j\leq n$. 
Then by
\refToLemma{prop:prePostGl}(\ref{prop:prePostGl5}) no two
events in  the sequence $\comocc'_2 ;\cdots ;\comocc'_n$  can be in conflict. 
Let $\comocc\in\EGG(\G'')$ and $\comocc< \comocc'_h$
for some $h$, $2\leq h\leq n$.  
By  \refToLemma{paltr}(\ref{paldf1})  $\preG{\comocc}\alpha$ and $\preG{\comocc'_h}\alpha$  belong to $\EGG(\G)$. By
\refToLemma{prop:prePostGl}(\ref{ppg4a})
$\preG{\comocc}\alpha<\preG{\comocc'_h}{\alpha}$. 
By \refToLemma{prop:prePostGl}(\ref{ppg1a})
$\preG{\comocc'_h}{\alpha}=\comocc_h$.
Let $\comocc' = \preG{\comocc}\alpha$. 
Then
$\comocc' < \comocc_h$ implies, by \refToDef{provseq} 
and the fact that $\ESG{\G}$ is a PES, that there is $k<h$ 
such that $\comocc'
=\comocc_k$.  By \refToLemma{prop:prePostGl}(\ref{ppg1a})
we get $\comocc=\postG{\comocc'}{\alpha} = \postG{\comocc_k}\alpha
= \comocc'_k$.\\
Since $\comocc'_2 ;\cdots ;\comocc'_n$ is a proving
sequence in $\ESG{\G''}$, by induction
$\G''\stackred{\comseq'}\G'$  where $\comseq' = 
\concat{\concat{\comm{\comocc'_2}}\ldots}{\comm{\comocc'_n}}$.  Let
$\comseq=\concat{\concat{\comm{\comocc_1}}\ldots}{\comm{\comocc_n}}$. Since
$\comm{\comocc'_j} = \comm{\comocc_j}$ for all $j, 2\leq j\leq n$, we
    have $\comseq = \concat{\alpha}{\comseq'}$. Hence
$\G\stackred\alpha\G''\stackred{\comseq'}\G'$
is the required transition sequence. 

\end{proof}

 The last ingredient required  to prove our main theorem is
the following separation result from~\cite{BC91} (Lemma 2.8 p. 12):

\begin{lemma}[Separation~\cite{BC91}]
\mylabel{separation}
Let $S=(E,\prec, \gr)$ be a flow event structure and $\ESet, \ESet' \in \Conf{S}$
be such that $\ESet \subset \ESet'$.
Then there exist $e \in \ESet'\backslash \ESet$ such that $\ESet \cup \set{e} \in \Conf{S}$.
\end{lemma}

We may now  finally  show the
correspondence between the configurations of the  FES of a network
and the configurations of the PES  of its global type.
Let $\simeq$ denote isomorphism on domains of configurations.

\begin{theorem}[Isomorphism]\mylabel{iso}
 If $\derN\Nt\RG$, then $\CD{\ESN{\Nt}} \simeq \CD{\ESG{\RG}}$.
\end{theorem}
\begin{proof}
  By \refToTheorem{uf12} if $\Seq{{\netev_1};\cdots}{\netev_n}$ is a
  proving sequence of $\ESN{\Nt}$, then $\Nt\stackred\comseq\Nt'$
  where $\comseq=\comm{\netev_1}\cdots\comm{\netev_n}$. By  applying iteratively  Subject
  Reduction (\refToTheorem{sr}) $\G\stackred\comseq\G'$ and
  $\derN{\Nt'}{\RG'}$.  By \refToTheorem{uf13} 
  $\gec\comseq$ is a
  proving sequence of $\ESG{\RG}$.
  
  By \refToTheorem{uf14} if $\Seq{{\globev_1};\cdots}{\globev_n}$ is a
  proving sequence of $\ESG{\RG}$, then $\G\stackred\comseq\G'$ where
  $\comseq=\comm{\globev_1}\cdots\comm{\globev_n}$. By  applying iteratively  Session
  Fidelity (\refToTheorem{sf}) $\Nt\stackred\comseq\Nt'$ and
  $\derN{\Nt'}{\RG'}$.  By \refToTheorem{uf10} 
  $\nec\comseq$ is a
  proving sequence of $\ESN{\Nt}$.

Therefore we have a bijection between $\CD{\ESN{\Nt}}$ and
$\CD{\ESG{\RG}}$, given by $\nec\comseq \leftrightarrow \gec\comseq$
for any $\comseq$ generated by the (bisimilar) LTSs of $\Nt$ and $\RG$.

   We show now that this bijection preserves inclusion of
  configurations.  By
  \refToLemma{separation} it is enough to prove that if 
  $\Seq{{\netev_1};\cdots}{\netev_n} \in
  \Conf{\ESN{\Nt}}$ is mapped to $\Seq{{\globev_1};\cdots}{\globev_n} \in \Conf{\ESG{\RG}}$, then
  $\Seq{\Seq{{\netev_1};\cdots}{\netev_n}}\netev \in
  \Conf{\ESN{\Nt}}$  iff  
  $\Seq{\Seq{{\globev_1};\cdots}{\globev_n}}\globev \in
  \Conf{\ESG{\RG}}$,  where
  $\Seq{\Seq{{\globev_1};\cdots}{\globev_n}}\globev$ is the image of
$\Seq{\Seq{{\netev_1};\cdots}{\netev_n}}\netev$ under the bijection. I.e. let $\nec{\concat{\comseq}{\alpha}} =
\Seq{{\netev_1};\cdots}{\netev_n;\netev}$ and $\gec{\concat{\comseq}{\alpha}} =
\Seq{{\globev_1};\cdots}{\globev_n;\globev}$. This implies $\comseq=\comm{\netev_1}\cdots\comm{\netev_n}=
\comm{\globev_1}\cdots\comm{\globev_n}$ and $\alpha=\comm\netev=\comm\globev$ by  Lemmas~\ref{ecn} and~\ref{ecg}. 

By \refToTheorem{uf12}, if
$\Seq{{\netev_1};\cdots}{\netev_n;\netev}$ is a proving sequence of
$\ESN{\Nt}$, then $\Nt\stackred\comseq\Nt_0\stackred\alpha\Nt'$.
By  applying iteratively  Subject Reduction (\refToTheorem{sr})
$\G\stackred\comseq\G_0\stackred\alpha\G'$ and $\derN{\Nt'}{\RG'}$.
By \refToTheorem{uf13} 
$\gec{\concat\comseq\alpha}$ is a proving
sequence of $\ESG{\RG}$.

By \refToTheorem{uf14}, if
  $\Seq{{\globev_1};\cdots}{\globev_n;\globev}$ is a proving sequence
  of $\ESG{\RG}$, then $\G\stackred\comseq\G_0\stackred\alpha\G'$.
By  applying iteratively  Session Fidelity (\refToTheorem{sf})
  $\Nt\stackred\comseq\Nt_0\stackred\alpha\Nt'$ and
  $\derN{\Nt'}{\RG'}$.   By \refToTheorem{uf10} 
  $\nec{\concat\comseq\alpha}$ is a proving sequence of
  $\ESN{\Nt}$. 
  \end{proof}


\section{Related Work and Conclusions}
\mylabel{sec:related}

Event Structures (ESs) were introduced in Winskel's PhD
Thesis~\cite{Win80} and in the seminal paper by Nielsen, Plotkin and
Winskel~\cite{NPW81}, roughly in the same frame of time as Milner's
calculus CCS~\cite{Mil80}. It is therefore not surprising that the
relationship between these two approaches for modelling concurrent
computations started to be investigated very soon afterwards. The
first interpretation of CCS into ESs was proposed by Winskel
in~\cite{Win82}. This interpretation made use of Stable ESs, because
PESs, the simplest form of ESs, appeared not to be flexible enough to
account for CCS parallel composition. Indeed, since CCS parallel
composition allows for two concurrent complementary actions to either
synchronise or occur independently in any order, each pair of such
actions gives rise to two forking computations: this requires
duplication of the same continuation process for these forking
computations in PESs, while the continuation process may be shared by
the forking computations in Stable ESs, which allow for disjunctive
causality.  Subsequently, ESs (as well as other nonsequential
``denotational models'' for concurrency such as Petri Nets) have been
used as the touchstone for assessing noninterleaving operational
semantics for CCS: for instance, the pomset semantics for CCS by
Boudol and Castellani~\cite{BC87,BC88a} and the semantics based on
``concurrent histories'' proposed by Degano, De Nicola and
Montanari~\cite{DM87,DDM88,DDM90}, were both shown to agree with an
interpretation of CCS processes into some class of ESs (PESs
for~\cite{DDM88,DDM90}, PESs with non-hereditary conflict
for~\cite{BC87}, and FESs for~\cite{BC88a}).  Among the early
interpretations of process calculi into ESs, we should also mention
the PES semantics for TCSP (Theoretical CSP~\cite{BHR84,Old86}),
proposed by Goltz and Loogen~\cite{LG91} and later generalised by
Baier and Majster-Cederbaum~\cite{BM94}, and the Bundle ES semantics
for LOTOS, proposed by Langerak~\cite{Lan93} and extended by
Katoen~\cite{Kat96}.  Like FESs, Bundle ESs are a subclass of Stable
ESs. We recall the relationships between the above
classes of ESs (the reader is referred to~\cite{BC94} for separating examples):
\[
Prime~ESs \subset Bundle~ESs \subset Flow~ESs \subset
  Stable~ESs \subset General~ESs 
  \]

More sophisticated ES semantics for CCS, based on FESs and designed to
be robust under action refinement~\cite{AH89,DD93,GGR96}, were 
subsequently  proposed by Goltz and van
Glabbeek~\cite{GG04}. Importantly, all the above-mentioned classes of
ESs, except General ESs, give rise to the same \emph{prime algebraic
  domains} of configurations, from which one can recover a PES by
selecting the complete prime elements.

More recently, ES semantics have been investigated for the
$\pi$-calculus by Crafa, Varacca and Yoshida~\cite{CVY07,VY10,CVY12}
and by Cristescu, Krivine and Varacca~\cite{Cri15,CKV15,CKV16}. 
Previously,  other
causal models for the $\pi$-calculus had already been put forward by
Jategaonkar and Jagadeesan~\cite{JJ95}, by Montanari and
Pistore~\cite{MP95}, by Cattani and Sewell~\cite{CS04} and by Bruni,
Melgratti and Montanari~\cite{BMM06}.  The main new issue, when
addressing causality-based semantics for the $\pi$-calculus, is the
implicit causality induced by scope extrusion. Two alternative views
of such implicit causality had been proposed in  early  work on
noninterleaving operational semantics for the $\pi$-calculus,
respectively by Boreale and Sangiorgi~\cite{BS98} and by Degano and
Priami~\cite{DP99}.  Essentially, in~\cite{BS98} an \emph{extruder}
(that is, an output of a private name) is considered to cause any
action that uses the extruded name, whether in subject or object
position, while in~\cite{DP99} it is considered to cause only the
actions that use the extruded name in subject position. Thus, for
instance, in the process $P = \nu a \,(\overline{b} \langle a\rangle
\pc \overline{c} \langle a\rangle \pc a)$, the two parallel extruders
are considered to be causally dependent in the former approach, and
independent in the latter. All the causal models for the
$\pi$-calculus mentioned above, including the ES-based ones, take one
or the other of these two stands.  Note that opting for the second one
leads necessarily to a non-stable ES model, where there may be causal
ambiguity within the configurations themselves: for instance, in the
above example the maximal configuration contains three events, the
extruders $\overline{b}\langle a\rangle$, $\overline{c} \langle
a\rangle$ and the input on $a$, and one does not know which of the two
extruders enabled the input. Indeed, the paper~\cite{CVY12} uses
non-stable ESs.  The use of non-stable ESs (General ESs) to express
situations where a computational step can merge parts of the state is
advocated for instance by Baldan, Corradini and Gadducci
in~\cite{BCG17}. These ESs give rise to configuration domains that are
not prime algebraic, hence the classical representation theorems have
to be adjusted.

In our simple setting, where we deal only with single sessions and do
not consider session interleaving nor delegation, we can dispense with
channels altogether, and therefore the question of parallel extrusion
does not arise. In this sense, our notion of causality is closer to
that of CCS than to the more complex one of the
$\pi$-calculus. However, even in a more general setting, where
participants would be paired with the channel name of the session they
pertain to, the issue of parallel extrusion would not
arise: indeed, in the above example $b$ and $c$ should be equal,
because participants can only delegate their own channel, but then
they could not be in parallel because of linearity, one of the
distinguishing features enforced by session types. Hence
we believe that in a session-based framework
the two above views of implicit causality should collapse into just
one.

We now briefly discuss our design choices. 
\begin{itemize}
\item  The calculus
considered in the present paper  uses synchronous communication -
rather than asynchronous, buffered communication - because this is how
communication is  classically  modelled in ESs, when they
are used to give semantics to process calculi.   We should
mention however that after first proposing the present study
in~\cite{CDG-LNCS19}, we also considered a calculus with asynchronous
communication in the companion paper~\cite{CDG21}.  In that work too,
networks are interpreted as FESs, and their associated global types,
which we called \emph{asynchronous types} as they split communications
into outputs and inputs, are interpreted as PESs.  The key result is
again an isomorphism between the configuration domain of the FES of a
typed network and that of the PES of its type.  
\item Concerning the choice operator, we adopted here the basic (and
  most restrictive) variant for it, as it was originally proposed for
  multiparty session calculi in \cite{CHY08}.
This is essentially a simplifying assumption, and we do not foresee any difficulty in extending our results
to a more general choice operator,
where the projection is rendered more flexible through the use of  a merge operator~\cite{DY11}. 
%
\item 
 As regards the preorder on processes, which is akin to a subtyping relation, we envisaged to use the
standard subtyping,  in which a process with fewer outputs
 can be used in place of 
 a process with more outputs.  However, in that case Session Fidelity
 would become weaker, since a transition in the LTS of a global type
 would only ensure a transition in the LTS of the corresponding
 network, but not necessarily with the same labelling communication.
 The main drawback would be that \refToTheorem{iso} would no longer
 hold: more precisely, the domains of network configurations would
 only be embedded in (and not isomorphic to) the domains of their
 global type configurations. Notably, typability is independent from
 the use of our preorder or of the standard one, as proved
 in~\cite{BDLT21}.
\end{itemize}

As regards future work, we plan to define an asynchronous
transition system (ATS)~\cite{Bed88} for our calculus, along the lines
of~\cite{BC94}, and show that it provides a noninterleaving
operational semantics for networks that is equivalent to their FES
semantics. This would enable us also to investigate the issue of
reversibility, jointly on our networks and on their FES
representations, since the ATS semantics would give us the handle to
unwind networks, while the corresponding FESs could be unrolled
following one of the methods proposed in existing work on reversible
event structures~\cite{PU2016,CKV16,GPY19,GPY21,Gra21}.

As mentioned at the end of \refToSection{sec:events}, the quest
for a semantic counterpart of our well-formedness conditions on global
types -- namely, for properties that characterise the FESs obtained
from typable networks -- is still open.  By way of comparison, such
semantic well-formedness conditions have been proposed in~\cite{TuostoG18} for
\emph{graphical choreographies}, a truly concurrent
graphical model for global specifications with two kinds of forking
nodes, representing respectively choice and parallel
composition. In~\cite{TuostoG18}, those well-formedness conditions, called
\emph{well-sequencing} and \emph{well-branchedness}, were shown to be
sufficient to ensure projectability on local specifications. In our
case, the property corresponding to well-sequencing is automatically
ensured by our ES semantics, and we conjecture that the well-branchedness
condition for choice nodes (corresponding to projectability) could amount in our
simpler setting\footnote{Our choice operator for global types is less
general than that of~\cite{TuostoG18}.} to the following semantic
condition:

Let $\netev_1, \netev_2\in \GE(N)$ and $\locev{\pp}{\actseq\cdot\pi}\in\netev_1$
and $\locev{\pp}{\actseq\cdot\pi'}\in\netev_2$ with
$\pi\neq\pi'$ and 
$\q = \ptone{\pi}=\ptone{\pi'}$.
If $\netev_1\precN^*\netev'_1$ for some $\netev'_1\in \GE(N)$
such that $\pr\in\loc{\netev'_1}$ with $\pr\not\in \set{\pp, \q}$, then
$\netev_2\precN^*\netev'_2$ for some $\netev'_2\in \GE(N)$ such that
$\pr\in\loc{\netev'_2}$.

This condition would allow us to rule out the FESs of both networks
$\Nt'$ and $\Nt''$ discussed at page~\pageref{wf-discussion}. However,
it should be completed with a condition corresponding to boundedness, and the
conjunction of these two conditions might still not be sufficient in general to ensure typability.
We plan to further investigate this question in the near future.

%

\newpage
\appendix

\section{}
This Appendix contains the proofs of Lemmas~\ref{pf},  \ref{keysr}, \ref{prop:prePostNet},  \ref{epp}, \ref{prop:prePostGl}, and \ref{paltr}.
\begin{lemmaa}{\ref{pf}}{If $\GP$ is bounded, then $\proj\GP\pr$ is a partial function for all $\pr$.  }\end{lemmaa}
\begin{proof}
  We redefine the projection   $\downarrow_\pr$    as the
  largest relation between global types and processes such that $\prR\G\PP\pr$
  implies: 
\begin{enumerate}[label=\roman*)]
\item if $\pr\not\in  \participant{\G}$, then $\PP=\inact$;
\item if $\G= \gt\pr\pp i I \la \G $, then $\PP= \oupP\q{i}{I}{\M}{\PP_i}$ and $\prR{\G_i}{\PP_i}\pr$ for all $i\in I$;
\item if $\G= \gt\pp\pr i I \la \G $, then $\PP= \inpP\pp{i}{I}{\M}{\PP_i}$ and $\prR{\G_i}{ \PP_i}\pr$ for all $i\in I$;
\item if $\G= \gt{\pp}\q i I \la \G $ and
  $\pr\not\in\set{\pp,\q}$ and $\pr\in   \participant{\G_i}$, then
  $\prR{\G_i}{\PP}\pr$ for all $i\in I$. 
\end{enumerate}
The equality $\mathcal E$ of processes is the largest symmetric binary relation   $\RR$   
on processes such that 
$\iR\PP\Q{  \RR  }$ implies: 
\begin{enumerate}[label=(\alph*)]
\item\label{cbf} if $\PP=\oupP\pp{i}{I}{\M}{\PP_i}$ , then $\Q=\oupP\pp{i}{I}{\M}{\Q_i}$ and $\iR{\PP_i}{\Q_i}{  \RR  }$ for all $i\in I$;
\item\label{caf}  if $\PP=\inpP\pp{i}{I}{\M}{\PP_i}$ , then $\Q=\inpP\pp{i}{I}{\M}{\Q_i}$ and $\iR{\PP_i}{\Q_i}{  \RR  }$ for all $i\in I$.
\end{enumerate}
It is then enough to show that the relation
$\RR_\pr  =\set{(\PP,\Q)\mid  \exists\, \G
  \, . \ \prR\G\PP\pr\text { and } \prR\G\Q\pr}$ 
satisfies Clauses~\ref{cbf} and~\ref{caf}   (with $\RR$ replaced by
$\RR_\pr$),   since this will imply $ \RR_\pr \subseteq
\mathcal{E}$.  Note first that $(\inact, \inact) \in \RR_\pr$ because
$(\End, \inact) \in \downarrow_\pr$, and that $(\inact, \inact) \in
\mathcal E$ because Clauses~\ref{cbf} and~\ref{caf} are vacuously
satisfied by the pair   $(\inact, \inact)$.
The proof is by induction on $d=\weight(\G,\pr)$. We only consider
Clause~\ref{caf}, the proof   for   Clause~\ref{cbf}
being similar.  So, assume 
  $(\PP,\Q)\in \RR_\pr$ and   
$\PP=\inpP\pp{i}{I}{\M}{\PP_i}$. \\[3pt]
{\it Case $d=1$.}  In this case  $\G= \gt\pp\pr i I \la \G $ and $\PP= \inpP\pp{i}{I}{\M}{\PP_i}$ and $\prR{\G_i}{  \PP_i}\pr$ for all $i\in I$. From $\prR\G\Q\pr$ we get 
$\Q= \inpP\pp{i}{I}{\M}{\Q_i}$ and $\prR{\G_i}{ \Q_i}\pr$ for all $i\in I$. 
  Hence  $\Q$ has the required form and
   $\iR{\PP_i}{\Q_i}\RR_\pr$ for all $i\in I$.\\
   {\it Case $d>1$.}  In this case  $\G= \gt\pp\q j  J {\la'} \G $ and $\pr\not\in \set{\pp,\q}$ and $\prR{\G_j}{\PP}\pr$ for all $j\in J$. From
$\prR\G\Q\pr$ we get $\prR{\G_j}{\Q}\pr$ for all $j\in J$.   Then   $\iR{\PP}{\Q}{{  \RR_\pr }}$.
\end{proof}
\begin{lemmaa}{\ref{keysr}}{Let $\G$ be a well-formed global type. 
\begin{enumerate}
\item
If $\proj\G\pp=\oup\q{i}{I}{\M}{\PP}$ and $\proj\G\q=\inp\pp{j}{J}{\M'}{\Q}$, then $I=J$, $\M_i=\M_i'$, $\G\stackred{\Comm\pp{\la_i}\q}\G_i$, $\proj{\G_i}\pp=\PP_i$ and $\proj{\G_i}\q=\Q_i$ for all $i\in I$.
\item If $\G\stackred{\Comm\pp{\la}\q}\G'$, then
  $\proj\G\pp=\oup\q{i}{I}{\M}{\PP}$,
  $\proj\G\q=\inp\pp{i}{I}{\M}{\Q}$, where $\la_i=\la$
for some $i\in I$, and 
$\proj{\G'}\pr=\proj\G\pr$ for all $\pr\not\in\set{\pp,\q}$.
\end{enumerate}}\end{lemmaa}
\begin{proof}
  (\ref{keysr1}). The proof is by induction on $d=\weight(\G,\pp)$.\\
 If $d=1$, then by definition of projection (see
  Figure~\ref{fig:proj}) $\proj\G\pp=\oup\q{i}{I}{\M}{\PP}$ implies
  $\G=\gt\pp\q i I {\la} {\G}$ with $\proj{\G_i}\pp=\PP_i$. By the
  same definition it follows that $J = I$ and $\la'_j=\la_j$ and $\Q_j
  = \proj{\G_j}\q$ for all $j\in J$.  Moreover
  $\G\stackred{\Comm\pp{\la_i}\q}\G_i$ by Rule
  \rulename{Ecomm} for all $i\in I$. 
  \\
  If $d>1$, then $\G=\gt\pr\ps h H {\la''} {\G'}$ with
  $\set{\pp,\q}\cap\set{\pr,\ps}=\emptyset$.  By definition of
  projection $\proj\G\pp=\proj{\G'_h}\pp$ and
  $\proj\G\q=\proj{\G'_h}\q$ for all $h\in H$. By Proposition~\ref{dd}
  $\weight(\G,\pp)>\weight(\G'_h,\pp)$ for all $h\in H$.  Then by
  induction $I=J$, $\M_i=\M_i'$,
  \mbox{$\G'_h\stackred{\Comm\pp{\la_i}\q}\G^i_h$,} $\proj{\G^i_h}\pp=\PP_i$
  and $\proj{\G^i_h}\q=\Q_i$ for all $i\in I$ and all $h\in H$. Let
  $\G_i=\gt\pr\ps h H {\la''} {\G^i}$.  By Rule \rulename{Icomm}
  $\G\stackred{\Comm\pp{\la_i}\q}\G_i$ for all $i\in I$.
  By definition of projection $\proj{\G_i}\pp=\PP_i$ and $\proj{\G_i}\q=\Q_i$ for all $i\in I$.\\
  (\ref{keysr2}). The proof is by induction on  the  transition rules of \refToFigure{ltgt}.  \\
  The interesting case is: \prooftree
  \G_h\stackred{\Comm\pp{\la}\q}\G_h' \quad h \in H
  \quad\set{\pp,\q}\cap\set{\ps,\pt}=\emptyset \justifies \gt\ps\pt h
  H {\la'} \G \stackred{\Comm\pp{\la}\q}\gt\ps\pt h H {\la'} {\G'}
  \using ~~~\rulename{Icomm}
  \endprooftree\\
  with $\G=\gt\ps\pt h H {\la'} \G$ and $\G'=\gt\ps\pt h H {\la'} {\G'} $. By induction $\proj{\G_h}\pp=\oup\q{i}{I}{\M}{\PP}$, $\proj{\G_h}\q=\inp\pp{i}{I}{\M}{\Q}$, $\M=\M_i$ for some $i\in I$ and 
$\proj{\G'_h}\pr=\proj{\G_h}\pr$ for all $\pr\not\in\set{\pp,\q}$ and all $h\in H$. By definition of projection $\proj\G\pp=\proj{\G_h}\pp$ and $\proj\G\q=\proj{\G_h}\q$ for all $h\in H$. For $\pr\not\in\set{\pp,\q,\ps,\pt}$ we get $\proj{\G'}\pr=\proj{\G_h'}\pr=\proj{\G_h}\pr=\proj{\G}\pr$. Moreover $\proj{\G'}\ps=\oupp\pt{h}{H}{\M'}{\proj{\G'_h}\ps}=\oupp\pt{h}{H}{\M'}{\proj{\G_h}\ps}=\proj{\G}\ps$ and $\proj{\G'}\pt=\inpp\pt{h}{H}{\M'}{\proj{\G'_h}\pt}=\inpp\ps{h}{H}{\M'}{\proj{\G_h}\ps}=\proj{\G}\pt$.
\end{proof}
\begin{lemmaa}{\ref{prop:prePostNet}}{\begin{enumerate}
\item  If $\post{\netev}{\alpha}$ is defined, then
  $\preP{\post{\netev}{\alpha}}{\alpha}=\netev$;
\item 
$\postP{\pre{\netev}{\alpha}}{\alpha}=\netev$;
 \item   If  $\netev\precN \netev'$, 
then $\pre{\netev}{\alpha}\precN \pre{\netev'}{\alpha}$;
\item  If  $\netev\precN \netev'$ and both $\post{\netev}{\alpha}$ and
  $\post{\netev'}{\alpha}$ are defined, then
  $\post{\netev}{\alpha}\precN \post{\netev'}{\alpha}$; 
\item If  $\netev\grr \netev'$, then $\pre{\netev}{\alpha}\grr
  \pre{\netev'}{\alpha}$;
\item  If  $\netev\grr
  \netev'$ and both $\post{\netev}{\alpha}$ and $\post{\netev'}{\alpha}$
  are defined, then $\post{\netev}{\alpha}\grr\post{\netev'}{\alpha}$;
  \item  If  $\pre{\netev}{\alpha}\grr
  \pre{\netev'}{\alpha}$, then $\netev\grr \netev'$.
\end{enumerate}}\end{lemmaa}

\begin{proof}
 For (\ref{ppn1}) and (\ref{ppn1b})  it  is enough to show
  the corresponding properties for
  located events.
  
  (\ref{ppn1}) Since $\postP{\locev{\pp}\event}{\alpha}$ is defined,
  we have $\event=\concat{(\projS{\alpha}{\pp})}\event'$ and
  $\postP{\locev{\pp}\event}{\alpha}= \locev{\pp}{\event'}$ for some
  $\event'$. Then $\preP{ \postP{\locev{\pp}\event}{\alpha}}{\alpha} =
  \preP{\locev{\pp}{\event'}}{\alpha}=
  \locev{\pp}{\concat{(\projS{\alpha}{\pp})}{\event'}}=
  \locev{\pp}{\event}$.

  (\ref{ppn1b}) Since
  $\preP{\locev{\pp}{\event}}\alpha=\locev{\pp}{\concat{(\projS{\alpha}{\pp})}\event}
  \,$ is always defined, we immediately get $
  \postP{\preP{\locev{\pp}{\event}}\alpha}{\alpha} =
  \postP{\locev{\pp}{\concat{(\projS{\alpha}{\pp})}\event}}{\alpha} =
  \locev{\pp}{\event}$.
 
  (\ref{ppn2b}) Let $\netev\precN \netev'$. By
  \refToDef{netevent-relations}(\ref{c1}), there are
  $\locev{\pp}{\event}\in\netev$ and $\locev{\pp}{\event'}\in\netev'$
  such that $\event<\event'$.  Then
  $\preP{\locev{\pp}{\event}}\alpha=\locev{\pp}{\concat{(\projS{\alpha}{\pp})}\event}
  \, \in \pre{\netev}{\alpha}$ and
  $\preP{\locev{\pp}{\event'}}\alpha=\locev{\pp}{\concat{(\projS{\alpha}{\pp})}\event'}
  \, \in \pre{\netev'}{\alpha}$. Since $\event<\event'$ implies
  $\concat{(\projS{\alpha}{\pp})}{\event} <
  \concat{(\projS{\alpha}{\pp})}{\event'}$, we conclude that
  $\pre{\netev}{\alpha}\precN \pre{\netev'}{\alpha}$.

  (\ref{ppn2}) 
As in the previous case,
there are $\locev{\pp}{\event}\in\netev$ and
$\locev{\pp}{\event'}\in\netev'$ such that $\event<\event'$.  Since
both $\post{\netev}{\alpha}$ and $\post{\netev'}{\alpha}$ are defined,
there exist $\event_0$ and $\event'_0$ such that
$\event=\concat{(\projS{\alpha}{\pp})}\event_0$ and $\event'
=\concat{(\projS{\alpha}{\pp})}\event'_0$ and
$\postP{\locev{\pp}\event}{\alpha}= \locev{\pp}{\event_0}$ and
$\postP{\locev{\pp}\event'}{\alpha}= \locev{\pp}{\event'_0}$. Since
$\event<\event'$ implies $\event_0<\event'_0$,
we conclude that $\post{\netev}{\alpha}\precN \post{\netev'}{\alpha}$.

(\ref{ppn3b}) Let $\netev\grr\netev'$. If Clause (\ref{c21}) of
\refToDef{netevent-relations} applies, then there are
$\locev{\pp}{\procev}\in \netev$ and $\locev{\pp}{\procev'}\in\netev'$
such that $\procev \grr \procev'$.  From
$\pre{(\locev{\pp}{\procev})}{\alpha}=\locev{\pp}{\concat{(\projS{\alpha}{\pp})}\event}$
and
$\pre{(\locev{\pp}{\procev'})}{\alpha}=\locev{\pp}{\concat{(\projS{\alpha}{\pp})}{\event'}}$
we get
$\concat{(\projS{\alpha}{\pp})}\event\grr\concat{(\projS{\alpha}{\pp})}{\event'}$. If
Clause (\ref{c22}) of \refToDef{netevent-relations} applies, then there
are $\locev{\pp}{\procev}\in \netev$ and
$\locev{\q}{\procev'}\in\netev'$ with $\pp \neq \q$ such that
$\cardin{\proj\procev\q} = \cardin{\proj{\procev'}\pp}$ and
$\neg(\dualev{\proj\procev\q}{\proj{\procev'}\pp})$. Let
$\event_0 =\concat{(\projS{\alpha}{\pp})}\event$ and $\event'_0
=\concat{(\projS{\alpha}{\q})}\event'$.  If
$\participant{\alpha}\neq\set{\pp,\q}$, then
$\proj{(\projS{\alpha}{\pp})}{\q} = \ee =
\proj{(\projS{\alpha}{\q})}{\pp}$ and thus
$\proj{\procev_0}\q=\proj{\procev}\q$ and
$\proj{\procev'_0}\pp=\proj{\procev'}\pp$. If
$\participant{\alpha}=\set{\pp,\q}$, say $\alpha=\Comm\pp\la\q$, then
$\procev_0=\sendL{\q}{\la}\cdot\procev$ and
$\procev'_0=\rcvL{\pp}{\la}\cdot\procev'$, which implies
$\cardin{\proj{\procev_0}\q} = \cardin{\proj{\procev}\q} +1 =
\cardin{\proj{\procev'}\pp} +1 = \cardin{\proj{\procev'_0}\pp}$ and
$\neg(\dualev{\proj{\procev_0}\q}{\proj{\procev'_0}\pp})$.  In both
cases we conclude that $\pre\netev\alpha\grr\pre{\netev'}\alpha$.

(\ref{ppn3}) The proof is similar to that of Point (\ref{ppn3b}),
considering that $\post{\netev}{\alpha}$ and $\post{\netev'}{\alpha}$
are defined.

(\ref{ppn7}) Let $\pre\netev\alpha\grr\pre{\netev'}\alpha$. If Clause
(\ref{c21}) of \refToDef{netevent-relations} applies, then there are
$\locev{\pp}{\procev}\in \netev$ and $\locev{\pp}{\procev'}\in
\netev'$ such that
$\concat{(\projS{\alpha}{\pp})}\event\grr\concat{(\projS{\alpha}{\pp})}{\event'}$.
Therefore $\procev \grr \procev'$ and thus $\netev \grr \netev'$.  If
Clause (\ref{c22}) of \refToDef{netevent-relations} applies, then there
are
$\locev{\pp}{\procev_0}=\pre{(\locev{\pp}{\procev})}{\alpha}\in\pre\netev\alpha
$ and
$\locev{\q}{\procev'_0}=\pre{(\locev{\q}{\procev'})}{\alpha}\in\pre{\netev'}\alpha$
with $\pp \neq \q$ such that $\cardin{\proj{\procev_0}\q} =
\cardin{\proj{\procev'_0}\pp}$ and
$\neg(\dualev{\proj{\procev_0}\q}{\proj{\procev'_0}\pp})$.  It follows
that $\event_0 =\concat{(\projS{\alpha}{\pp})}\event$ and $\event'_0
=\concat{(\projS{\alpha}{\q})}\event'$ and $\locev{\pp}{\procev}\in
\netev$ and $\locev{\q}{\procev'}\in \netev'$.  If
$\participant{\alpha}\neq\set{\pp,\q}$, then
$\proj{(\projS{\alpha}{\pp})}{\q} = \ee =
\proj{(\projS{\alpha}{\q})}{\pp}$ and thus
$\proj{\procev}\q=\proj{\procev_0}\q$ and
$\proj{\procev'}\pp=\proj{\procev'_0}\pp$.  If
$\participant{\alpha}=\set{\pp,\q}$, say $\alpha=\Comm\pp\la\q$, then
$\procev_0=\sendL{\q}{\la}\cdot\procev$ and
$\procev'_0=\rcvL{\pp}{\la}\cdot\procev'$, and thus
$\cardin{\proj{\procev}\q} = \cardin{\proj{\procev_0}\q} - 1 =
\cardin{\proj{\procev'_0}\pp} -1 = \cardin{\proj{\procev'}\pp}$ and
$\neg(\dualev{\proj{\procev}\q}{\proj{\procev'}\pp})$. In both cases
we conclude that $\netev \grr \netev'$.
\end{proof}

\begin{lemmaa}{\ref{epp}}{
Let $\Nt\stackred{\alpha}\Nt'$. Then 
\begin{enumerate}
\item $\set{\nec\alpha}\cup\set{\pre\netev\alpha\mid\netev\in\GE(\Nt')}\subseteq\GE(\Nt)$; 
\item $\set{\post\netev\alpha\mid\netev\in\GE(\Nt)\text{ and }\post\netev\alpha\text{ defined} }\subseteq\GE(\Nt')$. 
\end{enumerate}
}\end{lemmaa}
\begin{proof}
  Let $\alpha=\Comm{\pp}{\la}{\q}$. From $\Nt\stackred{\alpha}\Nt'$ we get
\[
  \Nt=\pP{\pp}{\textstyle{\oupp\q{i}{I}{\la}{\PP}}}\parN
    \pP{\q}{\inp\pp{j}{J}{\la}{\Q}}\parN\Nt_0
    \]
    where
  for some $k\in (I \cap J)$ we have $\la_k = \la$ and 
\[
  \Nt'=\pP{\pp}{\PP_k}\parN \pP{\q}{\Q_k}\parN\Nt_0
  \]
  
  (\ref{epp1})   Let
  $\RT=\set{\nec\alpha}\cup\set{\pre\netev\alpha\mid\netev\in\GE(\Nt')}$.
  We first show that $\RT\subseteq\DE(\Nt)$.  By
  \refToDef{netev-relations}(\ref{netev-relations1})
  $\nec{\alpha}\in\DE(\Nt)$.  Let $\netev=\set{\locev{\pr}{\event},
    \locev{\ps}{\event'}} \in \GE(\Nt')$. We want to prove
  that $\pre\netev\alpha\in\DE(\Nt)$.  By
  \refToDef{netev-relations}(\ref{netev-relations1}) there are $R,S$
  such that $\pP\pr R\in\Nt'$ and $\pP\ps S\in\Nt'$ and $\procev\in
  \ES(R)$ and $\procev'\in \ES(S)$.
There are two possible cases:
\begin{itemize}
\item
$\set{\pr,\ps} \cap \set{\pp,\q} = \emptyset$. Then
$\pP\pr R\in\Nt$ and $\pP\ps S\in\Nt$ and thus $\pre\netev\alpha = \netev \in\DE(\Nt)$;
\item $\set{\pr,\ps} \cap \set{\pp,\q} \neq \emptyset$. 
Suppose $\pr = \pp$. 
Then $\procev\in\ES(\PP_k)$ and 
$\locev\pp {\sendL\q{\la_k}\cdot\procev}\in{\pre\netev\alpha}$ and
$\sendL\q{\la_k}\cdot\procev\in\ES(\oup\q{i}{I}{\la}{\PP})$. 
There are two subcases:
\begin{itemize}
\item \mylabel{case1} $\ps = \q$. Then $\procev'\in\ES(\Q_k)$ and
  $\locev\q {\rcvL\pp{\la_k}\cdot\procev'}\in{\pre\netev\alpha}$ and
  $\sendL\q{\la_k}\cdot\procev'\in\ES(\inp\pp{j}{J}{\la}{\Q})$.  In
  this case we have $\pre\netev\alpha = \set{\locev\pp
      {\sendL\q{\la_k}\cdot\procev}, \locev\q
      {\rcvL\pp{\la_k}\cdot\procev'}} \in \DE(\Nt)$;
\item \mylabel{case2} $\ps \neq \q$. Then $\pre{\locev{\ps}{\event'}}{\alpha} =
\locev{\ps}{\event'}$, and thus $\pre{\netev}{\alpha} =
\set{\locev{\pp}{\sendL\q{\la_k}\cdot\procev}, \locev{\ps}{\event'}} \in \DE(\Nt)$.
\end{itemize}
\end{itemize}
Therefore, $\pre\netev\alpha\in\DE(\Nt)$. 
Hence
$\RT\subseteq\DE(\Nt)$.  We want now to show that
$\RT\subseteq\GE(\Nt)$.

Recall from \refToSection{sec:netS-ES} that 
$\GE(\Nt)$ is the greatest fixed point of the function
\[
f_{\DE(\Nt)}(X) = \set{\netev_0\in \DE(\Nt) \mid
    \exists E_0 \subseteq X. \, E_0 ~\text{is a causal set of }
    \netev_0 \text{ in } X}
    \]
Then $\GE(\Nt)$ is also the greatest post-fixed point of
$f_{\DE(\Nt)}(X)$, namely the greatest $X$ such that $X\subseteq
f_{\DE(\Nt)}(X)$.
Therefore, to show that $\RT \subseteq \GE(\Nt)$, it is enough to
show that $\RT$ is also a post-fixed point of
$f_{\DE(\Nt)}(X)$, namely that $\RT \subseteq f_{\DE(\Nt)}(\RT)$.

Consider first the event $\nec{\alpha}$. Since the only 
causal set of $\nec{\alpha}$ in any set is $\emptyset$,
it is immediate that $\nec{\alpha}\in f_{\RT}(\RT)$. 
Consider now $\pre\netev\alpha \in \RT$ for some $\netev\in \GE(\Nt')$  with $\loc\netev=\set{\pr,\ps}$.  Define
\[
\causpre{\alpha}{E}{\netev}= \begin{cases}
    \Xi    & \text{if }\set{\pr,\ps} \cap \set{\pp,\q} = \emptyset\\
    \set{\nec\alpha}\cup\Xi & \text{otherwise}
\end{cases}
\]
where $\Xi=
\set{\pre{\netev'}\alpha\mid\netev'
  \in E\text{ and $E$ is a causal set of $\netev$ in $\GE(\Nt')$}}$. 

  We show that $\causpre{\alpha}{E}{\netev}$ is a causal set of
$\pre\netev\alpha$ in $\RT$, namely that it is a minimal 
subset of $\RT$ satisfying Conditions (\ref{cs1}) and
(\ref{cs2}) of \refToDef{cs}.\\ 
{\em Condition }(\ref{cs1}) If $\nec{\alpha} \in \causpre{\alpha}{E}{\netev}$, then
$\set{\pr,\ps} \cap \set{\pp,\q} \neq \emptyset$. A conflict between $\nec{\alpha}$ and any other event of
$\causpre{\alpha}{E}{\netev} \cup \set{\pre\netev\alpha}$ can only be derived
by Clause (\ref{c21}) of \refToDef{netevent-relations}, since $\nec{\alpha}
=\set{\locev{\pp}{\sendL{\q}{\la}},
  \locev{\q}{\rcvL{\pp}{\la}}}$ and
$\proj{(\projS{\alpha}{\pp})}{\pt} = \proj{(\projS{\alpha}{\q})}{\pt}
= \ee$ for all $\pt\not\in\set{\pp,\q}$.  Suppose $\pr
= \pp$. Then $\locev{\pp}{\concat{\sendL{\q}\la}{\procev}}\in
\pre{\netev}{\alpha}$.  Since $\sendL{\q}\la <
\concat{\sendL{\q}\la}{\procev}$,  Clause (\ref{c21}) cannot be
used to derive a conflict $\nec{\alpha} \grr\pre{\netev}{\alpha}$.
Similarly, if $ \pre{\netev_1}\alpha\in \causpre{\alpha}{E}{\netev}$ and
$\locev{\pp}{\procev_1}\in \netev_1$, then
$\locev{\pp}{\concat{\sendL{\q}\la}{\procev_1}}\in \netev_1$.
Then $\sendL{\q}\la <
\concat{\sendL{\q}\la}{\procev_1}$, hence  Clause (\ref{c21})  
cannot be used to derive $\nec{\alpha}\grr \pre{\netev_1}\alpha$.\\
Suppose now $\pre{\netev_1}\alpha\in \causpre{\alpha}{E}{\netev}$ and  
$\pre{\netev_2}\alpha\in \causpre{\alpha}{E}{\netev}$. Since $E$ is a causal
set, we have $\neg(\netev_1\grr\netev_2)$. Thus 
$\neg(\pre{\netev_1}\alpha\grr\pre{\netev_2}{\alpha})$ by
\refToLemma{prop:prePostNet}(\ref{ppn7}). \\
{\em Condition }(\ref{cs2}) Let $\netev=\set{\locev{\pr}{\event},
  \locev{\ps}{\event'}}$, we have 
$\pre\netev\alpha =\set{\locev{\pr}{\concat{(\projS{\alpha}{\pr})}{\procev}},
  \locev{\ps}{\concat{(\projS{\alpha}{\ps})}{\procev'}}}$.  We 
show that if $\procev_0 < \concat{(\projS{\alpha}{\pr})}{\procev}$,
then $\locev{\pr}{\procev_0} \in
\netev_0$ for some $\netev_0 \in \causpre{\alpha}{E}{\netev}$.
From $\procev_0 < \concat{(\projS{\alpha}{\pr})}{\procev}$ we derive
$\procev_0 = \concat{(\projS{\alpha}{\pr})}{\actseq}$ for some
$\actseq$ such that $\actseq < \procev$. If $\actseq \neq \ee$, then
$\actseq = \procev'_0 < \procev$. Since $E$ is a causal set,
$\procev'_0 < \procev_0$ implies $\occ{\locev{\pr}{\procev'_0}}{E}$.
Hence $\occ{\locev{\pr}{\procev_0}}{\causpre{\alpha}{E}{\netev}}$. If
instead $\actseq = \ee$, then it must be $\procev_0 =
\projS{\alpha}{\pr}\neq \ee$ and thus $\pr \in \set{\pp,\q}$. In this
case $\set{\nec{\alpha}} \in \causpre{\alpha}{E}{\netev}$ and thus
$\occ{\locev{\pr}{\procev_0}}{\causpre{\alpha}{E}{\netev}}$.\\
As for {\em minimality }, we first show that $\netev'\prec
\pre\netev\alpha $ for all $\netev'\in\causpre{\alpha}{E}{\netev}$.
If $\nec{\alpha} \in \causpre{\alpha}{E}{\netev}$, then $\set{\pr,\ps}
\cap \set{\pp,\q} \neq \emptyset$. Then $\nec{\alpha} \prec
\pre\netev\alpha$.  If $ \netev_1\in \causpre{\alpha}{E}{\netev}$ and
$ \netev_1\neq \nec{\alpha}$, then there exists $\netev'_1 \in E$ such
that $\netev_1 =\pre{\netev'_1}\alpha$. Since $E$ is a causal set for
$\netev$, we have $\netev'_1 \prec \netev$. Therefore $\netev_1 =
\pre{\netev'_1}\alpha\prec\pre\netev\alpha$ by
\refToLemma{prop:prePostNet}(\ref{ppn2b}). Assume now that
$\causpre{\alpha}{E}{\netev}$ is not minimal.  Then there is
$E'\subset\causpre{\alpha}{E}{\netev}$ that verifies Condition
(\ref{cs2}) of \refToDef{cs} for $\pre\netev\alpha$. Let $\netev'\in
\causpre{\alpha}{E}{\netev}\setminus E'$. Then $\netev'\prec
\pre\netev\alpha=\set{\locev{\pr}{\procev_{\pr}},
  \locev{\ps}{\procev_{\ps}}}$. Assume that
$\locev{\pr}{\procev'_{\pr}}\in\netev'$ with
$\procev'_{\pr}<\procev_{\pr}$ (the proof is similar for $\ps$).  By
Condition (\ref{cs2}), there is $\netev''\in E'$ such that
$\locev{\pr}{\procev'_{\pr}}\in\netev''$.  But then
$\netev'\grr\netev''$ by \refToProp{prop:conf}, contradicting the fact
that $\causpre{\alpha}{E}{\netev}$ verifies Condition
(\ref{c1}). Therefore $\causpre{\alpha}{E}{\netev}$ is minimal.

(\ref{epp2}) Let $\RS=\set{\post\netev\alpha\mid\netev\in
  \GE(\Nt)\text{ and }\post\netev\alpha\text{ defined} }$.  We first
show that $\RS\subseteq\DE(\Nt')$.  Let
$\netev=\set{\locev{\pr}{\event}, \locev{\ps}{\event'}} \in \GE(\Nt)$
be such that $\post\netev\alpha$ is defined. We want to prove that
$\post\netev\alpha\in\DE(\Nt')$.  By
\refToDef{netev-relations}(\ref{netev-relations1}) there are $R,S$
such that $\pP\pr R\in\Nt$ and $\pP\ps S\in\Nt$ and $\procev\in
\ES(R)$ and $\procev'\in \ES(S)$.
There are two possible cases:
\begin{itemize}
\item
$\set{\pr,\ps} \cap \set{\pp,\q} = \emptyset$. Then
$\pP\pr R\in\Nt'$ and $\pP\ps S\in\Nt'$ and thus $\post\netev\alpha = \netev \in\DE(\Nt')$;
\item $\set{\pr,\ps} \cap \set{\pp,\q} \neq \emptyset$. 
Suppose $\pr = \pp$. 
Then $\procev\in\ES(\oup\q{i}{I}{\la}{\PP})$ and since  $\post\netev\alpha$ is defined we have that
$\procev=\sendL\q{\la_k}\cdot\procev_{k}$ where $\procev_k\in\ES(\PP_k)$.
There are two subcases:
\begin{itemize}
\item \mylabel{Case1} $\ps = \q$. Then $\procev'\in\ES(\inp\pp{j}{J}{\la}{\Q})$ 
and since  $\post\netev\alpha$ is defined
  $\procev'={\rcvL\pp{\la_k}\cdot\procev'_{k}}$ where $\procev'_k\in\ES(\Q_k)$. In
  this case we have $\post\netev\alpha = \set{\locev\pp
      {\procev_{k}}, \locev\q
      {\procev'_{k}}} \in \DE(\Nt')$;
\item \mylabel{Case2} $\ps \neq \q$. Then $\post{\locev{\ps}{\event'}}{\alpha} =
\locev{\ps}{\event'}$, and thus $\post{\netev}{\alpha} =
\set{\locev{\pp}{\procev_k}, \locev{\ps}{\event'}} \in \DE(\Nt')$.
\end{itemize}
\end{itemize}
Therefore $\RS\subseteq\DE(\Nt')$. We want now to show that
$\RS\subseteq\GE(\Nt')$.


We proceed as in the proof of Statement (\ref{epp1}). We know that
$\GE(\Nt')$ is the greatest post-fixed point of the function
\[ 
f_{\DE({\Nt'})}(X) = \set{\netev_0\in \DE(\Nt') \mid
    \exists E_0 \subseteq X. \, E_0 ~\text{is a causal set of }
    \netev_0 \text{ in } X}
    \]
%
Then, in order to obtain $\RS \subseteq \GE(\Nt')$ it is enough to
show that $\RS$ is a post-fixed point of
$f_{\DE({\Nt'})}(X)$, namely that $\RS \subseteq f_{\DE({\Nt'})}(\RS)$.  


Let $\post\netev\alpha \in \RS$ for some $\netev\in \GE(\Nt)$. Define
\[
\causpost{\alpha}{E}{\netev}=\set{\post{\netev'}\alpha\mid\netev'
    \in E\text{ and $E$ is a causal set of $\netev$ in $\GE(\Nt)$}}
    \]
We show that $\causpost{\alpha}{E}{\netev}$ is a causal set of
$\post\netev\alpha$ in $\RS$, namely that it is a minimal subset
of $\RS$ satisfying Conditions (\ref{cs1}) and (\ref{cs2}) of
\refToDef{cs}.\\ 
{\em Condition }(\ref{cs1}) Suppose $\post{\netev_1}\alpha\in
\causpost{\alpha}{E}{\netev}$ and $\post{\netev_2}\alpha\in
\causpost{\alpha}{E}{\netev}$. Since $E$ is a causal set  and
$\netev_1, \netev_2 \in E$, we have
$\neg(\netev_1\grr\netev_2)$. Thus
$\neg(\post{\netev_1}\alpha\grr\post{\netev_2}{\alpha})$ by
\refToLemma{prop:prePostNet}(\ref{ppn3b}) and (\ref{ppn1}). \\
{\em Condition }(\ref{cs2}) Since $\netev=\set{\locev{\pr}{\event},
  \locev{\ps}{\event'}}$ and $\post\netev\alpha$ is defined,
 we have $\procev={\concat{(\projS{\alpha}{\pr})}{\procev_{\pr}}}$ and
$\procev'=\concat{(\projS{\alpha}{\ps})}{\procev_{\ps}}$ and  $\post\netev\alpha
=\set{\locev{\pr}{\procev_{\pr}}, \locev{\ps}{\procev_{\ps}}}$.
Let $\procev_0 < \procev_{\pr}$. Then
$\concat{(\projS{\alpha}{\pr})}{\procev_0} <
\concat{(\projS{\alpha}{\pr})}{\procev_\pr} = \procev$. 
Since $E$ is a causal set for $\netev$ in $\GE(\Nt)$, this implies 
$\occ{\locev{\pr}{\concat{(\projS{\alpha}{\pr})}{\procev_0}}}{E}$. 
Hence $\occ{\locev{\pr}{\procev_0}}{\causpost{\alpha}{E}{\netev}}$. 
\\
As for {\em minimality}, we first show that $\netev'\prec
\post\netev\alpha $ for all $\netev'\in\causpost{\alpha}{E}{\netev}$.
If $ \netev_1\in \causpost{\alpha}{E}{\netev}$, then there exists
$\netev'_1 \in E$ such that $\netev_1 =\post{\netev'_1}\alpha$. Since
$E$ is a causal set for $\netev$, we have $\netev'_1 \prec
\netev$. Therefore $\netev_1 =
\pre{\netev'_1}\alpha\prec\pre\netev\alpha$ by
\refToLemma{prop:prePostNet}(\ref{ppn2b}). Assume now that
$\causpost{\alpha}{E}{\netev}$ is not minimal.  Then there is
$E'\subset\causpost{\alpha}{E}{\netev}$ that verifies Condition
(\ref{cs2}) of \refToDef{cs} for $\post\netev\alpha$.
Let $\netev'\in \causpost{\alpha}{E}{\netev}\setminus E'$. Then
$\netev'\prec \post\netev\alpha=\set{\locev{\pr}{\procev_{\pr}},
  \locev{\ps}{\procev_{\ps}}}$. Assume that
$\locev{\pr}{\procev'_{\pr}}\in\netev'$ with
$\procev'_{\pr}<\procev_{\pr}$ (the proof is similar for $\ps$).  By
Condition (\ref{cs2}), there is $\netev''\in E'$ such that
$\locev{\pr}{\procev'_{\pr}}\in\netev''$.  But then
$\netev'\grr\netev''$ by \refToProp{prop:conf}, contradicting the fact
that $\causpost{\alpha}{E}{\netev}$ verifies Condition
(\ref{c1}). Therefore $\causpost{\alpha}{E}{\netev}$ is minimal.
\end{proof}
\begin{lemmaa}{\ref{prop:prePostGl}}{\begin{enumerate}
\item  If $\postG{\comocc}{\alpha}$ is defined, then $\preG{(\postG{\comocc}{\alpha})}{\alpha}=\comocc$;
\item 
$\postG{(\preG{\comocc}{\alpha})}{\alpha}=\comocc$;
 \item   If  $\comocc_1< \comocc_2$, 
then $\preG{\comocc_1}{\alpha}< \preG{\comocc_2}{\alpha}$;
\item  If  $\comocc_1<\comocc_2$ and  both $\postG{\comocc_1}{\alpha}$ and
  $\postG{\comocc_2}{\alpha}$ 
  are 
  defined, then
  $\postG{\comocc_1}{\alpha}< \postG{\comocc_2}{\alpha}$;
 \item   If $\comocc_1\gr \comocc_2$, 
then $\preG{\comocc_1}{\alpha}\gr \preG{\comocc_2}{\alpha}$;
  \item  If $\comocc<\preG{\comocc'}{\alpha}$, then either $\comocc=\eqclass\alpha$ or $\postG{\comocc}{\alpha}<{\comocc'}$;
 \item  If $\participant{\alpha_1}\cap\participant{\alpha_2}=\emptyset$, then $\preG{(\preG{\comocc}{\alpha_2})}{\alpha_1}=\preG{(\preG{\comocc}{\alpha_1})}{\alpha_2}$;
 \item  If $\participant{\alpha_1}\cap\participant{\alpha_2}=\emptyset$ and both $\postG{(\preG{\comocc}{\alpha_1})}{\alpha_2}$, $\postG{\comocc}{\alpha_2}$ are defined, then $\preG{(\postG{\comocc}{\alpha_2})}{\alpha_1}= \postG{(\preG{\comocc}{\alpha_1})}{\alpha_2}$.
\end{enumerate}}\end{lemmaa}
\begin{proof}
  (\ref{ppg1a}) If $\postG{\eqclass\comseq}{\alpha}$ is defined, then
  in case $\participant\alpha\cap\participant\comseq=\emptyset$ we get
  $\postG{\eqclass\comseq}{\alpha}=\eqclass\comseq$ and also
  $\preG{\eqclass\comseq}\alpha=\eqclass\comseq$, so
  $\preG{(\postG{\eqclass\comseq}{\alpha})}{\alpha}=\eqclass\comseq$.
  Instead if $\participant\alpha\cap\participant\comseq\not=\emptyset$, then
  $\postG{\eqclass\comseq}{\alpha}=\eqclass{\comseq'}$ where
  $\comseq\sim\concat\alpha{\comseq'}$ and
  $\comseq'\neq\emptyseq$.   From $\participant\alpha\cap\participant\comseq\not=\emptyset$ we get 
  $\preG{\eqclass{\comseq'}}\alpha=\eqclass{\concat{\alpha}{\comseq'}}$ by \refToDef{causal-path}. This 
  implies $\preG{(\postG{\eqclass\comseq}{\alpha})}{\alpha}=\eqclass\comseq$. 
  
  (\ref{ppn1b}) By \refToDef{causal-path} either
  $\preG{\eqclass\comseq}\alpha=\eqclass{\concat{\alpha}{\comseq}}$
  if $\participant\alpha\cap\participant\comseq\not=\emptyset$, or
  $\preG{\comseq}\alpha=\eqclass{\comseq}$.  In the first
  case
  $\postG{\eqclass{\concat{\alpha}{\comseq}}}{\alpha}=\eqclass\comseq$
  and in the second
  $\postG{\eqclass{\comseq}}{\alpha}=\eqclass\comseq$, which proves the result. 
  
  (\ref{ppg4a})
   Let $\comocc_1=\eqclass{\comseq}$ and
  $\comocc_2=\eqclass{\concat{\comseq}{\comseq'}}$.  If  
  $\participant{\alpha} \cap \participant{\comseq} \neq \emptyset$,
  then 
  $\participant{\alpha}
  \cap \participant{\concat{\comseq}{\comseq'}} \neq \emptyset$, and
  we have
  $\preG{\comocc_1}{\alpha}=\eqclass{\concat{\alpha}{\comseq}}$ and
  $\preG{\comocc_2}{\alpha}=\eqclass{\concat{\alpha}{\concat{\comseq}{\comseq'}}}$.
Whence $\preG{\comocc_1}{\alpha} \leq \preG{\comocc_2}{\alpha}$.
  Suppose now $\participant{\alpha} \cap \participant{\comseq} =
  \emptyset$. Then $\preG{\comocc_1}{\alpha}=\eqclass{\comseq} =
  \comocc_1$. Now, if also  $\participant{\alpha} \cap \participant{\comseq'} =
  \emptyset$, then $\preG{\comocc_2}{\alpha}=\eqclass{\concat{\comseq}{\comseq}} =
  \comocc_2$ and we are done. If instead $\participant{\alpha} \cap \participant{\comseq'} \neq
  \emptyset$, then $\preG{\comocc_2}{\alpha}=\eqclass{\concat{\alpha}{\concat{\comseq}{\comseq'}}} =
\eqclass{\concat{\comseq}{\concat{\alpha}{\comseq'}}}$, whence $\comocc_1 \leq \preG{\comocc_2}{\alpha}$.

   (\ref{ppg4b}) Let $\comocc_1=\eqclass{\comseq}$ and
  $\comocc_2=\eqclass{\concat{\comseq}{\comseq'}}$. If
  $\participant{\alpha} \cap \participant{\comseq}= \participant{\alpha} \cap \participant{\concat{\comseq}{\comseq'}}= \emptyset$, then 
  $\postG{\comocc_1}{\alpha}=\comocc_1$ and $\postG{\comocc_2}{\alpha}=\comocc_2$. If $\participant{\alpha} \cap \participant{\comseq} \neq
  \emptyset$, then $\comseq\sim\concat\alpha{\comseq_0}$, which implies $\postG{\comocc_1}{\alpha}=\eqclass{\comseq_0}$ and $\postG{\comocc_2}{\alpha}=\eqclass{\concat{\comseq_0}{\comseq'}}$. If
  $\participant{\alpha} \cap \participant{\comseq}=\emptyset$ and $\participant{\alpha} \cap \participant{\concat{\comseq}{\comseq'}} \neq
  \emptyset$, then $\postG{\comocc_1}{\alpha}=\eqclass{\comseq}$ and $\comseq'\sim\concat\alpha{\comseq_0}$, which implies $\postG{\comocc_2}{\alpha}=\eqclass{\concat{\comseq}{\comseq_0}}$.

  (\ref{prop:prePostGl5}) Let $\comocc_1=\eqclass{\comseq}$ and
  $\comocc_2=\eqclass{\comseq'}$ and $\pro{\comseq}\pp\gr\pro{\comseq'}\pp$ for some $\pp$. The only interesting case is  $\participant\alpha\cap\participant{\comseq}=\emptyset$ and  $\participant\alpha\cap\participant{\comseq'}\not=\emptyset$. This implies $\preG{\comocc_1}{\alpha}=\eqclass{\comseq}$ and  $\preG{\comocc_2}{\alpha}=\eqclass{\concat\alpha{\comseq'}}$.
  We get $\pro{(\concat\alpha{\comseq'})}\pp=\pro{\comseq'}\pp$ since $\participant\alpha\cap\participant{\comseq}=\emptyset$ implies $\pp\not\in\participant\alpha$. We conclude  $\preG{\comocc_1}{\alpha}\gr \preG{\comocc_2}{\alpha}$.
  
  (\ref{ppg3})  Let $\comocc=\eqclass{\comseq}$ and $\preG{\comocc'}{\alpha}=\eqclass{\concat{\comseq}{\comseq'}}$.
  If $\postG{\comocc}{\alpha}$ is defined by Point~\ref{ppg4b}
  $\postG{\comocc}{\alpha}<\postG{(\preG{\comocc'}{\alpha})}{\alpha}$
  and by Point~\ref{ppg1b}
  $\postG{(\preG{\comocc'}{\alpha})}{\alpha}=\comocc'$. Otherwise
  either $\comocc=\eqclass\alpha$, in which case we are done, or 
  $\participant{\alpha} \cap \participant{\comseq} \neq
  \emptyset$ and $\comseq\not\sim\concat\alpha{\comseq_0}$. 
  This last case is impossible, since $\participant{\alpha} \cap \participant{\concat{\comseq}{\comseq'}} \neq
  \emptyset$ and $\concat{\comseq}{\comseq'}\not\sim\concat\alpha{\comseq_1}$  contradict the definition of $\circ$ (\refToDef{causal-path}(\ref{causal-path1})).

  (\ref{prop:prePostGl6}) Let $\comocc=\eqclass{\comseq}$. By \refToDef{causal-path}(\ref{causal-path1}) we have four cases: 
  \begin{enumerate}[label=(\alph*)]
  \item $\preG{(\preG{\comseq}{\alpha_2})}{\alpha_1}=\eqclass{\concat{\alpha_1}{(\concat{\alpha_2}\comseq)}}=\eqclass{\concat{\alpha_2}{(\concat{\alpha_1}\comseq)}}=\preG{(\preG{\comseq}{\alpha_1})}{\alpha_2}$ if $\participant{\alpha_1}\cap\participant{\comseq}\not=\emptyset$ and $\participant{\alpha_2}\cap\participant{\comseq}\not=\emptyset$, since $\participant{\alpha_1}\cap\participant{\alpha_2}=\emptyset$;
  \item $\preG{(\preG{\comseq}{\alpha_2})}{\alpha_1}=\eqclass{\concat{\alpha_1}\comseq}=\preG{(\preG{\comseq}{\alpha_1})}{\alpha_2}$ if $\participant{\alpha_1}\cap\participant{\comseq}\not=\emptyset$ and $\participant{\alpha_2}\cap\participant{\comseq}=\emptyset$;
   \item $\preG{(\preG{\comseq}{\alpha_2})}{\alpha_1}=\eqclass{\concat{\alpha_2}\comseq}=\preG{(\preG{\comseq}{\alpha_1})}{\alpha_2}$ if $\participant{\alpha_1}\cap\participant{\comseq}=\emptyset$ and $\participant{\alpha_2}\cap\participant{\comseq}\not=\emptyset$;
    \item $\preG{(\preG{\comseq}{\alpha_2})}{\alpha_1}=\eqclass{\comseq}=\preG{(\preG{\comseq}{\alpha_1})}{\alpha_2}$ if $\participant{\alpha_1}\cap\participant{\comseq}=\emptyset$ and $\participant{\alpha_2}\cap\participant{\comseq}=\emptyset$.
  \end{enumerate}
  
   (\ref{prop:prePostGl7}) Let $\comocc=\eqclass{\comseq}$. By Definitions~\ref{causal-path}(\ref{causal-path1}) and~\ref{def:PostPreGl}(\ref{def:PostPreGl1}) we have four cases:
  \begin{enumerate}[label=(\alph*)]
\item $\preG{(\postG{\comseq}{\alpha_2})}{\alpha_1}=\eqclass{\concat{\alpha_1}{\comseq'}}=\postG{(\preG{\comseq}{\alpha_1})}{\alpha_2}$ if $\participant{\alpha_1}\cap\participant{\comseq}\not=\emptyset$ and $\comseq\sim\concat{\alpha_2}{\comseq'}$, which implies $\concat{\alpha_1}\comseq=\concat{\alpha_1}{(\concat{\alpha_2}{\comseq'})}\sim\concat{\alpha_2}{(\concat{\alpha_1}{\comseq'})}$, since $\participant{\alpha_1}\cap\participant{\alpha_2}=\emptyset$;
\item $\preG{(\postG{\comseq}{\alpha_2})}{\alpha_1}=\eqclass{\concat{\alpha_1}{\comseq}}=\postG{(\preG{\comseq}{\alpha_1})}{\alpha_2}$ if $\participant{\alpha_1}\cap\participant{\comseq}\not=\emptyset$ and $\participant{\alpha_2}\cap\participant{\comseq}=\emptyset$;
\item $\preG{(\postG{\comseq}{\alpha_2})}{\alpha_1}=\eqclass{\comseq'}=\postG{(\preG{\comseq}{\alpha_1})}{\alpha_2}$ if $\participant{\alpha_1}\cap\participant{\comseq}=\emptyset$ and $\comseq\sim\concat{\alpha_2}{\comseq'}$;
\item $\preG{(\postG{\comseq}{\alpha_2})}{\alpha_1}=\eqclass{\comseq}=\postG{(\preG{\comseq}{\alpha_1})}{\alpha_2}$ if $\participant{\alpha_1}\cap\participant{\comseq}=\emptyset$ and $\participant{\alpha_2}\cap\participant{\comseq}=\emptyset$. \qedhere
  \end{enumerate}
  \end{proof}
\begin{lemmaa}{\ref{paltr}}{Let $\G\stackred\alpha \G'$.
\begin{enumerate}
\item If $\comocc\in\EGG(\G')$, then $\preG\comocc{\alpha}\in \EGG(\G)$;
\item If $\comocc\in\EGG(\G)$ and $\postG\comocc{\alpha}$ is defined, then $\postG\comocc{\alpha}\in\EGG( \G')$.
\end{enumerate}
}\end{lemmaa}
\begin{proof}
Both proofs are by induction on the inference of the transition
$\G\stackred\alpha \G'$, see \refToFigure{ltgt}.

(\ref{paltr1}) For rule \rulename{Ecomm} we get $\G=\gt\pp\q i I \la \G$ and $\G'=\G_k$ and $\alpha=\Comm\pp{\la_k}\q$ for some $k\in I$. We conclude $\preG\comocc{\alpha}\in \EGG(\G)$ by \refToLemma{paldf}(\ref{paldf1}).\\
For rule \rulename{Icomm} we get $\G=\gt\pp\q i I \la \G$ and $\G'=\gtp\pp\q i I \la \G$ and $\G_i\stackred\alpha \G'_i$ for all $i\in I$ and $\participant{\alpha}\cap\set{\pp,\q}=\emptyset$. By \refToDef{eg}(\ref{eg1a}) $\comocc\in\EGG( \G')$ implies $\comocc=\ev\comseq$ for some $\comseq\in\FPaths{\G'}$. This implies $\comseq=\concat{\Comm\pp{\la_k}\q}{\comseq'}$ and $\comocc=\eqclass{\comseq_0}$ 
with either $\comseq_0\sim \concat{\Comm\pp{\la_k}\q}{\comseq'_0}$ for some $k\in I$ or $\participant{\comseq_0}\cap\set{\pp,\q}=\emptyset$ by \refToDef{causal-path}. Then $\postG\comocc{\Comm\pp{\la_k}\q}$ is defined unless $\comseq_0=\Comm\pp{\la_k}\q$ by \refToDef{def:PostPreGl}(\ref{def:PostPreGl1}). We consider two cases.\\
If $\comseq_0=\Comm\pp{\la_k}\q$, then $\preG\comocc{\alpha}=\eqclass{\Comm\pp{\la_k}\q}$ since $\participant{\alpha}\cap\set{\pp,\q}=\emptyset$. We conclude $\preG\comocc{\alpha}\in \EGG(\G)$ by \refToDef{eg}(\ref{eg1a}). Otherwise let $\comocc'=\postG\comocc{\Comm\pp{\la_k}\q}$. By \refToLemma{paldf}(\ref{paldf2}) $\comocc'\in \EGG(\G_k')$. By induction $\preG{\comocc'}{\alpha}\in \EGG(\G_k)$.
By \refToLemma{paldf}(\ref{paldf1}) $\preG{(\preG{\comocc'}{\alpha})}{\Comm\pp{\la_k}\q}\in \EGG(\G)$. 
 We now show that $\preG{(\preG{\comocc'}{\alpha})}{\Comm\pp{\la_k}\q}=\preG\comocc{\alpha}$. 
 By  \refToLemma{prop:prePostGl}(\ref{prop:prePostGl6}) and $\participant{\alpha}\cap\set{\pp,\q}=\emptyset$ we get
 $\preG{(\preG{\comocc'}{\alpha})}{\Comm\pp{\la_k}\q}=\preG{(\preG{\comocc'}{{\Comm\pp{\la_k}\q}})}\alpha$ and
  by \refToLemma{prop:prePostGl}(\ref{ppg1a}) we have
 $\preG{\comocc'}{\Comm\pp{\la_k}\q}=\preG{(\postG\comocc{\Comm\pp{\la_k}\q})}{\Comm\pp{\la_k}\q}=\comocc$. Therefore
 $\preG{(\preG{\comocc'}{\alpha})}{\Comm\pp{\la_k}\q}=\preG\comocc{\alpha}\in \EGG(\G)$. 
 
%

(\ref{paltr2}) For rule \rulename{Ecomm} we get $\G=\gt\pp\q i I \la \G$ and $\G'=\G_k$ and $\alpha=\Comm\pp{\la_k}\q$ for some $k\in I$. We conclude $\postG\comocc{\alpha}\in \EGG(\G')$ by \refToLemma{paldf}(\ref{paldf2}).\\
For rule \rulename{Icomm} we get $\G=\gt\pp\q i I \la \G$ and $\G=\gtp\pp\q i I \la \G$ and $\G_i\stackred\alpha \G'_i$ for all $i\in I$ and $\participant{\alpha}\cap\set{\pp,\q}=\emptyset$. By \refToDef{eg}(\ref{eg1a}) $\comocc\in\EGG( \G)$ implies $\comocc=\ev\comseq$ for some $\comseq\in\FPaths{\G}$. This implies $\comseq=\concat{\Comm\pp{\la_k}\q}{\comseq'}$ and $\comocc=\eqclass{\comseq_0}$ 
with either $\comseq_0\sim \concat{\Comm\pp{\la_k}\q}{\comseq'_0}$ for some $k\in I$ or $\participant{\comseq_0}\cap\set{\pp,\q}=\emptyset$ by \refToDef{causal-path}. Then $\postG\comocc{\Comm\pp{\la_k}\q}$ is defined unless $\comseq_0=\Comm\pp{\la_k}\q$ by \refToDef{def:PostPreGl}(\ref{def:PostPreGl1}). We consider two cases.\\
If $\comseq_0=\Comm\pp{\la_k}\q$, then $\postG\comocc{\alpha}=\eqclass{\Comm\pp{\la_k}\q}$ since $\participant{\alpha}\cap\set{\pp,\q}=\emptyset$. We conclude $\postG\comocc{\alpha}\in \EGG(\G')$ by \refToDef{eg}(\ref{eg1a}). Otherwise let $\comocc'=\postG\comocc{\Comm\pp{\la_k}\q}$. By \refToLemma{paldf}(\ref{paldf2}) $\comocc'\in \EGG(\G_k)$. We first show that $\postG{\comocc'}{\alpha}$ is defined. Since $\postG{\comocc}{\alpha}$ and $\postG\comocc{\Comm\pp{\la_k}\q}$ are defined, by \refToDef{def:PostPreGl}(\ref{def:PostPreGl1}) we have four cases:
\begin{enumerate}[label=(\alph*)]
\item\label{ca} $\comseq_0\sim\concat\alpha{\comseq_1}$ for some $\comseq_1$ and $\comseq_0\sim \concat{\Comm\pp{\la_k}\q}{\comseq'_0}$;
\item\label{cb} $\comseq_0\sim\concat\alpha{\comseq_1}$ and  $\participant{\comseq_0}\cap\set{\pp,\q}=\emptyset$;
\item\label{cc}  $\participant{\alpha}\cap\participant{\comseq_0}=\emptyset$ and $\comseq_0\sim \concat{\Comm\pp{\la_k}\q}{\comseq'_0}$;
\item\label{cd}  $\participant{\alpha}\cap\participant{\comseq_0}=\emptyset$ and $\participant{\comseq_0}\cap\set{\pp,\q}=\emptyset$.
 \end{enumerate}
 In case \ref{ca} $\comseq_0\sim\concat\alpha{\concat{\Comm\pp{\la_k}\q}{\comseq'_1}}\sim \concat{\Comm\pp{\la_k}\q}{\concat\alpha{\comseq'_1}}$ for some $\comseq'_1$  since  $\participant{\alpha}\cap\set{\pp,\q}=\emptyset$. Notice that $\comseq'_1\not=\ee$ since $\comseq_0$ is pointed and $\participant{\alpha}\cap\set{\pp,\q}=\emptyset$. We get $\comocc'=\postG\comocc{\Comm\pp{\la_k}\q}=\eqclass{\concat\alpha{\comseq'_1}}$ and $\postG{\comocc'}{\alpha}=\eqclass{\comseq'_1}$.\\
 In case \ref{cb} $\comocc'=\comocc$ and $\postG{\comocc'}{\alpha}=\eqclass{\comseq_1}$.\\
 In case \ref{cc} $\comocc'=\eqclass{\comseq_0'}$ and $\postG{\comocc'}{\alpha}=\eqclass{\comseq_0'}$, since $\participant{\alpha}\cap\participant{\comseq_0}=\emptyset$ implies $\participant{\alpha}\cap\participant{\comseq_0'}=\emptyset$.\\
 In case \ref{cd} $\comocc'=\comocc$ and $\postG{\comocc'}{\alpha}=\comocc$.\\
 By induction $\postG{\comocc'}{\alpha}\in\EGG(\G'_k)$. By \refToLemma{paldf}(\ref{paldf1}) $\preG{(\postG{\comocc'}{\alpha})}{\Comm\pp{\la_k}\q}\in \EGG(\G')$. 
 
 We now show that $\preG{(\postG{\comocc'}{\alpha})}{\Comm\pp{\la_k}\q}=\postG\comocc{\alpha}$. From $\comocc'=\postG\comocc{\Comm\pp{\la_k}\q}$
 and \refToLemma{prop:prePostGl}(\ref{ppg1a})  
 $\preG{\comocc'}{\Comm\pp{\la_k}\q}=\comocc$. Therefore from $\postG{\comocc}{\alpha}$ defined
 we have $\postG{(\preG{\comocc'}{\Comm\pp{\la_k}\q})}{\alpha}$ defined.
 Since $\postG{\comocc'}{\alpha}$ is also defined and $\participant{\alpha}\cap\set{\pp,\q}=\emptyset$, by 
 \refToLemma{prop:prePostGl}(\ref{prop:prePostGl7}) we get 
 $\preG{(\postG{\comocc'}{\alpha})}{\Comm\pp{\la_k}\q}=\postG{(\preG{\comocc'}{\Comm\pp{\la_k}\q})}{\alpha}$.
 Therefore $\preG{(\postG{\comocc'}{\alpha})}{\Comm\pp{\la_k}\q}=\postG{\comocc}{\alpha}\in \EGG(\G')$.
\end{proof}

\end{document}